\documentclass[lettersize,journal]{IEEEtran}
\usepackage{amsmath,amsfonts}
\usepackage{algorithmic}
\usepackage{algorithm}
\usepackage{array}
\usepackage{textcomp}
\usepackage{stfloats}
\usepackage{url}
\usepackage{verbatim}
\usepackage{graphicx}
\usepackage{amssymb}
\usepackage{booktabs} % For better table lines
\usepackage{multirow} % For multi-row and multi-column cells

\usepackage{bbm}
\usepackage{subcaption}
\usepackage{xcolor}
\usepackage{diagbox}
\usepackage{amsthm}
\usepackage[style=ieee, citestyle=numeric-comp]{biblatex}
\usepackage{flushend}

\newtheorem{theorem}{Theorem}
\newtheorem{corollary}{Corollary}
\newtheorem{assumption}{Assumption}
\newtheorem{proposition}{Proposition}
\newtheorem{lemma}{Lemma}
\newtheorem{definition}{Definition}
\newtheorem{remark}{Remark}

\usepackage[hidelinks]{hyperref}

\begin{document}

\title{Structure-Aware Matrix Pencil for Estimating the Number of Damped Complex Exponentials}

\author{Yehonatan-Itay Segman, Alon Amar, and Ronen Talmon,~\IEEEmembership{Senior Member,~IEEE}
\thanks{The authors are with the Viterbi Faculty of Electrical and Computer Engineering, Technion -- Israel Institute of Technology, Haifa 32000, Israel. Corresponding author: Yehonatan-Itay Segman, email: yehonatans@campus.technion.ac.il.}
\thanks{
© 2026 IEEE. Personal use of this material is permitted.
Permission from IEEE must be obtained for all other uses.
Published in IEEE Transactions on Signal Processing, vol. 74,
pp. 2461--2475, 2026.
doi: 10.1109/TSP.2026.3698108}

}
% \markboth{Journal of \LaTeX\ Class Files,~Vol.~14, No.~8, October~2024}
% {Segman, Amar, and Talmon: Structure-Aware Matrix Pencil Method}

% % column for its text to clear the IEEEpubid mark.

\maketitle

\begin{abstract}
We address the problem of detecting the number of damped complex exponentials and estimating their parameters from a noisy signal using the Matrix Pencil (MP) method. 
We introduce the MP \textit{modes} and present their informative spectral structure. We show theoretically that these modes can be divided into signal and noise modes, where the signal modes exhibit a perturbed Vandermonde structure.
Leveraging this structure, we propose a new MP algorithm, termed the SAMP algorithm, which has two novel components: (i) a robust, theoretically grounded model order detection method, and (ii) an efficient estimation of the signal amplitudes.
We show empirically that SAMP significantly outperforms the standard MP method in challenging conditions, with closely-spaced frequencies and low Signal-to-Noise Ratio (SNR) values, approaching the Cramer-Rao lower bound (CRB) for a broad SNR range. 
Additionally, compared with prevalent information-based criteria, we show that SAMP achieves comparable or better results with lower computational time, and lower sensitivity to noise model mismatch.
\end{abstract}

\begin{refsection}
\section{Introduction}

\IEEEPARstart{C}{onsider} the classical model of a discrete-time signal consisting of a sum of damped complex exponentials and additive noise:
\begin{equation}\label{eq: noisy discrete signal}
    y(n) = x(n) + w(n),
\end{equation}
where $0 \leq n \leq N-1$, $w(n)$ is the additive noise, and $x(n)$ is the signal term expressed as a finite sum of damped complex exponentials given by:
\begin{equation}\label{eq:noiseless discrete signal}
    x(n) = \sum_{i=1}^M b_i e^{(-\alpha_i+ j\theta_i)n},
\end{equation}
where $M$ is the number of complex exponentials (or model order), $b_i = |b_i|e^{j\phi_i}$ are the complex amplitudes, $\phi_i\in \mathbb{R}$ are the initial phases, $\alpha_i \in \mathbb{R}$ are the (non-negative) damping factors, $\theta_i \in \mathbb{R}$ are the normalized frequencies, and the lowercase letter `j' is the complex number $j = \sqrt{-1}.$ 
The complex exponentials $e^{-\alpha_i+ j\theta_i}$ are conventionally represented as $z_i = e^{-\alpha_i+ j\theta_i}$, and referred to as the signal poles. The signal poles are assumed to be distinct. 
Given the finite interval of $N$ samples of $y(n)$, the goal is twofold: (i) detect the (discrete-valued) model order $M$ (a detection step), and (ii) determine the (continuous-valued) parameters of interest (an estimation step), i.e., $\{b_i\}_{i=1}^{M}$, $\{\alpha_i\}_{i=1}^{M}$, and $\{\theta_i\}_{i=1}^{M}$. 

This is a longstanding fundamental detection-estimation problem in signal processing, with a broad variety of applications in telecommunications, audio processing, radar systems, and biomedical signal analysis \cite{percival1993spectral, kay1981spectrum, stoica2005spectral}. Over the years, numerous methods have been proposed, including maximum likelihood-based techniques \cite{lang1980frequency, abatzoglou1985fast, bresler1986exact}, super-resolution pseudo-spectrum approaches such as Minimum Variance Distortionless Response (MVDR) (also known as Capon’s method) \cite{capon1969high}, MUltiple SIgnal Classification (MUSIC) \cite{schmidt1986multiple}, Estimation of Signal Parameters via Rotational Invariance Techniques (ESPRIT) \cite{roy1989esprit}, Method Of Direction Estimation (MODE) \cite{stoica1990maximum}, SParse Iterative Covariance-based Estimation (SPICE) \cite{stoica2010new}, sparse model-based methods \cite{hu2012compressed, duarte2013spectral, tang2013compressed, hu2013fast, yang2018sparse, chu2023new, wu2023spectral}, optimization-based techniques \cite{andersson2013frequency, hayes2023sinusoidal}, and more recently, neural network-based approaches \cite{verma2016frequency, izacard2019data, izacard2019learning, xie2021data, pan2021deep, pan2021complex, smith2024frequency}.

% \subsection{Literature Survey}
One common approach for such problems is based on information measures such as the Akaike Information Criterion (AIC), Bayesian Information Criterion (BIC), or Minimum Description Length (MDL) \cite{percival1993spectral, stoica2005spectral}. In this approach, the detection and estimation steps are \textit{coupled}; for each hypothesized model order, the most likely parameters are estimated. Despite its popularity and widespread use, this coupled approach has two main drawbacks. (i) It necessitates the computation of the likelihood function, which is usually unknown in practice. (ii) It is computationally demanding, as it involves repeated model parameter estimations for each hypothesized model order.

More computationally efficient is a \textit{decoupled detection and estimation} approach, where the model order is determined first, followed by parameter estimation. 
A common \textit{non-parametric} method involves computing a spectrum (or pseudo-spectrum) of the observations. For example, the model order can be determined by counting the number of significant peaks in the periodogram (or Bartlett's / Welch's spectrum).
This efficiency comes at the expense of two prominent limitations: limited frequency resolution, which is proportional to the reciprocal of the data length, and tapering / windowing effects \cite{porat1996course}. 
These limitations pose significant concerns for short time series, and many alternatives have been proposed to mitigate them.

Herein, we consider the \textit{decoupled detection and estimation} approach and focus on the MP method, introduced by Hua and Sarkar in \cite{hua1990matrix}, as a super-resolution method. 
The MP method estimates the frequencies of superimposed complex exponentials by solving a generalized eigenvalue problem derived from the signal's Hankel matrix. Consequently, compared to the classical Fourier analysis, it is not restricted to a specific regular grid and does not have a separation limit due to tapering. The MP method has gained considerable attention for its superiority over other linear methods, such as Prony-based methods, demonstrating robustness to noise and computational efficiency \cite{hua1988matrix, sarkar1995using}. 

% The MP method has shown near-optimality for Signal-to-Noise Ratio (SNR) values above 26 dB \cite{hua1988matrix, sarkar1995using}.
Several modifications to the MP algorithm that mitigate its sensitivity to additive white noise have been proposed in the literature, including the Forward-Backward MP (FBMP) approach \cite{hua1988matrix, sarkar1995using} which is limited to undamped exponentials, the Band-Pass MP (BPMP), which is effective in the case of damped exponentials, although its implementation requires prior knowledge of the signal and increased computation time \cite{hu1993utilization}, the Total-Least-Squares MP (TLS-MP) approach, which uses a pre-filtering step for noise mitigation \cite{sarkar1995using}, and the Total Forward-Backward MP (TFBMP), which surpasses the Fast Fourier Transform (FFT) in variance estimates beyond a specific SNR threshold, at the expense of a larger bias \cite{del1996comparison}. 

% Gaps
The above variations of the MP method first determine the model order of the problem through singular values truncation, then estimate the poles from the generalized eigenvalues, and finally estimate the amplitudes by solving a linear least-squares problem \cite{sarkar1995using}.  
Additionally, the main theoretical result, introduced in \cite[Theorem 2.1]{hua1990matrix}, only addresses noiseless scenarios, whereas noise is almost always present in real-world applications. The discrepancy between the ideal noiseless scenario in theory and the presence of noise in practice highlights the need to extend the MP theory to handle noisy signals. 

% Contributions
{In this paper, we present a new formulation and an enhanced version of the MP method, termed the Structure-Aware Matrix Pencil (SAMP) method, which implements a ``first estimate, then detect'' strategy, rather than the standard ``detect, then estimate'' approach. While the estimation stage follows the standard MP pipeline, the main contribution of SAMP is a post-estimation model order detection stage that also leverages the MP estimates.
At the core of our method lie the MP \emph{modes}, a new notion that is borrowed from the realm of dynamic mode decomposition (DMD) \cite{schmid2010dynamic} and introduced in this context for the first time to the best of our knowledge.
Concretely, the left and right MP \textit{modes} are defined as the columns and rows of the pseudo-inverse matrix of the generalized eigenvectors matrix obtained by the MP method, respectively. 
We propose a new model order detection method that leverages the structure of the MP modes.
In addition, the method includes an efficient procedure for amplitude estimation, based on the MP modes.}

% Contributions
% In this paper, we present a new MP method, termed the Structure-Aware Matrix Pencil (SAMP), which implements a new strategy of ``first estimate and then detect'', rather than the standard ``detect and then estimate'' approach. While the core pole estimation pipeline follows the standard MP approach, our main contribution lies in a new strategy for model order detection
% Our new method has two novel components: (i) a new approach for model order detection, and (ii) an efficient amplitude estimation. At the core of our method lie the MP \emph{modes}, a new notion that is borrowed from the realm of dynamic mode decomposition (DMD) \cite{schmid2010dynamic} and introduced in this context for the first time to the best of our knowledge. 
%
% Concretely, the left and right MP \textit{modes} are defined as the columns and rows of the pseudo-inverse matrix of the generalized eigenvectors matrix obtained by the MP method, respectively. 

We show theoretically that these modes can be divided into signal modes and noise modes, where the signal modes exhibit a perturbed Vandermonde structure. We theoretically study this perturbation and explicitly present the dependency on the model parameters. We put special focus on the impact of noise and extend the classical MP theorem in \cite{hua1990matrix} to address noisy scenarios and demonstrate that all the signal components in \eqref{eq:noiseless discrete signal} can be detected through the signal modes.
Leveraging the structure of the signal modes, we depart from singular values thresholding heuristics. Instead, we propose a new model order detection with theoretical guarantees that depend on the SNR and the separation of the poles.
We show empirically that the new SAMP algorithm obtains superior performance compared with the standard MP method that is based on singular values thresholding for the model order detection, particularly in challenging conditions with closely-spaced frequencies and low SNR values in the threshold region, where model order estimation becomes challenging. In addition, we show that compared with the prevalent information-based criteria, SAMP is less sensitive to noise model mismatch and is more computationally efficient.

The paper is organized as follows. In Sec. \ref{sec: MP method} we describe the existing MP method. In Sec. \ref{sec: the MP modes} we introduce the MP modes and analyze their structure. In Sec. \ref{sec: model order} we propose a novel order detection rule based on the MP modes structure. 
In Sec. \ref{sec: eff amp est} we propose an efficient method for estimating the signal amplitudes using the MP modes. 
In Sec. \ref{sec: Proposed Algorithm} we present the SAMP algorithm for estimating the model parameters. In Sec. \ref{sec:Results} we present simulation results. Finally, in Sec. \ref{sec: Conclusion} we summarize this paper with conclusions and discussion regarding possible practical implications and future research.

% \subsection{Notations and Acronyms}
% Throughout this paper, lowercase letters are used for scalars (e.g. z, $\lambda$), uppercase letters, for integers (e.g. N, L), lowercase bold letters for vectors (e.g. $\mathbf{p, q}$), and uppercase bold letters for matrices (e.g. $\mathbf{X, Y, Z}$). Additionally, the superscript $(\cdot)^{\mathrm{H}}$ denotes the complex conjugate transpose, the superscript $(\cdot)^{\mathrm{T}}$ denotes the transpose operation, and the superscript $(\cdot)^{\dagger}$ denotes the pseudo-inverse. 
% We note that the superscript $(\cdot)^{\dagger}$ for a rank-$M$ matrix denotes the rank-$M$ pseudo-inverse of the matrix. 
% Finally, a perturbation of a vector $\mathbf{v}$ will be denoted by $\widetilde{\mathbf{v}}$ or equivalently by $\mathbf{v} + \delta(\mathbf{v})$. We use the same notation for perturbations of scalars, vectors, or matrices. 

\section{The Classical Matrix Pencil Method} 
\label{sec: MP method}

The MP method estimates signal parameters by solving a generalized eigenvalue problem based on the available data samples. It begins by constructing an $(N-L) \times (L+1)$ Hankel matrix $\mathbf{Y}$ from the noisy measurements $y(n)$:

\begin{equation}\label{eq:full noisy Hankel matrix}
\mathbf{Y} =
    \begin{bmatrix}
    y(0)   & y(1)       & \cdots & y(L)   \\
    y(1)   & y(2)       & \cdots & y(L+1) \\
    \vdots & \vdots  & \ddots & \vdots \\
    y(N-L-1) & y(N-L)   & \cdots & y(N-1)
    \end{bmatrix}.
\end{equation} 

The parameter $L$, termed the \emph{pencil parameter}, is chosen to satisfy $M \leq L \leq N-M $ (see Sec. \ref{subsec: MP noiseless case} for more details). Using perturbation analysis, it was demonstrated that the optimal range for minimizing variance in the estimator is $\frac{N}{3} \leq L \leq \frac{N}{2}$ \cite{rao1985perturbation, hua1988perturbation, hua1988matrix, hua1990matrix, djermoune2009perturbation}. For the rest of this paper, we assume that $L$ lies within this optimal range.

Next, two sub-matrices \(\mathbf{Y}_0\) and \(\mathbf{Y}_1\) are formed from \(\mathbf{Y}\) by removing its last and first columns, respectively. A denoising step (see section \ref{subsec:MP method - model order detection}) is then performed using the singular value decomposition (SVD) of the Hankel matrix \(\mathbf{Y}_0\):
\begin{equation} \label{eq:SVD of Y_0}
    \mathbf{Y}_0 = \mathbf{U} \mathbf{\Sigma} \mathbf{V}^{\mathrm{H}},    
\end{equation}
where \(\mathbf{U} \in \mathbb{C}^{(N-L) \times L}\) and \(\mathbf{V} \in \mathbb{C}^{L \times L}\) are unitary matrices, \(\mathbf{\Sigma} = \text{diag}(\sigma_1, \ldots, \sigma_L)\) with $\sigma_1 \geq \sigma_2 \geq \ldots \geq \sigma_L$, and $(\cdot)^{\mathrm{H}}$ denotes the complex conjugate transpose.

Using \eqref{eq:SVD of Y_0} we can express the MP by:
\begin{equation} \label{eq:full MP decomposition}
    \mathbf{Y}_1 -\widetilde{\lambda} \mathbf{Y}_0 = \mathbf{U} \mathbf{\Sigma} ( \mathbf{A} -\widetilde{\lambda} \mathbf{I} ) \mathbf{V} ^{\mathrm{H}},
\end{equation}
where the non-symmetric matrix $\mathbf{A}\in\mathbb{C}^{L \times L}$ is given by:
\begin{equation} \label{eq:rank reduced representation}
\mathbf{A} = \mathbf{\Sigma}^{-1} \mathbf{U}^{\mathrm{H}} \mathbf{Y}_1 \mathbf{V},
\end{equation}
and we use the notation of $\widetilde{\lambda}$ to denote a perturbation of $\lambda$, an eigenvalue of the noiseless MP (see Section \ref{subsec: MP noiseless case} for more details).
The MP method is based on finding the eigenvalues of $\mathbf{Y}_1  - \widetilde{\lambda} \mathbf{Y}_0$ as an estimate for the signal poles $z_i$. Instead of directly computing the eigenvalues of $\mathbf{Y}_1 -\widetilde{\lambda}\mathbf{Y}_0$, the estimates of the signal poles are obtained by the eigenvalues of $\mathbf{A}$ \cite{hua1990matrix}. Accordingly, an eigenvalue decomposition (EVD) of the matrix $\mathbf{A}$ is performed:
\begin{equation}\label{eq:A decompos}
    \mathbf{A} = \mathbf{Q} \widetilde{\mathbf{\Lambda}} \mathbf{Q}^{-1},
\end{equation}
where $\widetilde{\mathbf{\Lambda}} = \text{diag}(\widetilde\lambda_1,\ldots,\widetilde\lambda_L)$ and $\mathbf{Q}\in\mathbb{C}^{L\times L}$ contains the eigenvectors as columns.

A critical problem arises as the matrices $\mathbf{Y}_0$ and $\mathbf{Y}_1$ have a full rank (which equals $L$) stemming from the additive noise contamination. Consequently, the matrix $\mathbf{A}$ has $L$ non-zero eigenvalues, consisting of $M$ signal-related eigenvalues, and $L-M$ extraneous eigenvalues, containing only noise-related information \cite{sarkar1995using}.
Therefore, the model order (the number of complex exponentials) needs to be determined in addition to, and usually before, the estimation of the parameters.
%Assuming that we have an estimate of the model order, $\widehat M$, the second step is to estimate the signal parameters. We will now discuss these two steps.

\vspace{-1mm}
\subsection{Model order detection} 
\label{subsec:MP method - model order detection}

Determining the model order in the MP method involves truncating the singular values from the SVD of $\mathbf{Y}_0$ in \eqref{eq:SVD of Y_0}, and the standard approach for this step is the decoupled detection-estimation approach. A commonly used technique is the Significant Decimal Digit (SDD) \cite{hua1990matrix, hua1988matrix, sarkar1995using, laroche1993use, bhuiyan2012advantages}, where the model order is estimated by:

\begin{equation}
    % \widehat{M}_{\text{SD}} = \left| \{\sigma_i : \sigma_i / \sigma_{\text{max}} \geq 10^{-p} \}\right|,
    \widehat{M}_{\text{SD}} = \underset{i=1,\ldots,L}{\text{arg}\min} \; \{\sigma_i : \sigma_i / \sigma_{\text{max}} \geq 10^{-p} \},
\end{equation} 
where $p$ is defined as the number of significant decimal digits in the data, $\sigma_{\text{max}}$ is the maximal singular value, and the singular values are assumed to be in descending order.
Another closely related technique is the ``gap'' technique \cite{del1996comparison, yin2011model}, where the model order is estimated by:
\begin{equation}
    \widehat{M}_{\text{GAP}}= \underset{i=1,\ldots,L-1}{\text{arg}\max} \; \left\{\sigma_i /\sigma_{i+1} \right\}.
\end{equation} 

Although less mentioned in the MP literature, the information-theoretic based estimators (ITE) \cite{stoica2004model, stoica2005spectral, nadler2011model} can also be used:
\begin{equation}
    \widehat{M}_{\text{ITE}} =   \underset{k=1,\ldots,M_{\text{max}}}{\text{arg}\min} \; \{ -\text{log}\left( p_k (\mathbf{y}|\hat{\mathbf{v}}_k )\right) + k C(k,N)\},
\end{equation}
where $M_{\text{max}}$ is the maximum number of hypothesized exponentials (naturally, $M_{\text{max}}\leq L)$, $\mathbf{y} = \{y(n)\}_{n=0}^{N-1}$ is the vector of observations \eqref{eq: noisy discrete signal}, and $p_k(\mathbf{y}|\hat {\mathbf{v}}_k)$ is the probability density function (PDF) of $\mathbf{y}$ under the estimated parameter vector $\hat {\mathbf{v}}_k$, assuming the model order is $k$. Here, the estimated parameter vector is $\hat{\mathbf{v}}_k = [\{\hat b_i\}_{i=1}^{k}, \{\hat \alpha_i\}_{i=1}^{k}, \{\hat \theta_i\}_{i=1}^{k}]$, obtained by the MP method.
The term $C(k, N)$ is a model-complexity penalty term, used as a regularization for maximum-likelihood based methods.
% \begin{remark}
We emphasize that this is a coupled detection-estimation approach, which requires prior knowledge of the likelihood function.
% \end{remark}

\vspace{-1mm}
\subsection{Parameter estimation}\label{subsec: classic MP params est}

Once we determine the model order $\widehat M$, we truncate the SVD of $\mathbf{Y}_0$ in \eqref{eq:SVD of Y_0} to rank $\widehat M$ and use the eigenvalues of the resulting $\widehat{M} \times \widehat{M}$ matrix $\mathbf{A}$ in \eqref{eq:A decompos} to estimate the signal poles. The estimated damping factor and normalized frequency are given by:
\begin{equation}\label{eq:params calc}
    \hat{\alpha}_i = -\text{log}|{\widetilde\lambda}_i| \,\,, \,\, \hat{\theta}_i = \text{arg}(\widetilde {\lambda}_i)
    \,\,, i=1,2,\ldots,\widehat{M}
\end{equation}
Next, the amplitudes $\{b_i\}_{i=1}^{\widehat M}$ are estimated by solving the linear least-squares problem \cite{sarkar1995using}:
\begin{equation}
    \hat{b}_i = \underset{\mathbf b\in\mathbb{C}^M} {\text{arg}\min}\sum\limits_{n=0}^{N-1} |y(n) - \sum\limits_{i=1}^{\widehat{M}}\widetilde {\lambda}_i^n b_i|^2,
\end{equation} which has a closed-form solution of:
\begin{equation}\label{eq:MP coeff estimation}
\hat{\mathbf{b}} = 
   \begin{bmatrix}
    1 & 1 & \cdots & 1\\
    \widetilde\lambda_1 & \widetilde\lambda_2 & \cdots & \widetilde\lambda_{\widehat M}\\
    \widetilde\lambda_1^2 & \widetilde\lambda_2^2 & \cdots & \widetilde\lambda_{\widehat M}^2\\
    \vdots & \vdots & \ddots & \vdots\\
    \widetilde\lambda_1^{N-1} & \widetilde\lambda_2^{N-1} & \cdots & \widetilde\lambda_{\widehat M}^{N-1}
    \end{bmatrix}^{\dagger}
    \begin{bmatrix} 
     y(0)\\
     y(1)\\
     y(2)\\
     \vdots\\
     y(N-1)
     \end{bmatrix},
\end{equation} where $(\cdot)^\dagger$ denotes the pseudo-inverse and $\hat{\mathbf{b}} = [\hat{b}_1, \hat{b}_2, \ldots,  \hat{b}_{\widehat M} ]^{\mathrm{T}} \in \mathbb{C}^{\widehat{M}}$.

\subsection{The degenerate noiseless case} \label{subsec: MP noiseless case}

In the noiseless case, Hua et al. presented a complete theorem \cite[Theorem 2.1]{hua1990matrix}, showing that the signal can be fully characterized using the noiseless MP solutions. Specifically, {let $\mathbf{X}_0$ and $\mathbf{X}_1$ be the noiseless Hankel matrices, constructed from $x(n)$, similarly to \eqref{eq:full noisy Hankel matrix}.} If $M\leq L\leq N-M$, the solutions $(\lambda, \mathbf{q}, \mathbf{p}^{\mathrm{H}})$ to the generalized noiseless eigenvalue problem:

\begin{equation}
\label{eq:generalize eigenvalue problem}
( \mathbf{X}_1-\lambda \mathbf{X}_0 ) \mathbf{q} = 0, \;\;\; \mathbf{p} ^{\mathrm{H}} ( \mathbf{X}_1-\lambda\mathbf{X}_0 ) = 0,
 \end{equation}
such that $ \mathbf{q} \in{\text{Range}( \mathbf{X}_0 ^{\mathrm{H}})}$ and $ \mathbf{p} \in{\text{Range}( \mathbf{X}_0 )}$ are $(z_i, \mathbf{q}_i, \mathbf{p}^{\mathrm{H}}_i)$, where $ \mathbf{q}_i $ is the $i$-th column of 
\[
\mathbf{Z}_R ^{\dagger} = \mathbf{Z}_R ^{\mathrm{H}}( \mathbf{Z}_R \mathbf{Z}_R ^{\mathrm{H}})^{-1},
\]
$\mathbf{p}_i ^{\mathrm{H}}$ is the $i$-th row of 
\[
\mathbf{Z}_L ^{\dagger} = ( \mathbf{Z}_L ^{\mathrm{H}} \mathbf{Z}_L)^{-1} \mathbf{Z}_L ^{\mathrm{H}},
\]
and:
\begin{enumerate}
    \item{$ \mathbf{Z}_L\in \mathbb{C}^{(N-L)\times M} $ is a Vandermonde matrix given by:
    \begin{equation} \label{eq:Z_L}
        \begin{bmatrix}
            1             & 1               & \cdots  & 1\\
            z_1           & z_2             & \cdots  & z_M\\
            z_1^2         & z_2^2           & \cdots  & z_M^2\\
            \vdots        & \vdots          & \ddots  & \vdots\\
            z_1^{N-L-1}   &  z_2^{N-L-1}    & \cdots  &  z_M^{N-L-1}
        \end{bmatrix}.
    \end{equation}}
    \item{$\mathbf{Z}_R\in \mathbb{C}^{M\times L} $ is a Vandermonde matrix given by:
    \begin{equation}
        \begin{bmatrix}
            1       & z_1     & z_1^2  &\cdots   & z_1^{L-1}\\
            1       & z_2     & z_2^2  &\cdots   & z_2^{L-1}\\
            \vdots  & \vdots  & \vdots & \ddots  & \vdots   \\
            1       & z_M     & z_M^2  & \cdots  & z_M^{L-1}
        \end{bmatrix}.
    \end{equation}}        
\end{enumerate} 
The theorem is based on the following Vandermonde factorization of the noiseless Hankel matrices:
\begin{equation}\label{eq:clean Vandermonde factorization}
\begin{aligned} 
    \mathbf{X}_0 &= \mathbf{Z}_L \mathbf{B} \mathbf{Z}_R,\\
    \mathbf{X}_1 &= \mathbf{Z}_L \mathbf{B} \mathbf{Z} \mathbf{Z}_R,
\end{aligned}    
\end{equation}
where $\mathbf{B}$ and $\mathbf{Z}$ are $M\times M$ diagonal matrices comprising the signal's amplitudes and poles, respectively.
For the proof and more details, see \cite[Theorem 2.1]{hua1990matrix}.

As noted in \cite{hua1990matrix}, left-multiplying \eqref{eq:generalize eigenvalue problem} by $\mathbf{X}_0^{\dagger}$, implies that the non-zero eigenvalues of the matrix $\mathbf{X}_0^{\dagger}\mathbf{X}_1 \in\mathbb{C}^{L\times L}$ are precisely the $M$ signal poles, $\{z_i\}_{i=1}^M$. Furthermore, since the matrix $\mathbf{X}_0^{\dagger}\mathbf{X}_1 = \mathbf{Z}_R^{\dagger}\mathbf{Z}\mathbf{Z}_R$ has a rank of $M\leq L$, it also has $L-M$ zero eigenvalues with corresponding right eigenvectors in the null space of $\mathbf{Z}_R$. 

While the MP theorem only addresses noiseless scenarios, practical situations involve noisy measurements. The impact of noise is substantial and current literature is built upon singular-values thresholding techniques, as discussed in \ref{subsec:MP method - model order detection}. However, in cases of closely-spaced frequencies and low SNR values in the threshold region, these techniques tend to underestimate or overestimate the model order $M$ (see Section \ref{sec:Results}).
In Section \ref{sec: the MP modes}, we introduce the MP \textit{modes}, which are derived by accounting for noise effects, and reveal the spectral structure of the signal-related modes. This structure allows us to identify the signal-related modes and estimate $M$ more accurately (see Section \ref{sec: Proposed Algorithm}), eliminating the need for singular values thresholding employed in the decoupled approach.

\vspace{-2mm}
\section{Introducing The MP Modes} \label{sec: the MP modes}
The traditional model order detection in the MP method relies on the SVD of $\mathbf{Y}_0 =  \mathbf{U} \mathbf{\Sigma} \mathbf{V}^{\mathrm{H}}$, given in \eqref{eq:SVD of Y_0}. However, we now demonstrate that $\mathbf{Y}_0$ can be further decomposed into a product of matrices, and the structure of these matrices can be leveraged to determine the model order. 
Recall that, in the noiseless case, the factorization in \eqref{eq:clean Vandermonde factorization} consists of the matrices $\mathbf{Z}_R$ and $\mathbf{Z}_L$, which are (i) Vandermonde matrices, and (ii) the pseudo-inverses of the right and left eigenvector matrices of the noiseless MP $
\mathbf{X}_1-\lambda \mathbf{X}_0$. To introduce a similar factorization in the noisy case, we first define the noisy counterparts of $\mathbf{Z}_R$ and $\mathbf{Z}_L$ as the MP \emph{modes}:
\begin{definition}[MP modes]
Define the \textit{right MP modes} as the rows of:
\begin{equation} 
    \widetilde{\mathbf{Z}}_{R,un}  := \mathbf{Q}^{-1}\mathbf{V}^{\mathrm{H}} \in \mathbb{C}^{L \times L}, \label{def:right mp mode}
\end{equation}
{where the SVD of $\mathbf{Y}_0 = \mathbf{U\Sigma V}^{\mathrm{H}}$ from \eqref{eq:SVD of Y_0} is used, and $\mathbf{Q}$ consists of the eigenvectors of $\mathbf{A}$ from \eqref{eq:A decompos}.}

\noindent Similarly, define the \textit{left MP modes} as the columns of:
\begin{equation}
     \widetilde{\mathbf{Z}}_{L,un}  := \mathbf{U \Sigma Q} \in \mathbb{C}^{(N-L) \times L}, \label{def:left mp mode} 
\end{equation}
where 'un' stands for 'un-normalized'.
\end{definition}
 
We note that the notion of ``modes'' is new in the context of the MP method, but is borrowed from the related DMD approach \cite[Chap. 1]{kutz2016dynamic}.
In the sequel, we demonstrate that these modes contain spectral information that can be utilized in the detection step.
These matrices are the pseudo-inverses of the right and left eigenvector matrices of the noisy MP $\mathbf{Y}_1 -\lambda\mathbf{Y}_0$, given according to \eqref{eq:full MP decomposition} and \eqref{eq:A decompos} by:
\begin{align}
    \widetilde{\mathbf{Z}}_{R,un}^{\dagger} &= \mathbf{V} \mathbf{Q}, \label{eq:right noisy eigenvectors}\\
    \widetilde{\mathbf{Z}}_{L,un}^{\dagger} &= \mathbf{Q}^{-1}\mathbf{\Sigma}^{-1}\mathbf{U}^{\mathrm{H}}, \label{eq:left noisy eigenvectors}
\end{align}
respectively. By normalizing $\widetilde{\mathbf{Z}}_{R,un}$ and $\widetilde{\mathbf{Z}}_{L,un}$ (See Appendix \ref{app:Extending the MP Theorem: Accounting for Noise} in the Supplementary Materials (SM)) and combining them with the SVD of $\mathbf{Y}_0$, we can express the data matrices $\mathbf{Y}_0$ and $\mathbf{Y}_1$ as follows:
\begin{equation}\label{eq:noisy Vandermonde factorization}
    \begin{aligned}
         \mathbf{Y}_0 &= \widetilde{\mathbf{Z}}_L \widetilde{ \mathbf{B} } \widetilde{\mathbf{Z}}_R,\\ 
         \mathbf{Y}_1 &= \widetilde{\mathbf{Z}}_L\widetilde{ \mathbf{B} } \widetilde{\mathbf{\Lambda}} \widetilde{\mathbf{Z}}_R,
    \end{aligned}
\end{equation}
where $\widetilde{\mathbf{Z}}_R$ and $\widetilde{\mathbf{Z}}_L$ are the normalized modes, $\widetilde{ \mathbf{B} }, \widetilde{\mathbf{\Lambda}} \in \mathbb{C}^{L \times L}$ are diagonal matrices containing the noisy amplitudes and poles, respectively. This decomposition forms a noisy analogous to the decomposition of $\mathbf{X}_0$ and $\mathbf{X}_1$ in \eqref{eq:clean Vandermonde factorization}. For more details, see Appendix \ref{app:Extending the MP Theorem: Accounting for Noise} in the SM, where we provide a comprehensive theorem that establishes the extension of the classical MP theorem to accommodate signals contaminated by noise. 
The next result shows that $\widetilde{\mathbf{Z}}_L$ has a valuable structure (a similar argument applies to $\widetilde{\mathbf{Z}}_R$, but is omitted for brevity).

\begin{proposition}\label{prop:Z_L recast}
    The matrix $\widetilde{\mathbf{Z}}_L$, whose columns are the left MP modes, can be recast as 
    \begin{equation} 
    \label{eq:left MP mode decomposition}   
    \widetilde{\mathbf{Z}}_L = 
    \begin{bmatrix} 
      \mathbf{Z}_L + \mathbf{E}_L, \; \mathbf{C}_L \\ 
    \end{bmatrix}_{(N-L) \times L},
    \end{equation}
    where the $M$ leftmost columns of $\widetilde{\mathbf{Z}}_L$ are the sum of the Vandermonde matrix $\mathbf{Z}_L\in \mathbb{C}^{(N-L)\times M}$ and a noise-related component, denoted by $\mathbf{E}_L \in \mathbb{C}^{(N-L)\times M}$. The remaining $L-M$ columns are noise-related spurious columns, denoted by $\mathbf{C}_L\in \mathbb{C}^{(N-L)\times (L-M)}$.
\end{proposition}
\begin{proof}
    See Appendix \ref{app:Extending the MP Theorem: Accounting for Noise} in the SM.
\end{proof}

\textit{Signal and noise modes}: in the noiseless case, the matrix $\mathbf{Z}_L$ assumes a Vandermonde structure \eqref{eq:Z_L}.
In the noisy case, by \eqref{eq:left MP mode decomposition}, the matrix $\widetilde{\mathbf{Z}}_L$ can be divided into two sub-matrices: the signal-related sub-matrix comprising the \textit{signal modes}, represented by the $M$ leftmost columns $\mathbf{Z}_L+\mathbf{E}_L$, and the noise-related sub-matrix comprising the \textit{noise modes}, represented by the $L-M$ rightmost columns $\mathbf{C}_L$.
In addition, each signal mode is a column of a Vandermonde matrix associated with a signal pole, obscured by the additive noise term $\mathbf{E}_L$. See Fig. \ref{fig:noisy left mode structure} for illustration. 

Critically, Proposition \ref{prop:Z_L recast} is meaningful and has a practical value only when the noise term $\mathbf{E}_L$ has a structure that allows accurate extraction of the Vandermonde structure. Therefore, in the sequel, we analyze the noise matrix $\mathbf{E}_L$. The analysis of $\mathbf{E}_R$ follows similarly and is therefore omitted.

\begin{figure}
    \centering
    \includegraphics[width=1\linewidth]{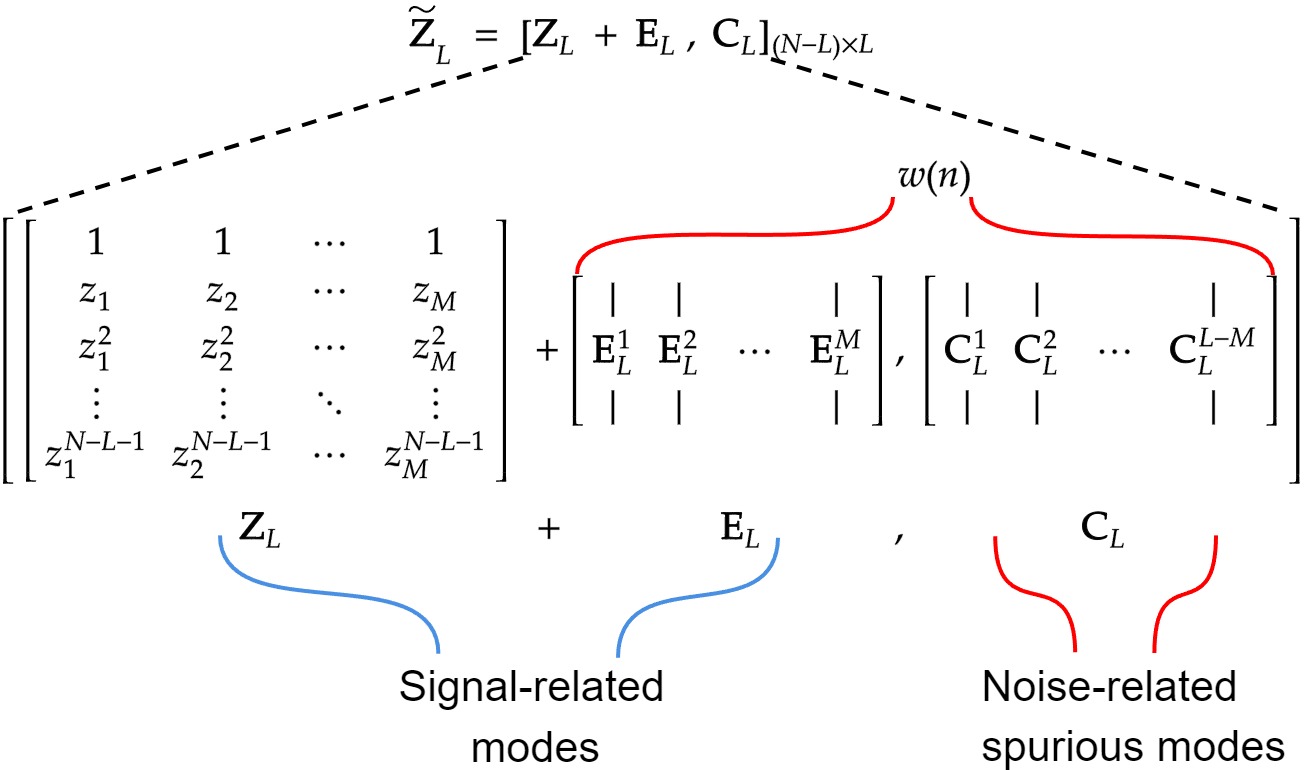}
    \caption{The structure of $\widetilde{\mathbf{Z}}_L$ by Proposition \ref{prop:Z_L recast}. The matrices $\mathbf{E}_L$ and $\mathbf{C}_L$ are a direct result of the additive noise $w(n)$.}\vspace{-3mm}
    \label{fig:noisy left mode structure}
\end{figure}

\subsection{Analysis of the noise induced matrix $\mathbf{E}_L$}
%\subsection{Analysis of noise effects} \label{sec: Analysis of noise effects}

The analysis of $\mathbf{E}_L$ is based on the analysis of the matrix $\mathbf{Y}_0^{\dagger}\mathbf{Y}_1$ as a perturbation of the matrix $\mathbf{X}_0^{\dagger}\mathbf{X}_1$, denoted by:
\begin{equation}\label{eq: perturbation of X0X1}
    \mathbf{Y}_0^{\dagger}\mathbf{Y}_1 = \mathbf{X}_0^{\dagger}\mathbf{X}_1 + \delta (\mathbf{X}_0^{\dagger}\mathbf{X}_1).
\end{equation}
Since the matrix $\mathbf{X}_0^{\dagger}\mathbf{X}_1$ has exactly $M$ distinct non-zero eigenvalues and another $L-M$ zero eigenvalues, the Rank-Nullity Theorem guarantees that its eigenvectors form a basis. 
Let $ \mathbf{T}_L^{\mathrm{H}} (\mathbf{X}_0^{\dagger}\mathbf{X}_1) \mathbf{T}_R = \mathbf{\Lambda}$ be the eigenvalue decomposition of $\mathbf{X}_0^{\dagger}\mathbf{X}_1$, where $\mathbf{\Lambda}$ is a diagonal matrix consisting of the eigenvalues $\{\lambda_k\}_{k=1}^L$, and $\mathbf{T}_L$ and $\mathbf{T}_R$ consist of the complete sets of left and right eigenvectors, $\{\mathbf{p}'_k\}_{k=1}^L$ and $\{\mathbf{q}'_k\}_{k=1}^L$, respectively, with $\mathbf{p}'^{\mathrm{H}}_k \mathbf{q}'_i = 0$ for $k \neq i$.
%
% Unless stated otherwise, it is assumed that all eigenvectors are of unit length.
%
Without loss of generality, suppose the eigenvalues $\{\lambda_i\}_{i=1}^M$ are the non-zero, simple, and equal the signal poles $\{z_i\}_{i=1}^M$, as guaranteed by the MP theorem for the noiseless case (Section \ref{subsec: MP noiseless case}).

For a simple eigenvector $\mathbf{q}'_i$, let $\widetilde{\mathbf{q}}'_i$ be its analogous perturbed eigenvector, which is an eigenvector of $\mathbf{Y}_0^{\dagger}\mathbf{Y}_1$ (see proof of Lemma \ref{lemma:first order approximation of the vector of coefficients} in Appendix \ref{app:simple eigenvector pert}, for more details).
Appropriate normalization of these eigenvectors is essential for the subsequent theoretical results, which is formulated in the following assumption.
\begin{assumption}\label{assum:eigenvectors normalization}
For $1\le i \le M$, we assume that $\mathbf{p}'^{\mathrm{H}}_i\widetilde{\mathbf{q}}'_i = 1$.
\end{assumption}
We remark that other alternative normalizations of $\widetilde{\mathbf{q}}'_i$, such as $ \widetilde{\mathbf{q}}'^{{\mathrm{H}}}_i \widetilde{\mathbf{q}}'_i = 1$ or $\mathbf{
q}'^{{\mathrm{H}}}_i\widetilde{\mathbf{q}}'_i = 1$, exist in the literature \cite{magnus1985differentiating}. However, using such normalization will lead, in the sequel, to significant numerical errors due to the inversion of an ill-conditioned matrix.
Next, we present a first-order approximation of the noise term $\mathbf{E}_L$ with respect to the additive perturbation in \eqref{eq: perturbation of X0X1}, using $\mathbf{W}_0$ and $\mathbf{W}_1$ to represent the noise-related Hankel matrices formed by the additive noise $w(n)$.

\begin{proposition}\label{prop:FOA of E_L}
    Under Assumption \ref{assum:eigenvectors normalization}, the $i$-th column of $\mathbf{E}_L \in \mathbb{C}^{(N-L)\times M}$ defined in Proposition \ref{prop:Z_L recast} is approximated by:
    \begin{equation} \label{eq:first order approx}
        \mathbf{E}_L^i \cong \sum\limits_{\substack{m=1\\m \neq i}}^M \frac{\mathbf{p}_m^{{\mathrm{H}}} (\mathbf{W}_1-z_i\mathbf{W}_0) \mathbf{q}_i}{b_i(\Tilde{z}_i-z_m)} \begin{bmatrix}
           1 \\
           z_m \\
           \vdots \\
           z_m^{N-L-1}
         \end{bmatrix} + \frac{\mathbf{W}_0 \mathbf{q}_i}{b_i},
    \end{equation} 
    provided that 
    \begin{equation}\label{eq:spectral radius}\rho(\mathbf{D}_i\mathbf{T}_L^{{\mathrm{H}}}\delta(\mathbf{X}_0^{\dagger}\mathbf{X}_1)\mathbf{T}_R) <1,    
    \end{equation}
    where $\rho(\mathbf{A})$ is the spectral radius of the matrix $\mathbf{A}$, and the matrix $\mathbf{D}_i$ is an $L\times L$ diagonal matrix defined by:
    \begin{equation}\label{eq:D_i} 
    \mathbf{D}_i = \text{diag}\Bigg(\frac{1}{\mathbf{p}'^{{\mathrm{H}}}_1\mathbf{q}'_1(\Tilde{z}_i - \lambda_1)}, \ldots,0_i ,\ldots,\frac{1}{\mathbf{p}'^{{\mathrm{H}}}_L\mathbf{q}'_L(\Tilde{z}_i - \lambda_L)}\Bigg).
    \end{equation}
\end{proposition}
\begin{proof}
See Appendix \ref{appen:first order approx of E_L} in the SM.
\end{proof}

It is noteworthy that the condition in \eqref{eq:spectral radius} and \eqref{eq:D_i} implies a delicate balance between the $i$-th pole separation and perturbation (for details, see Remark \ref{remark: sep_pert balance} in Appendix \ref{appen:first order approx of E_L}).
In addition, note that the first term in the right hand side (r.h.s.) of \eqref{eq:first order approx} vanishes if $\mathbf{W}_1 = z_i \mathbf{W}_0$, i.e., when the noise signal $w(n)$ consists only of the signal pole $z_i$. In such a degenerate case, the noise coincides with one of the signal's poles.
Finally, we remark that Proposition \ref{prop:FOA of E_L} does not imply that the noise term $\mathbf{E}_L$ is small. However, in the sequel, we show that the structure of $\mathbf{E}_L$ in \eqref{eq:first order approx} allows us to extract and exploit the obscured Vandermonde structure of the signal modes in $\widetilde{\mathbf{Z}}_L$. 

Using Propositions \ref{prop:Z_L recast} and \ref{prop:FOA of E_L}, any signal mode $\widetilde{\mathbf{Z}}_L^i$ can be expressed as: 
\begin{equation}\label{eq:signal_mode_compact}
    \widetilde{\mathbf{Z}}_L^i = \mathbf{Z}_L^i + \mathbf{E}_L^i,\,\, 1 \leq i \leq M,
\end{equation}
where $\widetilde{\mathbf{Z}}_L^{i}$, $\mathbf{Z}_L^{i}$ and $\mathbf{E}_L^{i}$ are the $i$-th columns of $\widetilde{\mathbf{Z}}_L$, $\mathbf{Z}_L$ and $\mathbf{E}_L$, respectively. 

Let $\mathbf{Q}_i \in\mathbb{C}^{(N-L+1)\times N}$ be a convolution matrix, constructed from the vector $\mathbf{q}_i \in \mathbb{C}^{L\times 1}$, such that $\mathbf{Q}_i \mathbf{v} = \mathbf{q}_i * \mathbf{v}$ and $\mathbf{u}^{{\mathrm{T}}} \mathbf{Q}_i  = \mathbf{u}^{{\mathrm{T}}} * \mathbf{q}_i $, for any two vectors $\mathbf{v}\in\mathbb{C}^{N\times 1}$ and $\mathbf{u}\in\mathbb{C}^{(N-L+1)\times 1}$. Additionally, denote by $\mathbf{a}(z)$ the Vandermonde column of length $N-L$, built from $z$:
\begin{equation}\label{eq: test_vector}
    \mathbf{a}(z) = [1,z,\dots,z^{N-L-1}]^{\mathrm{T}}.
\end{equation}
Combining these notations with \eqref{eq:signal_mode_compact} and Proposition \ref{prop:FOA of E_L}, we have the following result.
\begin{corollary}\label{cor:signal_mode}
Any signal mode $\widetilde{\mathbf{Z}}^i_L$ can be recast as:
\begin{equation}\label{eq:signal_mode}
    \widetilde{\mathbf{Z}}_L^i \cong \mathbf{a}(z_i) + 
         \sum\limits_{\substack{m=1\\m \neq i}}^M \gamma_{i,m}\mathbf{a}(z_m) + \boldsymbol\xi_i,
\end{equation} 
where 
\begin{equation}\label{eq:noise-terms defs}
\gamma_{i,m} = \frac{\mathbf{u}_m^{{\mathrm{T}}} \mathbf{Q}_i \mathbf{w}}{(\Tilde{z}_i-z_m)b_i},\;\; \boldsymbol\xi_i = \frac{\mathbf{I}_0\mathbf{Q}_i \mathbf{w}}{b_i},
\end{equation}    
the matrix $\mathbf{I}_0$ is given by:
\[
\mathbf{I}_0=
\begin{bmatrix}
\mathbf{I}_{N-L} , \;  \mathbf{0}\\
\end{bmatrix}_{(N-L)\times(N-L+1)},
\] 
where $\mathbf{I}_{N-L}$ is the identity matrix of size $N-L$ and $\mathbf{0}\in\mathbb{C}^{(N-L) \times 1}$ is the zero vector. The vector $\mathbf{u}_m \in \mathbb{C}^{ (N-L+1)}$ is defined by:
\begin{equation}\label{eq:u_m def}
    \mathbf{u}_m^{{\mathrm{T}}} = \begin{bmatrix} 0,\; \mathbf{p}_m^{\mathrm{H}} \end{bmatrix} -z_i \begin{bmatrix} \mathbf{p}_m^{\mathrm{H}}, \;0 \end{bmatrix} ,
\end{equation} 
and $\mathbf{w} = [w(0),\ldots, w(N-1)]^{{\mathrm{T}}}$ is the discrete noise vector introduced in \eqref{eq: noisy discrete signal} with $\mathbb{E}[\mathbf{w}\mathbf{w}^{\mathrm{H}}] = \sigma_w^2\mathbf{I}_N$.
\end{corollary}
% \begin{equation*} \label{eq:def of Qi}
% \mathbf{Q}_i = 
% \begin{bmatrix} 
% &\horzbar\;\mathbf{q}_i^{{\mathrm{T}}}\horzbar& & 0 & 0 & \cdots & 0 \\
% 0 & &\horzbar\;\mathbf{q}_i^{\mathrm{T}}\horzbar& & 0 & \cdots & 0 \\
% 0 & 0 & &\horzbar\;\mathbf{q}_i^{\mathrm{T}}\horzbar& & \cdots & 0 \\
% \vdots & \vdots & \vdots & \vdots & \ddots & \vdots & \vdots  \\
% 0 & 0 & \cdots & 0 & &\horzbar\;\mathbf{q}_i^{\mathrm{T}}\horzbar& \\
% \end{bmatrix},\\
% \end{equation*}
% constructed from the vector $\mathbf{q}_i \in \mathbb{C}^{L}$. 
%
% Note that the terms $\mathbf{u}_m \mathbf{Q}_i$ and $\mathbf{Q}_i \mathbf{w}$ are equivalent to the convolutions $\mathbf{u}_m * \mathbf{q}_i$ and $\mathbf{q}_i * \mathbf{w}$, respectively. 
%We use the conventional definition of discrete convolution, in which the convolution of two vectors $u$ and $v$ of length $n$ and $m$ is a vector denoted by $u * v$ of length $n+m-1$.

Equation \eqref{eq:signal_mode} shows that the signal modes consist of the true mode $\mathbf{Z}_L^i$, corrupted by two factors: (i) superimposed $M-1$ modes $\mathbf{Z}_L^m$, $m \neq i$ weighted by the coefficients $\gamma_{i,m}$, and (ii) a noise term $\boldsymbol\xi_i$. 
In the next result, we show that these two factors are controlled by the separation of the signal poles and the SNR at the $i$-th component, defined by:
\[
\text{SNR}_i = \frac{|b_i|^2}{\sigma_w^2},
\]
and in dB units as $\text{SNR}_{i,\text{dB}} = 10\text{log}_{10} (\text{SNR}_i)$.

\begin{proposition}\label{prop:bounds of noise-related terms}
    Under the conditions of Proposition \ref{prop:FOA of E_L}, for any $\varepsilon \in(0,1)$, the noise-related terms in \eqref{eq:signal_mode} are bounded by:
    \begin{align}
        &\left|\gamma_{i,m} \right| \;\leq  2\sqrt{\frac{2\text{log}(\frac{1}{\varepsilon})}{\text{SNR}_i}}
        \frac{\left\Vert \mathbf{u}_m *\mathbf{q}_i\right\Vert_2}{|z_i-z_m|}, \label{eq:gamma_ik upper bound}\\
        &\left\Vert \boldsymbol\xi_i \right\Vert_{\infty} \leq \sqrt{\frac{2\text{log}(\frac{1}{\varepsilon})}{\text{SNR}_i}},\label{eq:xi_i upper bound}
    \end{align} with a probability of at least $1- \varepsilon$, provided that:
    \begin{equation}\label{eq:perturb-separation assumption of prop 2}
        \left|z_i-z_m\right| \geq 2| \delta z_i |, \; 1 \leq m \neq i \leq M.
    \end{equation}
\end{proposition}
\begin{proof}
    See Appendix \ref{app:bounds of noise-related terms} in the SM.
\end{proof}
Note that by \eqref{eq:u_m def}, $\mathbf{u}_m$ consists of the vector $\mathbf{p}_m$. Hence, the convolution $\left\Vert \mathbf{u}_m *\mathbf{q}_i\right\Vert_2$ in \eqref{eq:gamma_ik upper bound} implicitly depends on the poles and the distances between the poles, as the vectors $\mathbf{p}_m^{\mathrm{H}}$ and $\mathbf{q}_i$ are given by the $m$-th row of $\mathbf{Z}_L^{\dagger}$ and $i$-th column of  $\mathbf{Z}_R^{\dagger}$, respectively. Additionally, we remark that the condition in \eqref{eq:perturb-separation assumption of prop 2} is common in perturbation analysis. Informally, this condition guarantees that the signal poles $z_i$ are sufficiently distinct relative to their perturbations. 

Proposition \ref{prop:bounds of noise-related terms} implies that the coefficients of the corrupting factors in \eqref{eq:signal_mode} can be bounded with high probability by a function of the $\text{SNR}_i$ and the $i$-th pole separation, provided that the separation is sufficiently large.
This allows us to present the following detection method, which is based on the structure of the signal modes in \eqref{eq:signal_mode}, and on $\Vert \mathbf{E}_L\Vert $.

\section{Structure-Aware Model Order Detection} \label{sec: model order}

In this section, we propose a detection method by identifying the signal modes. By relying on the modes and the estimated poles for the detection, this approach fundamentally differs from existing approaches, where only the singular values are considered for the model order detection.

Based on Corollary \ref{cor:signal_mode}, identifying the signal modes could be viewed as a relaxation of the original problem presented in \eqref{eq: noisy discrete signal} and \eqref{eq:noiseless discrete signal}; we are only interested in a \emph{binary} classification, i.e., to determine whether the $i$-th mode is a signal mode or not, without parameter estimation. 
In turn, once the signal modes are identified, their corresponding eigenvalues enable us to estimate the signal parameters.
Specifically, we propose to detect the signal modes through a model-matching approach by measuring the similarity of the modes to a pre-defined test vector. % similarly to the conventional beamformer \cite{krim2002two}.}
\subsection{Similarity Measure: Definition and Properties}
According to \eqref{eq:signal_mode}, the $i$‑th signal mode is modeled as a superposition of a Vandermonde vector built from the signal pole $z_i$, weighted Vandermonde vectors built from the other signal poles, and an additive noise term.
To exploit this model, we define a normalized similarity measure:
\begin{equation} \label{eq:Pi_def}
    P_i(z)\;:=\;\frac{\bigl|\mathbf a^{\mathrm{H}}(z)\,\widetilde{\mathbf{Z}}_L^i\bigr|^{2}}{\lVert\mathbf a(z)\rVert^{2}\lVert\widetilde{\mathbf{Z}}_L^i\rVert^{2}}, \quad 1\leq i\leq L,
\end{equation}
where $\mathbf{a}(z)$ is defined in \eqref{eq: test_vector}, %serves as a test vector, 
% \begin{equation}\label{eq: test_vector}
%     \mathbf{a}(z) = [1,z,\dots,z^{N-L-1}]^{\mathrm{H}},\quad z\in\mathbb{C}\setminus\{0\},
% \end{equation}
and analyze $P_i(z)$ for $1 \leq i \leq M$, where $\widetilde{\mathbf{Z}}_L^i$ is a signal mode.

In the absence of noise, $\widetilde{\mathbf{Z}}_L^i$ reduces to the Vandermonde vector $\mathbf{Z}_L^i$, and $P_i(z)$ attains a unique global maximum at the true signal pole, i.e., $P_i(z_i)=1$. 
%
%Following this intuition, we analyze $P_i(z_i)$ in the presence of noise and relate it to the maximal value of $P_i(z)$. 
%
In the presence of noise, applying triangle and Cauchy–Schwarz inequalities to \eqref{eq:Pi_def}, we obtain the following bound on $P_i(z_i)$:
\begin{equation}\label{eq: P_i(z_i) lower bound}
        0 \leq \left(\frac{1-\text{ISR}_i}{1+\text{ISR}_i}\right)^2 \leq  P_i(z_i) \leq 1,
\end{equation}
where $\text{ISR}_i$ is the Interference-to-Signal Ratio at the $i$-th mode:%, given by: 
\begin{equation}\label{eq: delta_i}
    \text{ISR}_i = \frac{\Vert \mathbf{E}_L^i\Vert}{\Vert \mathbf a(z_i)\Vert}.
\end{equation}
% and can be viewed as the Interference-to-Signal Ratio (ISR) at the $i$-th mode. 
% 
The lower bound in \eqref{eq: P_i(z_i) lower bound} implies that, in the presence of noise, the similarity measure at the true signal pole, $P_i(z_i)$, attains high values when the $\text{ISR}_i$ is small. Specifically, when $\text{ISR}_i \rightarrow 0$, $P_i(z_i) \rightarrow 1$.

By Prop. \ref{prop:bounds of noise-related terms}, the contributions of the other poles and the noise term in \eqref{eq:signal_mode} are controlled by the problem parameters, making $\Vert \mathbf{E}_L^i\Vert$, and consequently, $\text{ISR}_i$, small under some conditions. 
Specifically, for any $\varepsilon \in (0,1)$:
\begin{equation}\label{eq: delta_i upper bound}
    \text{ISR}_i \leq  \sqrt{\frac{2\log(\frac{2}{\varepsilon})}{\text{SNR}_i}}(2\Delta_i+1) \frac{\sqrt{N-L}}{\Vert \mathbf{a}(z_i) \Vert}, \quad \text{w.p. } \geq 1-\varepsilon,
\end{equation}
where $\Delta_i = \sum_{m\neq i} \frac{\left\Vert \mathbf{u}_m *\mathbf{q}_i\right\Vert_2}{|z_i-z_m|}$. 
This inequality presents an explicit condition on the problem parameters for which $\text{ISR}_i$ remains small, ensuring high values of $P_i(z_i)$. 
%In contrast, if the $i$-th mode is not a signal mode, it does not have any specific structure, and $P_i(z)$ may exhibit arbitrary behavior.
%Next, we leverage this distinction to design detection features for identifying the signal modes, which are also employed as an order detection rule.

\begin{remark}\label{remark: connection to DFT}
If $z\in \mathbb{C}$ lies on the unit circle at the $k$-th Discrete Fourier Transform (DFT) frequency, $z = e^{\frac{j2\pi k}{N-L}}$, the test vector \eqref{eq: test_vector} becomes:
    \[
        \mathbf{a}(z) = [1, e^{\frac{j2\pi k}{N-L}},\dots,e^{\frac{j2\pi k(N-L-1)}{N-L}}]^\mathrm{T},
    \]
    which coincides with the $k$-th column of a DFT matrix of length $N-L$. Therefore, in this case, the numerator of the similarity measure $P_i(z)$ in \eqref{eq:Pi_def} is given by:
    \[
    \bigl|\mathbf a^{\mathrm{H}}(z)\,\widetilde{\mathbf{Z}}_L^i\bigr|^{2} = \left|\sum\limits_{kn=0}^{N-L-1} \widetilde{\mathbf{Z}}_L^i[n]e^{\frac{-j2\pi kn}{N-L}}\right|^2 = \bigl|\mathcal{F}(\widetilde{\mathbf{Z}}_L^i)[k]\bigr|^2,
    \]
    where $\mathcal{F}(\widetilde{\mathbf{Z}}_L^i)[k]$ is the DFT of $\widetilde{\mathbf{Z}}_L^i$ at the $k$-th frequency bin.
    This means that evaluating the numerator in \eqref{eq:Pi_def} on the DFT grid is equivalent to analyzing the power spectrum of the mode, up to normalization.
\end{remark}
\vspace{-1mm}
\subsection{The Detection Method}\label{sebsec: detection method}

We propose to use the similarity measure $P_i(z)$ in \eqref{eq:Pi_def} for identifying the signal modes.
We showed above that the similarity measure at the true pole, $P_i(z_i)$, is high (under some conditions on the problem parameters), if the $i$-th mode $\widetilde{\mathbf{Z}}_L^i$ is a signal mode. Since the true pole $z_i$ is unknown, we compute the maximal value of $P_i(z)$ {over $\mathcal{D}\subset\mathbb{C}$}:
\begin{equation}\label{eq: epsilon measure}
    \epsilon_i := P_i(z_i^*),
\end{equation}
for $1\leq i \leq L$, where
\begin{equation*}
    z_i^* = \underset{{z\in \mathcal{D}}}{\text{argmax}}\;P_i(z),
\end{equation*}
{and 
\[   {
     \mathcal{D} = \{z\in\mathbb{C}| z = e^{-\alpha + j\theta}, \: \alpha\in[0,1], \:\theta \in [-\pi,\pi)\}}\footnote{{The interval $[0,1]$ covers the practically relevant range from undamped to strongly damped poles. If $\alpha$ is expected to lie outside this interval, $\mathcal{D}$ should be adjusted accordingly; otherwise, the maximizer of $P_i(z)$ may be biased.}}.
\]}
Intuitively, $\epsilon_i$ measures the maximal alignment between the $i$-th mode and the test vector $\mathbf{a}(z)$ defined in \eqref{eq: test_vector}. 
Large values of $\epsilon_i$ indicate that $\widetilde{\mathbf{Z}}_L^i$ is likely a structured signal mode, while small values suggest a spurious noise mode.
More precisely, by \eqref{eq: delta_i upper bound}, since $P_i(z_i^*) \ge P_i(z_i)$, signal modes yield large $\epsilon_i$ values, when a certain balance between the $\text{SNR}_i$, the $i$-th pole separation, the modes dimension $N-L$, and the $i$-th damping factor is maintained.
This is illustrated in Fig.~\ref{fig: P_i_z illustration}.

\begin{remark}
    We emphasize that $P_i(z)$ is not used for estimating the poles, i.e., $z_i^*$ is not considered an estimate of the pole, due to potential bias introduced by the noise terms in $\widetilde{\mathbf{Z}}_L^i$. The poles are extracted from the eigenvalues associated with the identified signal modes.
\end{remark}

%\vspace{0.5em}
%\noindent\textbf{Detection criterion.} 

The features $\{\epsilon_i\}_{i=1}^L$, computed in \eqref{eq: epsilon measure}, are used to determine the order of the model by dividing them into two distinct subsets: signal-related and noise-related. The number of the signal-related modes then determines the model order. 
Since $\epsilon_i \ge P_i(z_i)$, the theoretical lower bound of $P_i(z_i)$ in \eqref{eq: P_i(z_i) lower bound} naturally serves as a mode-dependent threshold: the $i$-th mode is classified as signal mode if: 
\begin{equation}\label{eq: theoretic feature threshold}
    \left( \frac{1-\text{ISR}_i}{1+\text{ISR}_i} \right)^2 \leq \epsilon_i.
\end{equation}

In summary, large values of $\epsilon_i$, occur when $\text{ISR}_i$ is small, indicate that $\widetilde{\mathbf{Z}}_L^i$ is likely a structured signal mode. The lower bound in \eqref{eq: P_i(z_i) lower bound} yields a mode-dependent threshold, controlled by the problem parameters, for discerning between signal and noise modes, allowing model order detection by counting the signal modes. 

\subsection{Adjustment for Closely-spaced Poles}
An additional advantage of determining the model order from the modes, rather than the singular values, is the ability to use additional information about the poles associated with each mode, which can be used to refine the features in \eqref{eq: epsilon measure}.

When clusters of poles are present, neighboring modes can corrupt $P_i(z_i)$ due to high sidelobes of the Poisson kernel \cite{katznelson2004introduction}. Specifically, the numerator of \eqref{eq:Pi_def} is given by:
\begin{equation}\label{eq: P_i explicit numerator}
    \mathbf a^{\mathrm{H}}(z_i)\,\widetilde{\mathbf{Z}}_L^i = \Vert \mathbf a(z_i)\Vert^2 + \sum\limits_{\substack{m=1\\m \neq i}}^M \gamma_{i,m} \mathcal{P}(z_i,z_m) + \mathbf a^{\mathrm{H}}(z_i) \boldsymbol{\xi}_i,
\end{equation}
where $\mathcal{P}(z_i,z_m) = \sum_{n=0}^{N-L-1}(z_i^*z_m)^n$ denotes the finite-order Poisson kernel. This effect amplifies when $z_i$ is close to other poles $z_m$ and reduced when $z_i$ is an isolated pole. To compensate for such leakage effects, we propose a normalized variant of \eqref{eq: epsilon measure}:
\begin{equation}\label{eq: epsilon measure normalized}
    \epsilon_i := \frac{1}{d_i}\cdot\underset{{z\in \mathcal{D}}}{\text{max}}\;P_i(z), \quad d_i = \sum\limits_{m=1}^L |\widetilde \lambda_m / \widetilde \lambda_i|^2,
\end{equation}
where $\widetilde \lambda_m$ and $\widetilde \lambda_i$ are the MP-estimated poles, $d_i$ is interpreted as a measure of the poles concentration near $\widetilde \lambda_i$. {Although the features in \eqref{eq: epsilon measure normalized} lie in $[0,1]$ by construction, their values often have limited dynamic range; we therefore rescale $\{\epsilon_i\}_{i=1}^L$ by their maximum.}

\begin{figure}[t]
    \centering
    \begin{subfigure}[b]{0.49\columnwidth}
        \centering
        \includegraphics[width=\linewidth]{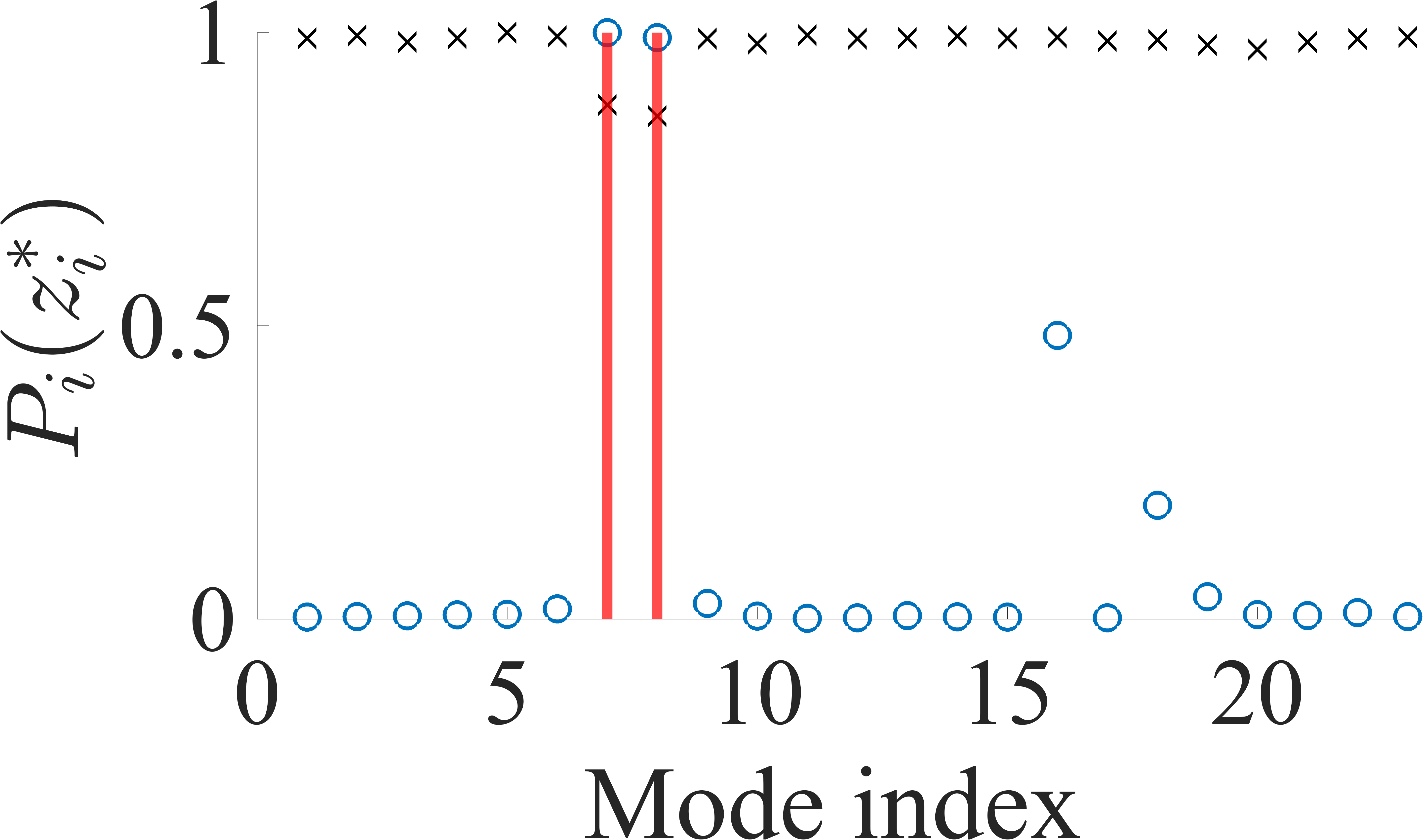}
        \caption{Undamped}
        \label{fig:sub1}
    \end{subfigure}
    \begin{subfigure}[b]{0.49\columnwidth}
        \centering
        \includegraphics[width=\linewidth]{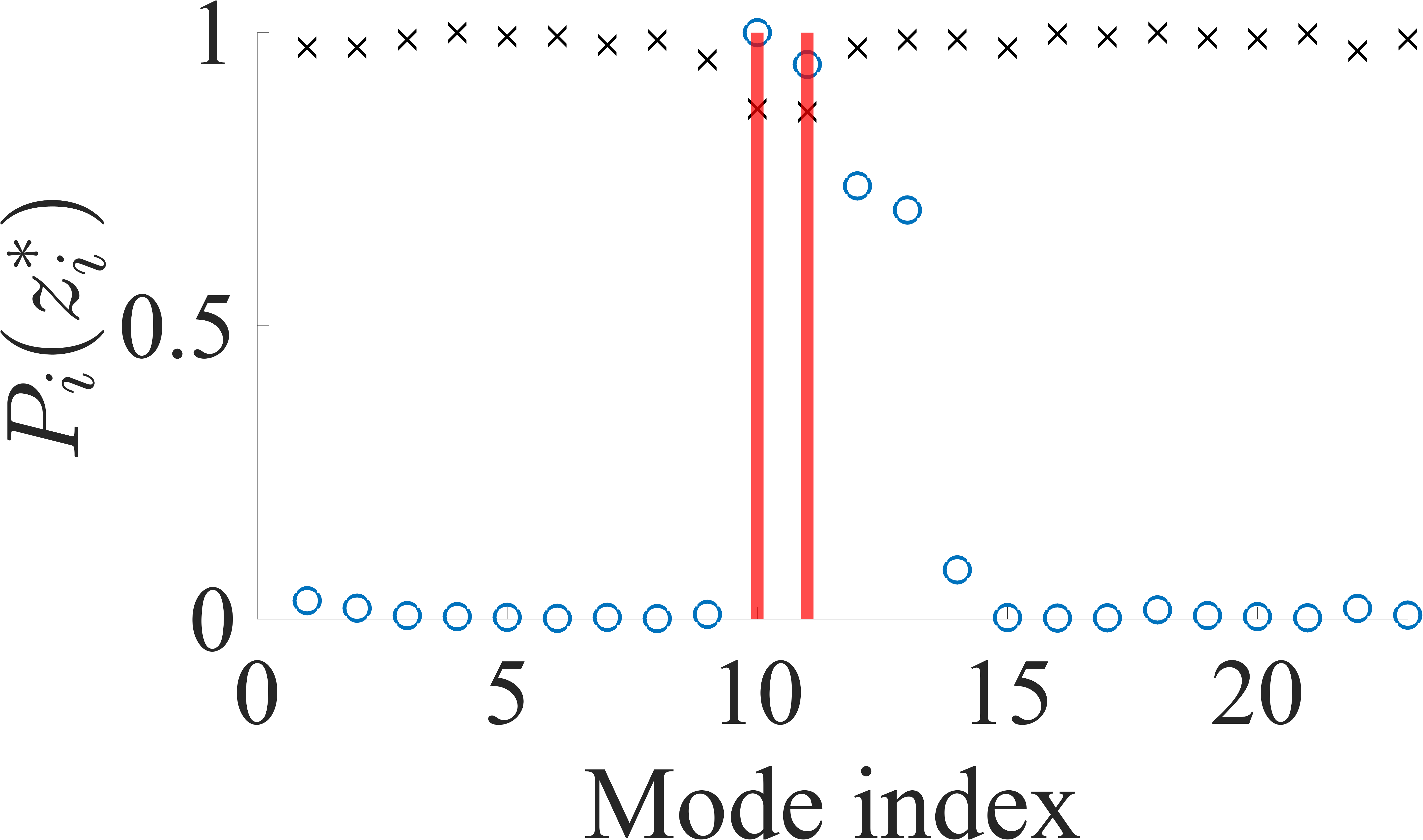}
        \caption{Damped}
        \label{fig:sub2}
    \end{subfigure}
    \caption{Illustration of $P_i(z_i^*)$ for each mode $1 \leq i \leq L$, with a signal composed of two complex exponentials, $N = 71$, $b_1 = b_2 = 1$, and $|\theta_2 - \theta_1| = 2\pi/N$. In the damped case, $\alpha_1 = 0.03$ and $\alpha_2 = 0.05$. Red lines mark the signal modes; black crosses denote the lower bound \eqref{eq: P_i(z_i) lower bound}. The $\text{SNR}_{\text{dB}}=10$.}\vspace{-3mm}
    \label{fig: P_i_z illustration}
\end{figure}

\begin{remark}
We leveraged the Vandermonde structure globally, by examining the similarity between the test vector in \eqref{eq: test_vector} and the modes. Alternatively, we could follow a local perspective, viewing each signal mode in \eqref{eq:signal_mode} as a column of a perturbed  Vandermonde matrix, demonstrating a near multiplicative structure. For example, in the context of DMD, Bronstein et al. \cite{bronstein2022spatiotemporal} introduced a detection feature based on the mean ratio of consecutive rows of the DMD modes (see Appendix \ref{app: Local perspectives for signal modes detection} in the SM for more details).
Our simulations show that using the global features in \eqref{eq: epsilon measure normalized} is superior to `local' approaches, and therefore, their empirical results are not presented in the sequel.
\end{remark}
\vspace{-4mm}
\section{Efficient Amplitudes Extraction} \label{sec: eff amp est}

We propose a computationally efficient approach for estimating the amplitudes, leveraging the structure of the MP modes as described in Theorem \ref{theo:MSMP}, Appendix \ref{app:Extending the MP Theorem: Accounting for Noise}.
{In these results, we show that the un-normalized modes defined in \eqref{def:right mp mode} and \eqref{def:left mp mode} can be recast as:
\begin{equation*}
        \widetilde{\mathbf{Z}}_L = 
    \begin{bmatrix} 
      \mathbf{Z}_L\mathbf{B}_L + \mathbf{E}_L\mathbf{B}_L, \; \mathbf{C}_L \\ 
    \end{bmatrix}, \quad 
        \widetilde{\mathbf{Z}}_R = 
    \begin{bmatrix}
    \mathbf{B}_R\mathbf{Z}_R +  \mathbf{B}_R\mathbf{E}_R \\
    \mathbf{C}_R
    \end{bmatrix},
\end{equation*}
where $\mathbf{B}_L,\mathbf{B}_R \in \mathbb{C}^{M \times M}$ are diagonal matrices satisfying $$\mathbf{B}_L \mathbf{B}_R = \mathbf{B}.$$

In theory, one may set $\mathbf{B}_L=\mathbf{B}_R=\mathbf{B}^{1/2}$, as in Theorem~\ref{theo:MSMP}. In practice, however, both $\mathbf{B}_L$ and $\mathbf{B}_R$ must be estimated. By Corollary~\ref{cor:signal_mode}, $\Vert\mathbf{E}_L\Vert\to 0$ as $\sigma_w^2\to 0$; thus, we view the un-normalized, left signal modes
\[
\widetilde{\mathbf{Z}}_{L,un}^{sig} = \mathbf{Z}_L\mathbf{B}_L + \mathbf{E}_L\mathbf{B}_L,
\]
as a perturbation of $\mathbf{Z}_L\mathbf{B}_L$. 
Since 
\begin{equation}
    \mathbf{Z}_L\mathbf{B}_L = 
    \begin{bmatrix}
        1                        & \cdots  & 1\\
        z_1              & \cdots  & z_M\\
        \vdots               & \ddots  & \vdots\\
        z_1^{N-L-1}     & \cdots  &  z_M^{N-L-1}
    \end{bmatrix}
    \begin{bmatrix}
    b_L^1 & 0 & \cdots & 0\\
    0 & b_L^2 & \cdots & 0\\
    \vdots & \vdots & \ddots & \vdots\\
    0 & 0 & \cdots & b_L^M
    \end{bmatrix},
\end{equation}
% $\mathbf{Z}_L\mathbf{B}_L$ is given by 
% \begin{equation}
%     \left[\mathbf{Z}_L\mathbf{B}_L\right]_{n,m} = z_m^{n-1}b_L^m, 
% \end{equation} where $n\in \{1,\dots,N-L-1\}$, $m\in\{1,\dots,M\}$, and 
the first row of $\mathbf{Z}_L\mathbf{B}_L$ equals $[b_L^1,\dots,b_L^M]$. Therefore, we estimate $\{b_L^m\}_{m=1}^M$ by:
\begin{equation*}
    \hat{b}_L^m =\left[\widetilde{\mathbf{Z}}_{L,un}^{sig}\right]_{1,m} = b_L^m (1+ \left[\mathbf{E}_L\right]_{1,m}).
\end{equation*}
In the same manner, we estimate $\{b_R^m\}_{m=1}^M$ by:
\begin{equation*}
    \hat{b}_R^m = \left[\widetilde{\mathbf{Z}}_{R,un}^{sig}\right]_{m,1} = b_R^m(1 + \left[\mathbf{E}_R\right]_{m,1}),
\end{equation*}
and propose the amplitude estimator
\begin{equation}
    \hat{b}_m = \hat{b}_L^m \hat{b}_R^m,\quad m\in\{1,\dots, M\}.
\end{equation}
Equivalently, this estimator can be written in a vectorized way by the element-wise product:
\begin{equation*}
    \widehat{\mathbf{b}} =(\mathbf{e}_1^{\mathrm{T}}\widetilde{\mathbf{Z}}_{L,un}^{sig})^{\mathrm{T}} \odot (\widetilde{\mathbf{Z}}_{R,un}^{sig}\mathbf{e}_1),
\end{equation*}
where $\odot$ denotes element-wise multiplication, the vector $\mathbf{e}_1$ is the first standard unit vector, and  $\widehat{\mathbf{b}}$ consists of the estimated amplitudes.

In practice, the indices of the signal modes are unknown a priori, and we use the entire set of modes:}
\begin{equation} \label{eq:new_amps}
\widehat{\mathbf{b}} =(\mathbf{e}_1^{\mathrm{T}} \widetilde{\mathbf{Z}}_{L, un})^{\mathrm{T}} \odot (\widetilde{\mathbf{Z}}_{R, un} \mathbf{e}_1) \in \mathbb{C}^L,
\end{equation}
where the $i$-th entry of $\widehat{\mathbf{b}}$ is given by:
\begin{equation*}
\hat{{b}}_i = 
\begin{cases}
b_i(1+ \mathbf{E}_L(1,i))(1+ \mathbf{E}_R(i,1)) &, \text{if}\ 1 \leq i \leq M, \\
\mathbf{C}_L(1, i-M) \mathbf{C}_R(i-M,1) &, \text{if}\ M < i \leq L
\end{cases}
\end{equation*}
{followed by a selection of the signal-related amplitudes as described in Algorithm \ref{alg:SAMP algo}.}

This method involves simple element-wise multiplication of two vectors in $\mathbb{C}^L$, requiring  $\mathcal{O}(L)$ operations. The existing method in \eqref{eq:MP coeff estimation} involves pseudo-inverting a matrix of size $N \times \widehat M$, where $\widehat M$ is based on the truncated-SVD of $\mathbf{Y}_0$, and hence $1\leq \widehat M \leq L < N$. The computational complexity of the existing method is $\mathcal{O}(N\widehat{M}^2 + \widehat{M}^3)$. As noted in Section \ref{sec: MP method}, estimating the model order using the truncated-SVD step is prone to over- or underestimation. In the extreme case of underestimation, where $\widehat M=1$, the existing method requires $\mathcal{O}(N)$ operations, while if $\widehat M=L$, the existing method requires $\mathcal{O}(NL^2 + L^3)$ operations. Either way, the proposed approach is more efficient.
In Section \ref{sec:Results}, we compare its performance to the existing method for estimating the signal amplitudes given in \eqref{eq:MP coeff estimation}, which is based on matrix inversion.

We highlight that similarly to Section~\ref{sec: model order}, where spectral information is extracted from the MP modes to identify signal-related components, in \eqref{eq:new_amps}, temporal information is extracted from the MP modes to estimate the noisy amplitudes efficiently. This new approach of extracting meaningful information from the MP modes fundamentally differs from existing approaches, where the MP modes (or eigenvectors) are not typically recognized as carrying valuable information for estimating the signal parameters.
{
\begin{remark}
    We present a simple and computationally efficient approach that uses only the first row/column of the modes, which depends solely on the amplitudes. While one may attempt to exploit additional information from the remaining rows/columns, such alternatives rely on pole-estimates and demonstrate inferior empirical results (see Appendix~\ref{app: amlitude estimation alternatives} for further discussion and comparisons).
\end{remark}}

%% -----------------------------------------------------------------------------------------------------%%
\vspace{-3mm}
\section{Proposed Algorithm} \label{sec: Proposed Algorithm}

Relying on the analysis and results regarding model order detection and amplitude estimation in Sections \ref{sec: model order} and \ref{sec: eff amp est}, we present the Structure-Aware Matrix Pencil (SAMP) algorithm, which is summarized in Algorithm \ref{alg:SAMP algo}. Compared with the standard MP method, our main novelties are a new model order detection approach, and an efficient amplitude estimator.

\vspace{-3mm}
\subsection{Practical Considerations} \label{subsec: practice}
\emph{(i) Weak truncation.}
Although the proposed model order detector does not require truncation in \eqref{eq:SVD of Y_0}, it is known, and supported by our numerical studies, that removing strong noise-related components reduces the noise effects in the estimates \cite{hua1991svd, sarkar1995using}. Therefore, we {apply} a weak truncation in \eqref{eq:SVD of Y_0} to eliminate strong noise components only. Alternatively, one can assume, for example, that $2M\leq L$ and truncate the smallest $\left\lfloor \frac{L}{2} \right\rfloor$ singular values in \eqref{eq:SVD of Y_0}.
{In our setup, we do not assume additional prior knowledge of $M$, except for the standard MP assumption, that $M \leq L\leq N-M$, stated in Section~\ref{sec: MP method}.} By our simulations, an example of a weak-truncation method is the effective rank method \cite{roy2007effective}.

The effective rank method estimates the dimensionality or rank of a matrix in the presence of noise by utilizing the entropy of its normalized singular values. Given a set of singular values \(\{\sigma_i\}_{i=1}^L\), the normalized singular values are defined as \(\bar\sigma_i = \frac{\sigma_i}{\sum_{j=1}^L \sigma_j}\). The entropy \(H(\boldsymbol{\bar\sigma})\) of these normalized singular values is calculated using \(H(\boldsymbol{\bar\sigma}) = -\sum_{i=1}^L \bar\sigma_i \text{log}(\bar\sigma_i)\), and the effective rank is then defined as $\text{exp}(H(\boldsymbol{\bar\sigma}))$. The model order is estimated by:
\begin{equation}\label{eq:effective-rank}
    \widehat{M}_{\text{ER}} = \text{round}(\text{exp}\left(H(\boldsymbol{\bar\sigma}))\right).
\end{equation}
We highlight that the proposed SAMP algorithm remains valid without truncation, at the cost of increased bias at high SNRs.

\emph{(ii) Computing the maximal value of $P_i(z)$}.
    The features in \eqref{eq: epsilon measure} or \eqref{eq: epsilon measure normalized} require finding the maximal value of $P_i(z)$, for $z\in\mathcal{D}$.
    Solving this problem typically involves a coarse grid search followed by a local refinement step (to mitigate grid-resolution effects and to handle off-grid poles), such as a fine-grid search, two-dimensional interpolation, or a derivative-based refinement \cite{candan2011method}. In our simulations, we employ a coarse grid search followed by biquadratic interpolation using a $3\times 3$ stencil, which provides a simple, efficient, and accurate procedure (see Section~\ref{subsec: coarse-grid granularity} {for details on the stencil construction,} a discussion on grid selection, and a sensitivity analysis of the proposed algorithm with respect to grid granularity).

\emph{(iii) Practical threshold for $\epsilon_i$.}
While the threshold in \eqref{eq: theoretic feature threshold} provides a theoretically grounded criterion for identifying the signal modes, estimating $\text{ISR}_i$ in practice is challenging.
To address this, we propose a simplified variant of the detection rule in \eqref{eq: theoretic feature threshold}. 
{Using the bound in \eqref{eq: delta_i upper bound}, we observe that 
\begin{equation}
   \text{ISR}_i \leq c_i\frac{\sqrt{N-L}}{|b_i| \Vert \mathbf a(z_i)\Vert},
\end{equation}
where $c_i\in \mathbb{R}$ is a tunable parameter, capturing the remaining terms in \eqref{eq: delta_i upper bound}, which depend on $\Vert \mathbf{E}_L^i \Vert$.}

% Using the structure of $\mathbf{E}_L^i$ in Proposition \ref{prop:FOA of E_L}, we observe that

% $\Vert \mathbf{E}_L^i \Vert$ is inversely proportional to $|b_i|$, and therefore, the $\text{ISR}_i$ can be written as:
% \[
% \text{ISR}_i = \frac{c_i}{|b_i| \Vert \mathbf a(z_i)\Vert},
% \]
% where $c_i\in \mathbb{R}$ is a tunable parameter, capturing the remaining terms in $\Vert \mathbf{E}_L^i \Vert$.
Since the true amplitudes and poles $b_i$ and $z_i$ are unknown, we use their estimates $\hat b_i$ and $\widetilde \lambda_i$, and define the following simplified threshold: 
the $i$-th mode is classified as signal mode if: 

\begin{equation}\label{eq: feature threshold}
        {\left( \frac{1-c_i\sqrt{N-L}/|\hat b_i|\Vert \mathbf a(\widetilde \lambda_i)\Vert}{1+c_i\sqrt{N-L}/|\hat b_i|\Vert \mathbf a(\widetilde \lambda_i)\Vert} \right)^2 \leq \epsilon_i,}
\end{equation}
{where we set $c_i=10$ in our experiments and provide a sensitivity analysis in Appendix~\ref{app: c_i sensitivity analysis}.
Notably, $c_i$ may also be tuned individually per mode to improve flexibility, using estimates of the quantities in \eqref{eq: delta_i upper bound}. However, similar to Appendix~\ref{app: amlitude estimation alternatives}, where we considered amplitude-estimation alternatives that rely on additional estimated parameters, our simulations indicate that such additions, although potentially informative, introduce estimation errors that degrade performance. Nevertheless, per-mode tuning of $c_i$ remains an open direction for future work.}

{
\emph{(iv) Operating regime and practical guidelines.}
In practice, the proposed method is expected to yield the largest gains when $\text{ISR}_i \ll 1$. From \eqref{eq: delta_i upper bound}, a sufficient condition for this regime is:
\begin{equation}\label{eq: practical condition}
    \sqrt{2\log(\frac{2}{\varepsilon})}(2\Delta_i+1)^2 \frac{N-L}{\Vert \mathbf{a}(z_i) \Vert^2} \ll \text{SNR}_i.
\end{equation}
Condition \eqref{eq: practical condition} is easier to satisfy under (i) mild damping, for which $\|\mathbf{a}(z_i)\|^2 \approx N-L$, and (ii) sufficiently large frequency separation, which yields smaller $\Delta_i$.

As a representative example, consider the signal in \eqref{eq:noiseless discrete signal} with two unit-amplitude exponentials. For $|\theta_2-\theta_1|=2\pi/N$ and $\alpha_1=\alpha_2=0$, \eqref{eq: practical condition} holds for $\text{SNR}_i>10$ dB when $N>50$. For fixed $N=100$, it holds for $\text{SNR}_i>10$ dB when $|\theta_2-\theta_1|\geq 1.2\pi/N$. Finally, for $|\theta_2-\theta_1|=2\pi/N$ and $N=100$, it holds for $\text{SNR}_i>10$ dB when $\alpha_1+\alpha_2 \leq 0.05$.

In contrast, when $\text{ISR}_i$ is large, the performance of SAMP is expected to degrade. The simulations in Section~\ref{sec:Results} further illustrate these trends across a wide range of scenarios.}

\begin{remark}
The proposed SAMP algorithm (Algorithm \ref{alg:SAMP algo}) deviates from the conventional ``detect and then estimate'' strategy. Instead, it adopts an ``estimate and then select'' approach, wherein the model parameters are first estimated from the eigenvalues and modes, and then, the signal-related components are selected based on these estimated parameters. 
\end{remark}

\begin{algorithm}[t]
\caption{SAMP proposed algorithm}
\label{alg:SAMP algo}
\textbf{Input}: Noisy measurements $\{y(n)\}_{n=0}^{N-1}$, pencil parameter $L$, and the tunable parameters $\{c_i\}_{i=1}^L$.\\
\textbf{Output}: Estimates of the signal frequencies, damping factors, and amplitudes. 
\begin{algorithmic}[1]
    \STATE{Modes and Eigenvalues Computation:}
    \begin{itemize}
    \item {Construct the Hankel matrices $\mathbf{Y}_0, \mathbf{Y}_1$ from the measured signal $\{y(n)\}_{n=0}^{N-1}$ by \eqref{eq:full noisy Hankel matrix}. }
    \item {Compute the truncated-SVD of $\mathbf{Y}_0 \approx \mathbf{U} \mathbf{\Sigma} \mathbf{V}^{\mathrm{H}}$, using the effective-rank method by \eqref{eq:effective-rank}.}
    \item {Decompose the matrix $\mathbf{A} = \mathbf{\Sigma}^{-1}\mathbf{U}^{\mathrm{H}}\mathbf{Y}_1\mathbf{V}$ to get $\mathbf{A} \mathbf{Q} = \mathbf{Q} \widetilde{\mathbf{\Lambda}}$, by \eqref{eq:A decompos}.}
    \item {Compute the MP modes $\widetilde{\mathbf{Z}}_{R,un}=\mathbf{Q}^{-1}\mathbf{V}^{\mathrm{H}}$ and $\widetilde{\mathbf{Z}}_{L,un}=\mathbf{U \Sigma Q}$, by \eqref{def:right mp mode}, and \eqref{def:left mp mode}.}
     \end{itemize}
        \STATE{Parameter Estimation:}
        \begin{itemize}
            \item Estimate the frequencies and damping factors by \eqref{eq:params calc}:
                \begin{equation*}
                        \hat{\theta}_i = \text{arg}(\widetilde {\lambda}_i), \;\;
                        \hat{\alpha}_i = -\text{log}|{\widetilde\lambda}_i|, \;\; 1\leq i \leq L.
                \end{equation*}
            \item Estimate the complex amplitudes by \eqref{eq:new_amps}: $$\widehat{\mathbf{b}} =(\mathbf{e}_1^{\mathrm{T}} \widetilde{\mathbf{Z}}_{L,un})^{\mathrm{T}} \odot (\widetilde{\mathbf{Z}}_{R,un} \mathbf{e}_1)$$
        \end{itemize} 
    
    \STATE{Model Order Detection:}
        \begin{itemize}
            \item Compute the $L$ features by \eqref{eq: epsilon measure normalized}: $$\epsilon_i = \frac{1}{d_i}\cdot\underset{{z\in \mathcal{D}}}{\text{max}}\;P_i(z),\;\; 1\leq i \leq L,$$ and {rescale them by $\max\{\epsilon_i\}$ to span the full $[0,1]$ range.} 
           
            \item {Partition the features into two distinct subsets by defining the signal-related subset as the subset of features satisfying \eqref{eq: feature threshold}:
            \[
                {\left( \frac{1-c_i\sqrt{N-L}/|\hat b_i|\Vert \mathbf a(\widetilde \lambda_i)\Vert}{1+c_i\sqrt{N-L}/|\hat b_i|\Vert \mathbf a(\widetilde \lambda_i)\Vert} \right)^2 \leq \epsilon_i.}
            \]
            } \label{algo: division step}
            \item Set the model order to:
            $
                \widehat{M} = \left| \mathcal{S} \right|,
            $ where $\mathcal{S}$ is the set of indices corresponding to the signal-related subset.
        \end{itemize}

    \STATE{Parameter Selection:}
        \begin{itemize}
        \item{Select the signal components by: $\hat{\theta}_s, \, \hat{\alpha}_s, \, \hat{b}_s, \;\; s\in \mathcal{S}$.}
        \end{itemize}
        
\end{algorithmic}
\end{algorithm} 
\vspace{-3mm}

\subsection{Computational Complexity}\label{sebsec: comp complex}
We evaluate the computational complexity of the proposed algorithm by counting the number of complex operations required for each step, assuming that the number of complex operations required for multiplying two matrices of size $a\times b$ and $b\times c$ is $\mathcal{O}(2abc)$.

In step 1, we compute the modes and eigenvalues. This step requires $\mathcal{O}((N-L)^2L+6L^2(N-L)+L^3 +4L)$ operations for computing the eigenvalues and additional $\mathcal{O}(4(N-L)L^2 + 3L^3)$ operations for computing the left and right modes.
In step 2, we estimate the signal parameters. This step requires $\mathcal{O}(4L)$ operations. 
In step 3, we determine the model order. First, we compute the $L$ measures $\{\epsilon_i\}_{i=1}^L$ using \eqref{eq: epsilon measure normalized}, and then partition $\{\epsilon_i\}_{i=1}^L$ into two distinct subsets.
{Computing $\epsilon_i$ involves computing $d_i$, and maximizing $P_i(z)$ for each $1 \leq i \leq L$. 
The computation of $d_i$ requires $\mathcal{O}(4L)$ operations. As described in Section~\ref{subsec: practice}, maximizing $P_i(z)$ is performed in two stages: a coarse-grid search over a $K_{\alpha} \times K_{\theta}$ grid, followed by a fixed-size biquadratic interpolation. Direct evaluation of $P_i(z)$ requires $\mathcal{O}(3(N-L))$ operations for each grid point, assuming that the mode norms are precomputed, and finding the maximum over the grid requires $\mathcal{O}(K_{\alpha}K_{\theta}-1)$ operations.
The interpolation stage uses a fixed-size stencil and therefore is independent of $N$, and requires $\mathcal{O}(1)$ operations.
Partitioning the set $\{\epsilon_i\}_{i=1}^L$ into two subsets using \eqref{eq: feature threshold} requires $\mathcal{O}(2(N-L)L + L)$ operations.
In our simulations, $K_{\alpha} = 2(N-L)+1,\: K_{\theta} = 2(N-L)$ (see Section~\ref{subsec: coarse-grid granularity}), and hence a total of $\mathcal{O}(4(N-L)^3(L+2) + 2(N-L)^2(3L+2) + 4L(N-L) + 4L^2 + L)$ operations are required for this step.}
In step 4, we select the signal-related parameters. This step requires $\mathcal{O}(\widehat{M})$ operations.
{For large $N$, and assuming $L=N/3$, Algorithm~\ref{alg:SAMP algo} has leading-order complexity of $\mathcal{O}(N^4)$, whereas the thresholding-based MP algorithm has leading-order complexity of $\mathcal{O}(N^3)$.}

% For comparison, the complexity of the standard MP algorithm is $\mathcal{O}(2(N-L)\widehat{M}^2 + 5\widehat{M}^3-3\widehat{M}^2)$ for computing the eigenvalues, $\mathcal{O}(4\widehat{M})$ for estimating the signal parameters from the eigenvalues, and another $\mathcal{O}(N\widehat{M}^2 +\widehat{M}^3)$ for estimating the amplitudes.

% and the most demanding step in both methods is the calculation of the matrix $\mathbf{A}$ which involves the multiplication of four matrices.

\vspace{-1mm}
\section{Numerical results}
\label{sec:Results}

In this section, we present simulation results focusing on closely-spaced frequencies under various SNR levels. We begin by showcasing the detection and estimation capabilities of the proposed SAMP algorithm (Algorithm \ref{alg:SAMP algo}). Then, we demonstrate the proposed amplitude estimation method and compare it in terms of computation time and accuracy to the traditional amplitude estimation method in \eqref{eq:MP coeff estimation}.

We compare the proposed SAMP algorithm with five other methods based on singular value truncation. Two methods are based on ITE: AIC \cite{stoica2005spectral} and EVT \cite{nadler2011model} criteria. These methods use the \textit{coupled detection and estimation} approach, which requires prior knowledge of the noise distribution (that is unknown in our setting), and involves computing the likelihood function for every possible hypothesized model order, making them computationally inefficient, as will be evident in the sequel. The EVT criterion is used because it was shown to outperform the common MDL and MAP criteria, and to be less sensitive to noise model mismatch \cite{nadler2011model}.

The remaining baseline methods use the \textit{decoupled detection and estimation} approach, which is the standard practice in the MP literature. In particular, we examine the Significant Decimal Digits (SDD) method (see Section \ref{subsec:MP method - model order detection}), a widely used method in the MP literature \cite{hua1990matrix, hua1988matrix, sarkar1995using, laroche1993use, bhuiyan2012advantages}. Additionally, we evaluate the GAP method, introduced in Sec. \ref{subsec:MP method - model order detection}, which is another common approach for SVD truncation \cite{del1996comparison, yin2011model}. For a fair comparison, we also include the effective rank method (EFF) \cite{roy2007effective}, as we use it in our algorithm as a pre-processing step to eliminate strong noise-related components. 

Throughout this section, we simulate two scenarios involving two and four closely-spaced complex exponentials, generated according to \eqref{eq:noiseless discrete signal}. In both cases, $\{b_i=1\}_{i=1}^M$, and $L = \text{round}(N/3)$.
In each scenario, we conduct two experiments. The first involves \emph{undamped} exponentials, with $\{\alpha_i = 0\}_{i=1}^M$. The second involves \emph{damped} exponentials: for the two-components case, we set $\alpha_1 = 0.03$, and $\alpha_2 = 0.05$; for the four-components case, we set $\alpha_1 = \alpha_3 = 0.03$, and $\alpha_2 = \alpha_4 = 0.05$.
Our goal is twofold. Given a finite sample of the noisy signal $y(n)$, as in \eqref{eq: noisy discrete signal}, we aim to detect the model order (the number of complex exponentials), which is $M=2$ or $M=4$ in our simulations, and estimate the corresponding $M$ frequencies, damping factors, and amplitudes.
{We highlight that all reported results correspond to off-grid simulations (see Section~\ref{subsec: coarse-grid granularity} for more details).}
The code to reproduce all the results in this section is available in the following \href{https://github.com/YehonatanSeg/--SAMP--Stracture-Aware_Matrix_Pencil_Method.git}{GitHub link}.

%----------------------------
\vspace{-3mm}
\subsection{Model order detection}
We determine the probability of correctly identifying the model order, $p_d$, defined by:
\[
    p_d = \mathbb{P}(\widehat{M}=M) \cong \frac{1}{N_{\text{exp}}}\sum_{i=1}^{N_{\text{exp}}} \mathbbm{1}_{\widehat{M}=M},
\] 
where $N_{\text{exp}}=500$ is the number of independent Monte-Carlo trials, and $\mathbbm{1}_{\widehat{M}=M}$ is an indicator function that equals $1$ if the estimated model order $\widehat{M}$ matches the true model order $M$, and $0$ otherwise. To simplify comparisons, we report the Area Under the Curve (AUC) for each method.

Fig. \ref{fig: detection_probability_SNR}  displays the probability of correct detection $p_d$ as a function of the SNR in dB for $N=71$ samples. In Figs. \ref{fig: pd_SNR_2_undmp}-\ref{fig: pd_SNR_2_dmp}, we simulate $M=2$ exponentials, with $\theta_1=2 \; \text{rad/sample}$, and $\theta_2= \theta_1 + \frac{2\pi}{N}$. In Figs. \ref{fig: pd_SNR_4_undmp}-\ref{fig: pd_SNR_4_dmp}, we simulate $M=4$ exponentials, divided into two clusters, with $\theta_1=2 \; \text{rad/sample}$, $\theta_2= \theta_1 + \frac{2\pi}{N}$, and $\theta_3 = -\theta_1$, $\theta_4=-\theta_2$.
The value of $\theta_2$ is set because the difference $|\theta_2 - \theta_1| = \frac{2\pi}{N}$ is the Rayleigh limit \cite{dharanipragada1996resolution}. We note that theoretical frequency resolution limits are also discussed in \cite{shahram2005resolvability, amar2008fundamental}. 
The SNR is defined as: 
\[
\text{SNR} = \sum\limits_{i=1}^M \frac{|b_i|^2}{\sigma_w^2},
\]
and we denote $\text{SNR}_{\text{dB}} = 10\text{log}_{10}(\text{SNR})$.
Fig. \ref{fig: detection_probability_SAMPLE} displays $p_d$ versus the number of samples $N$ for $\text{SNR}_{\text{dB}}$ of $8$ dB. 
We simulate $M=2$ exponentials in Figs. \ref{fig: pd_SAMPLE_2_undmp}-\ref{fig: pd_SAMPLE_2_dmp} and $M=4$ exponentials in \ref{fig: pd_SAMPLE_4_undmp}-\ref{fig: pd_SAMPLE_4_dmp}, where the frequencies spacing is the same as in Fig. \ref{fig: detection_probability_SNR}.
Fig. \ref{fig: detection_probability_DIFF} displays $p_d$ versus the frequencies spacing $|\Delta \theta|$ for $\text{SNR}_{\text{dB}}$ value of $10$ dB and $N=71$ samples. We simulate $M=2$ exponentials in Figs. \ref{fig: pd_DIFF_2_undmp}-\ref{fig: pd_DIFF_2_dmp} and $M=4$ exponentials in \ref{fig: pd_DIFF_4_undmp}-\ref{fig: pd_DIFF_4_dmp}, divided into two clusters. In the case of $M=4$, the varying frequencies spacing is within each cluster: $|\Delta \theta| = |\theta_2-\theta_1|=|\theta_4-\theta_3|$.

The proposed SAMP method outperforms the GAP, SDD, and EFF methods by a wide margin in all the presented scenarios. 
We note that the EFF method fails as a stand-alone detection method, but we use it as a pre-processing step in the proposed SAMP method to eliminate strong noise components. Empirically, we see that the EFF method consistently overestimates the model order, even at high SNR.
Additionally, the SAMP method also outperforms the EVT and AIC criteria in all the presented cases. When comparing the undamped and damped cases, we observe a general degradation in performance; however, the relative trends remain the same, with SAMP outperforming all the baseline methods. {Overall, SAMP achieves an average AUC that is approximately $22\%$ higher than that of the best competing method.}

\textit{Noise distribution effects:} To assess the impact of deviations from the Gaussian noise assumption, we also simulate non-Gaussian additive noise. Table \ref{table:auc_comparison} shows the AUC of $p_d$ versus the $\text{SNR}_{\text{dB}}$ (in the range of $[-10,20]$ dB), for each method under both normal and bi-normal noise distributions, in the case of $M=2$ (the complementary results for $M=4$ show similar trends and are omitted for brevity). 
The first experiment (second column from left) depicts the results of a normal distribution, as seen in Figs. \ref{fig: pd_SNR_2_undmp}-\ref{fig: pd_SNR_2_dmp}. In the second experiment (third column from left), we simulated a Bi-normal distribution:
\[
    \mathcal{B}_d = \mathbbm{1}_{p<r}\cdot\mathcal{N}(\mu_1,\sigma_1) + \mathbbm{1}_{p\geq r}\cdot\mathcal{N}(\mu_2,\sigma_2),
\] where $p$ is a random number drawn from the standard uniform distribution on the open interval $(0,1)$, and $r$ is a pre-defined threshold, set to $r=0.85$ in our simulations. 
The results demonstrate that SAMP is moderately sensitive to deviations from the assumed noise distribution, with an average degradation of approximately $5\%$. In contrast, ITE methods such as EVT and AIC are more affected, with an averaged degradation of approximately $14\%$ to EVT and $17\%$ to AIC.
Interestingly, GAP remains largely insensitive, likely due to its reliance solely on spectral thresholding, but its detection capabilities are low.
Due to poor performance, the SDD and EFF methods yield AUC values near zero and cannot be analyzed.

%% Model order detection plots----------------------------------

\begin{figure}[t]
    \centering
    \begin{subfigure}[b]{0.49\columnwidth}
        \centering
        \includegraphics[width=\linewidth]{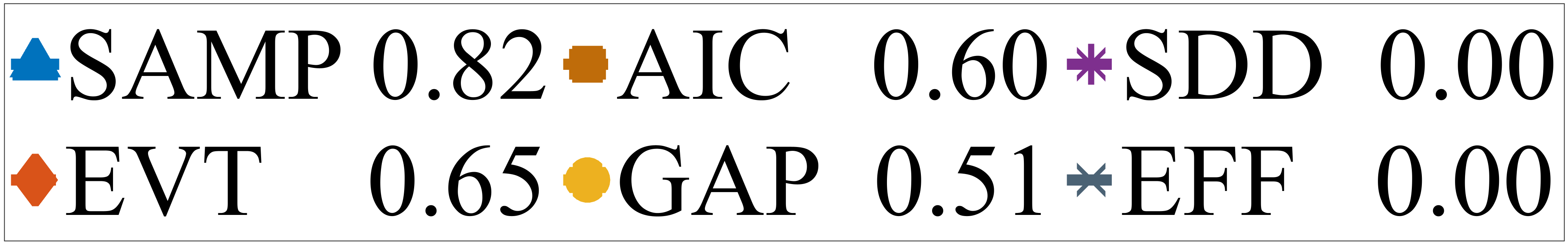}
    \end{subfigure}
    \begin{subfigure}[b]{0.49\columnwidth}
        \centering
        \includegraphics[width=\linewidth]{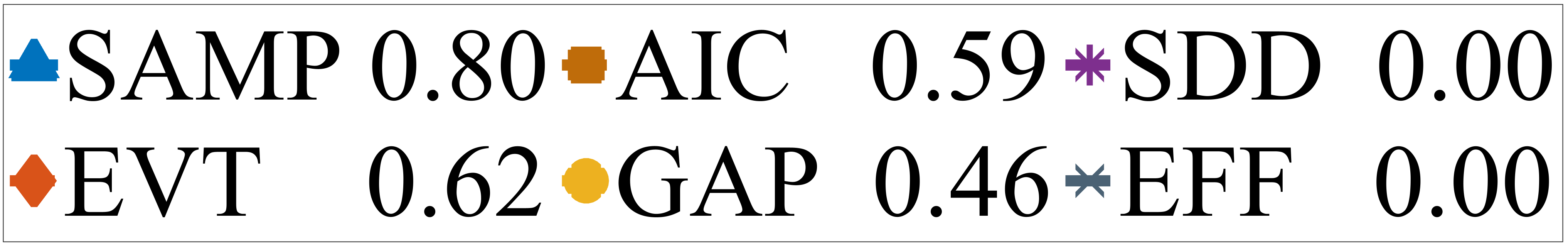}
    \end{subfigure}
    
    \begin{subfigure}[b]{0.49\columnwidth}
        \centering
        \includegraphics[width=\linewidth]{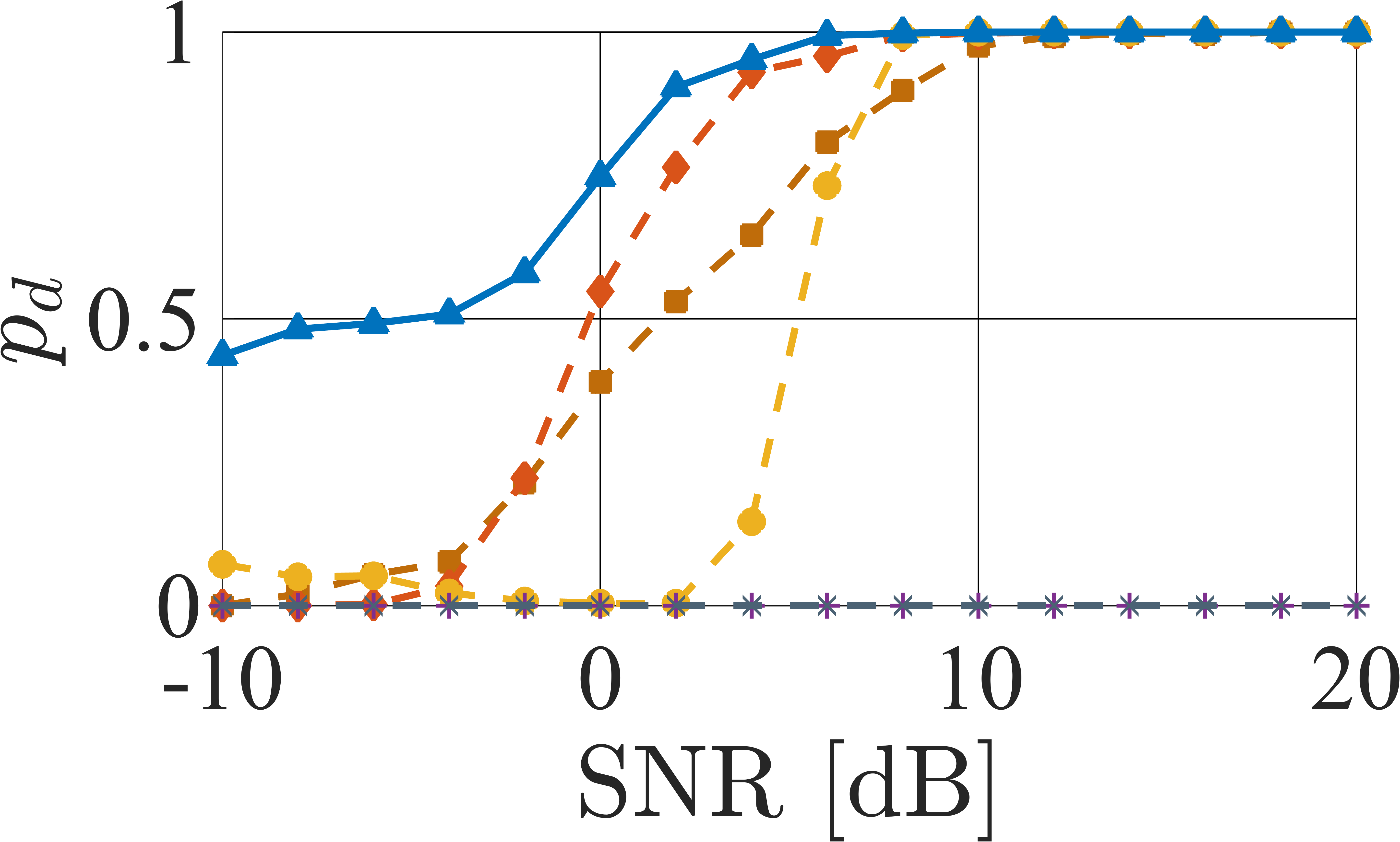}
        \caption{$M=2$, undamped}
        \label{fig: pd_SNR_2_undmp}
    \end{subfigure}
    \begin{subfigure}[b]{0.49\columnwidth}
        \centering
        \includegraphics[width=\linewidth] {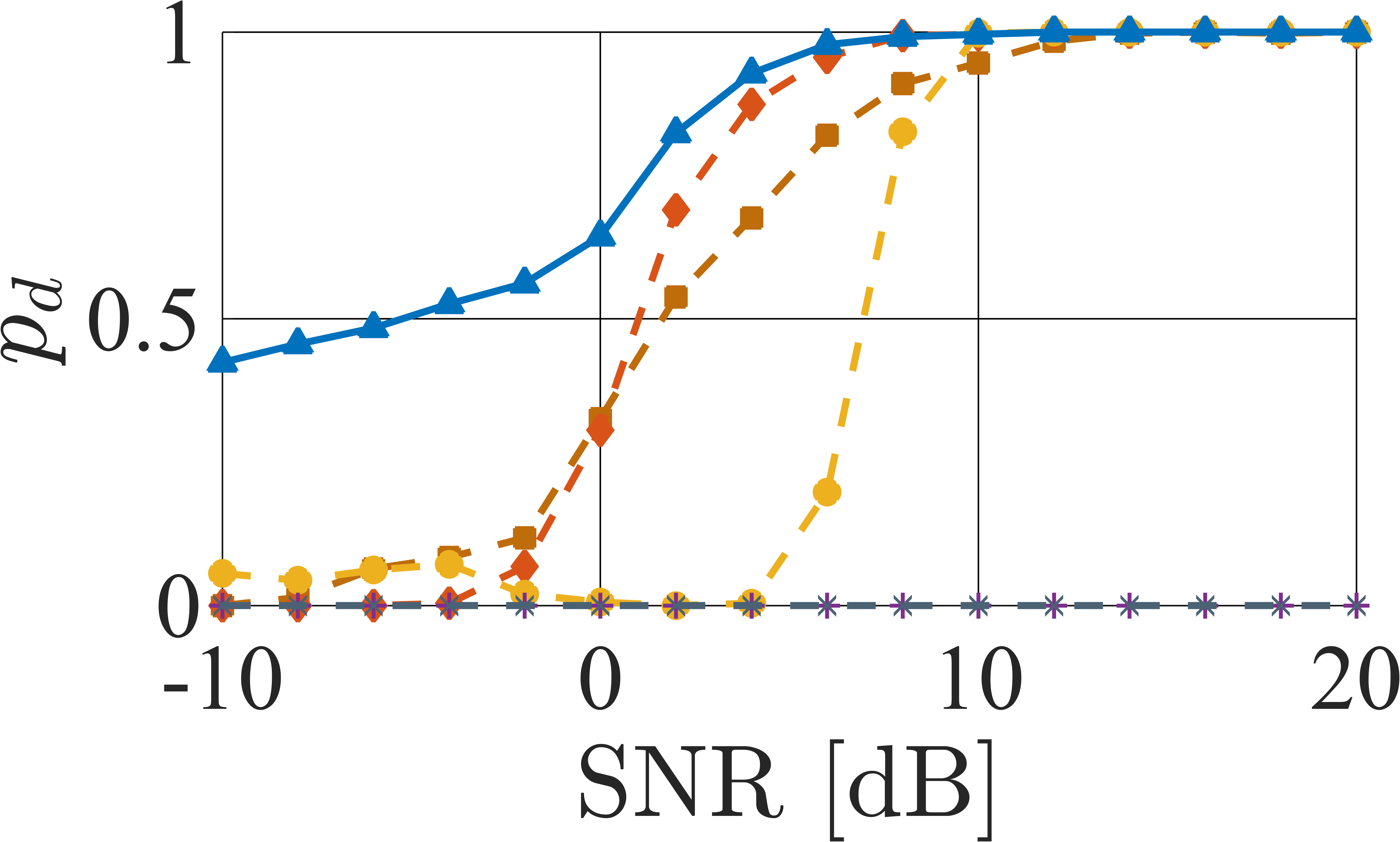}
        \caption{$M=2$, damped}
        \label{fig: pd_SNR_2_dmp}
    \end{subfigure}

    \hfill
\begin{subfigure}[b]{0.49\columnwidth}
        \centering
        \includegraphics[width=\linewidth]{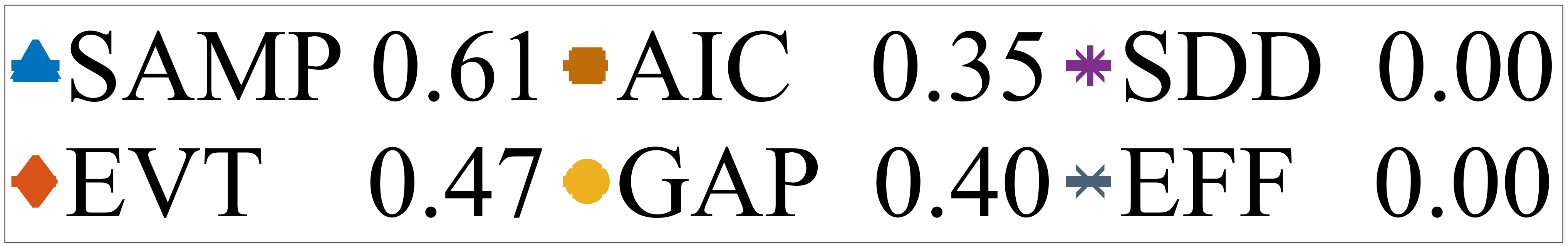}
    \end{subfigure}
    \begin{subfigure}[b]{0.49\columnwidth}
        \centering
        \includegraphics[width=\linewidth]{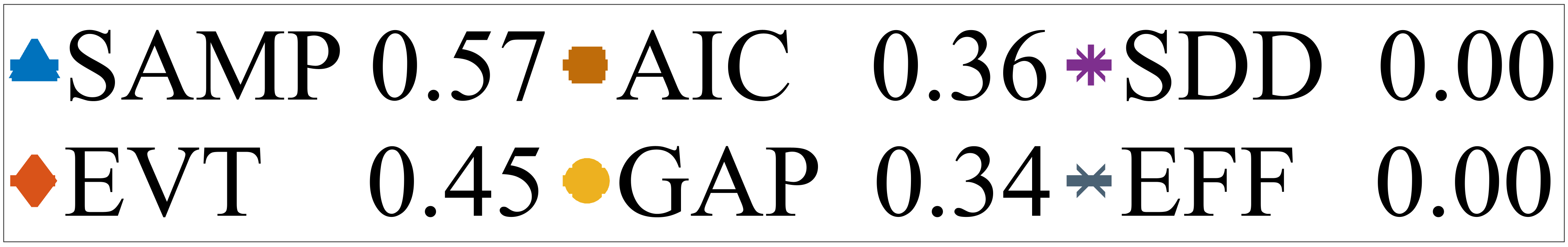}
    \end{subfigure}
    
    \begin{subfigure}[b]{0.49\columnwidth}
        \centering
        \includegraphics[width=\linewidth]{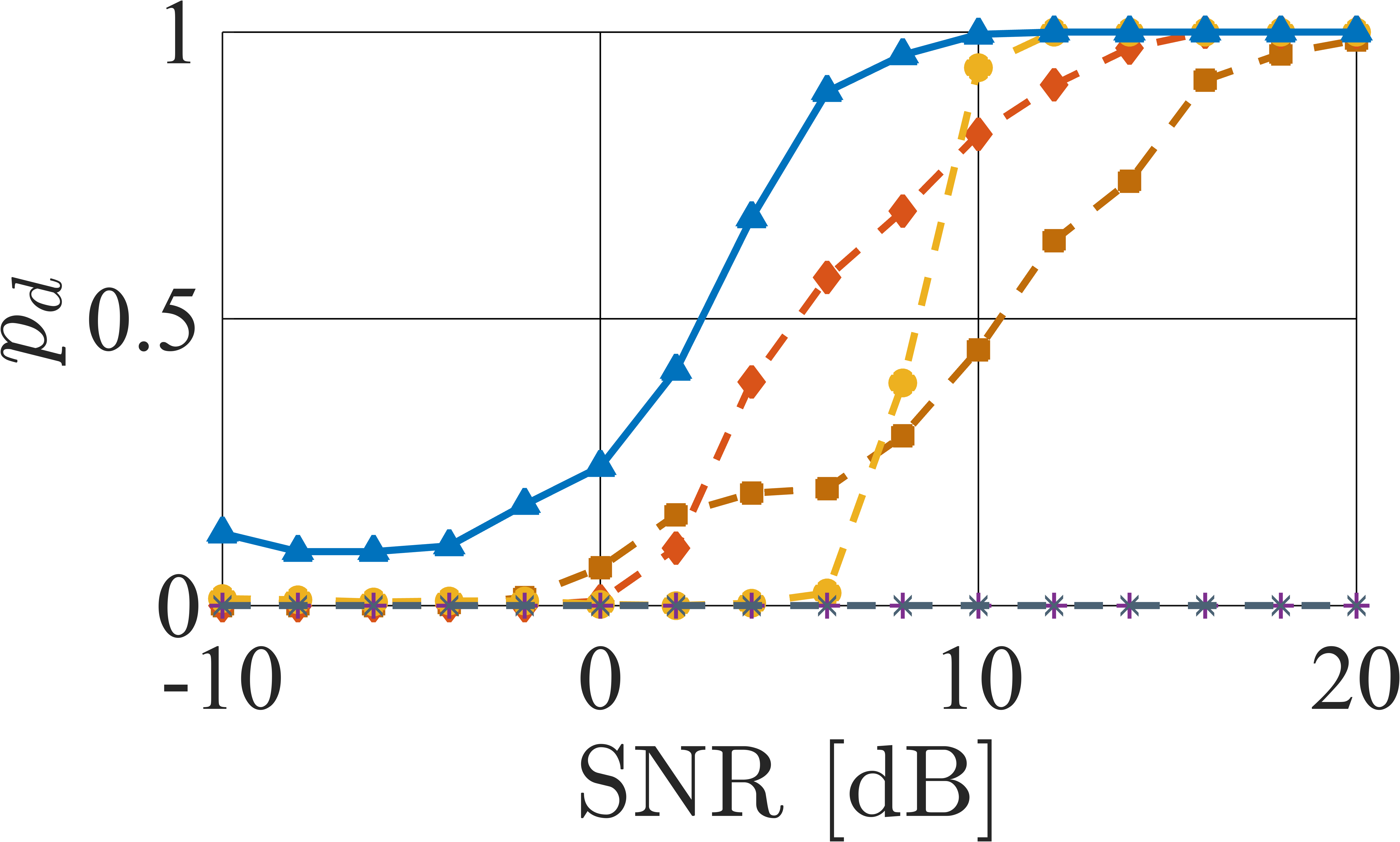}
        \caption{$M=4$, undamped}
        \label{fig: pd_SNR_4_undmp}
    \end{subfigure}
    \begin{subfigure}[b]{0.49\columnwidth}
        \centering
        \includegraphics[width=\linewidth] {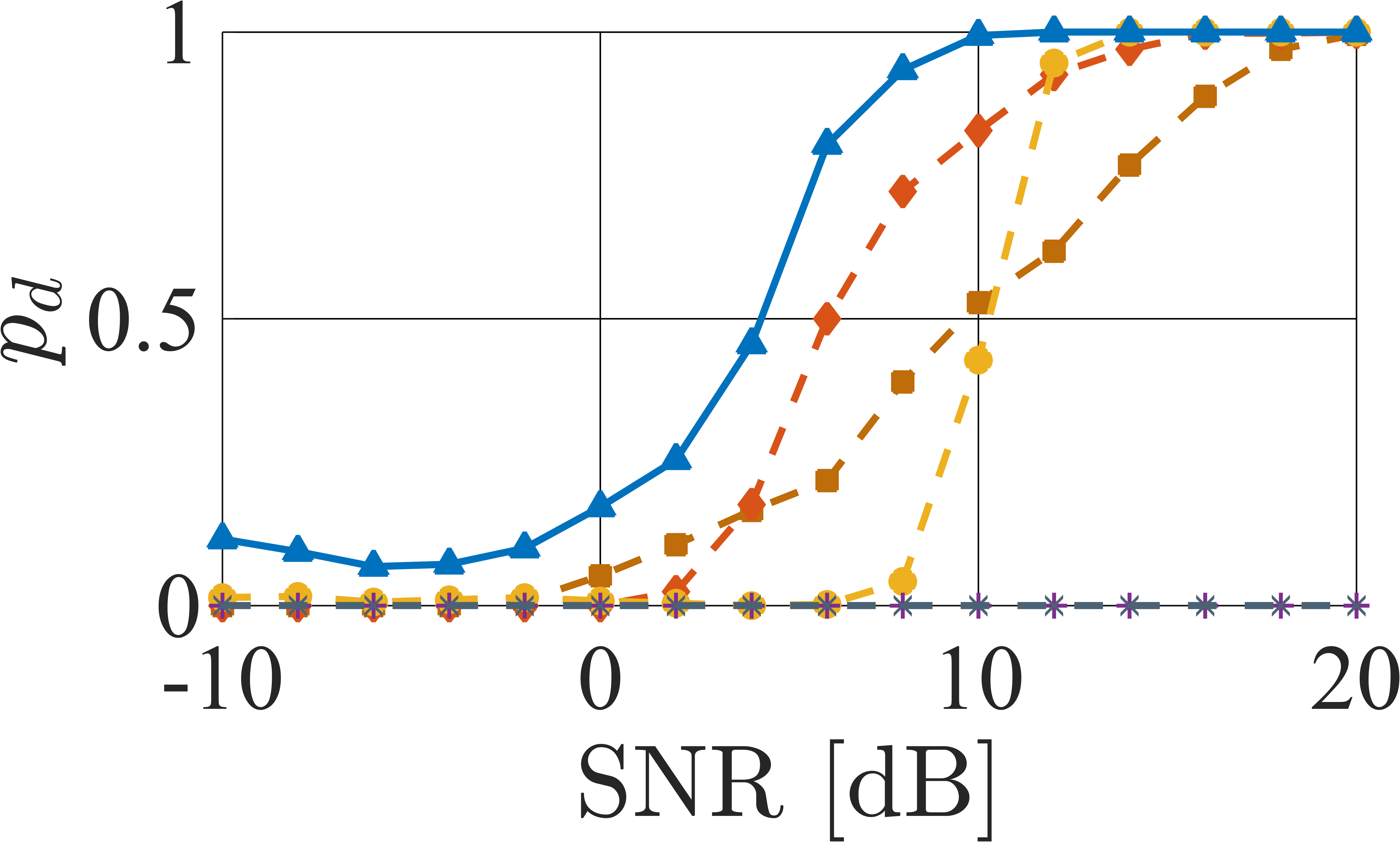}
        \caption{$M=4$, damped}
        \label{fig: pd_SNR_4_dmp}
    \end{subfigure}
        \caption{Probability of correct model order estimation versus $\text{SNR}_{\text{dB}}$. AUC scores are provided in the legend.}
    \label{fig: detection_probability_SNR}
\end{figure}

\begin{figure}[t]
    \centering
           \begin{subfigure}[b]{0.49\columnwidth}
        \centering
        \includegraphics[width=\linewidth]{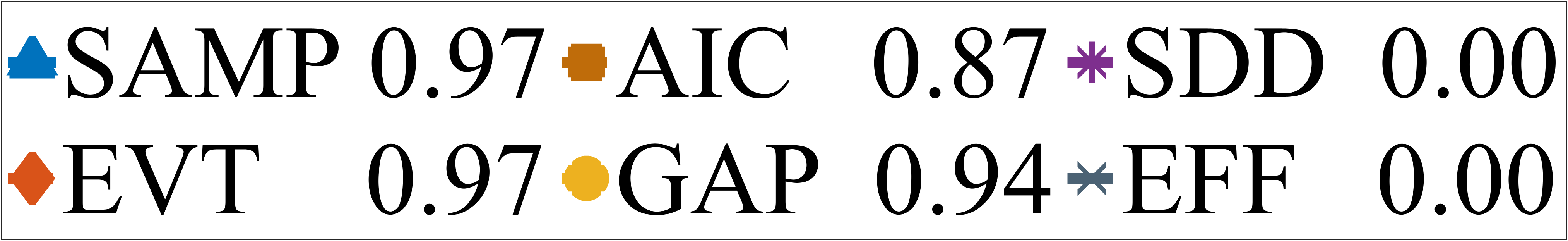}
    \end{subfigure}
    \begin{subfigure}[b]{0.49\columnwidth}
        \centering
        \includegraphics[width=\linewidth]{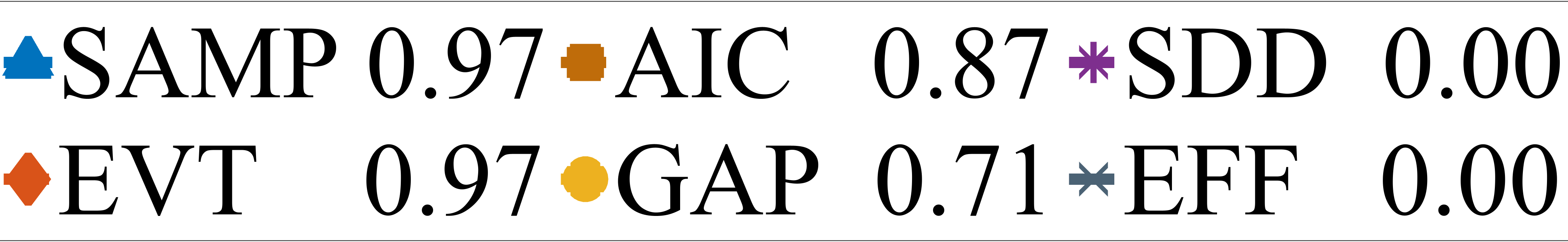}
    \end{subfigure}
    
    \begin{subfigure}[b]{0.49\columnwidth}
        \centering
        \includegraphics[width=\linewidth]{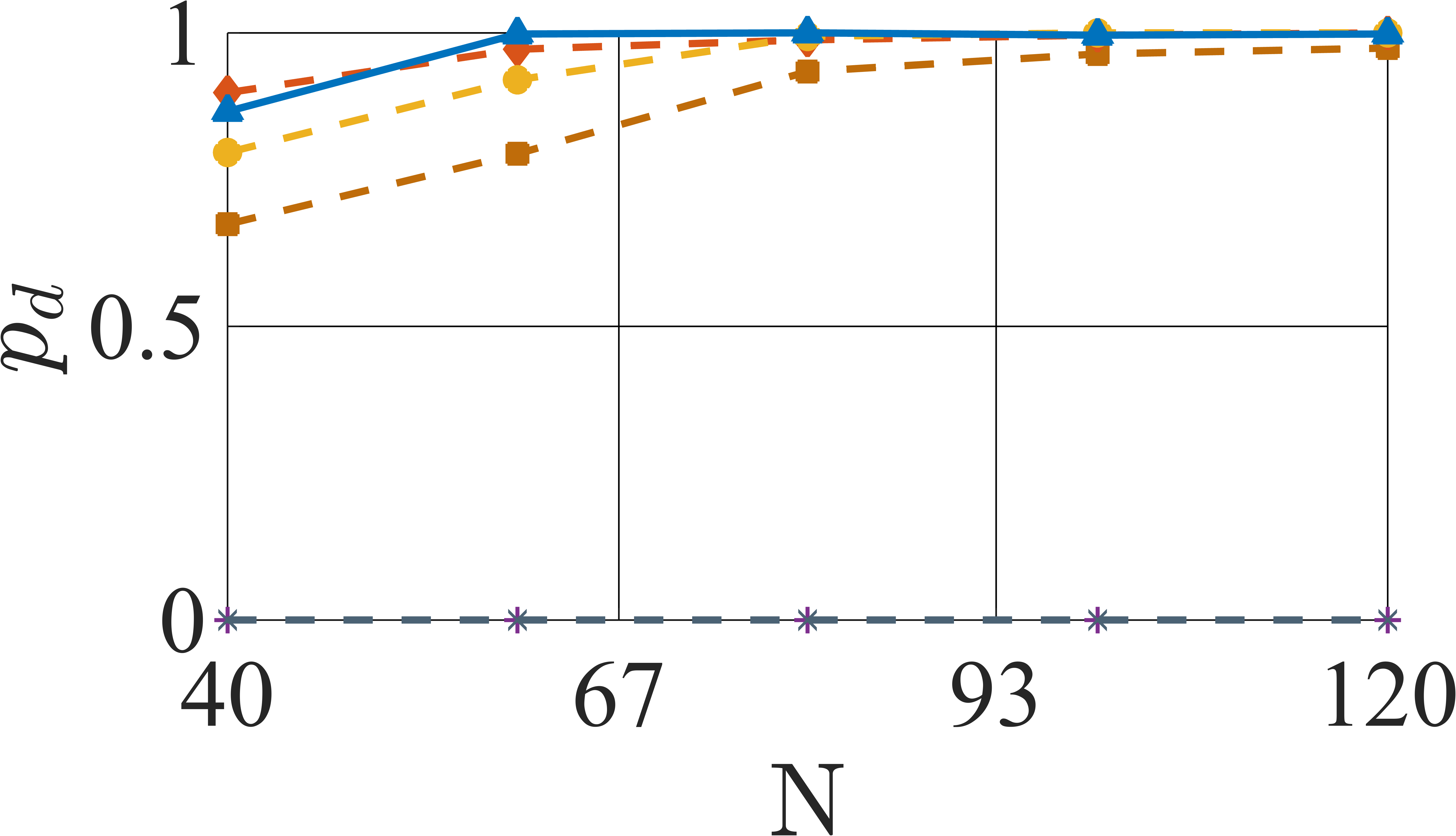}
        \caption{$M=2$, undamped}
        \label{fig: pd_SAMPLE_2_undmp}
    \end{subfigure}
    \begin{subfigure}[b]{0.49\columnwidth}
        \centering
        \includegraphics[width=\linewidth] {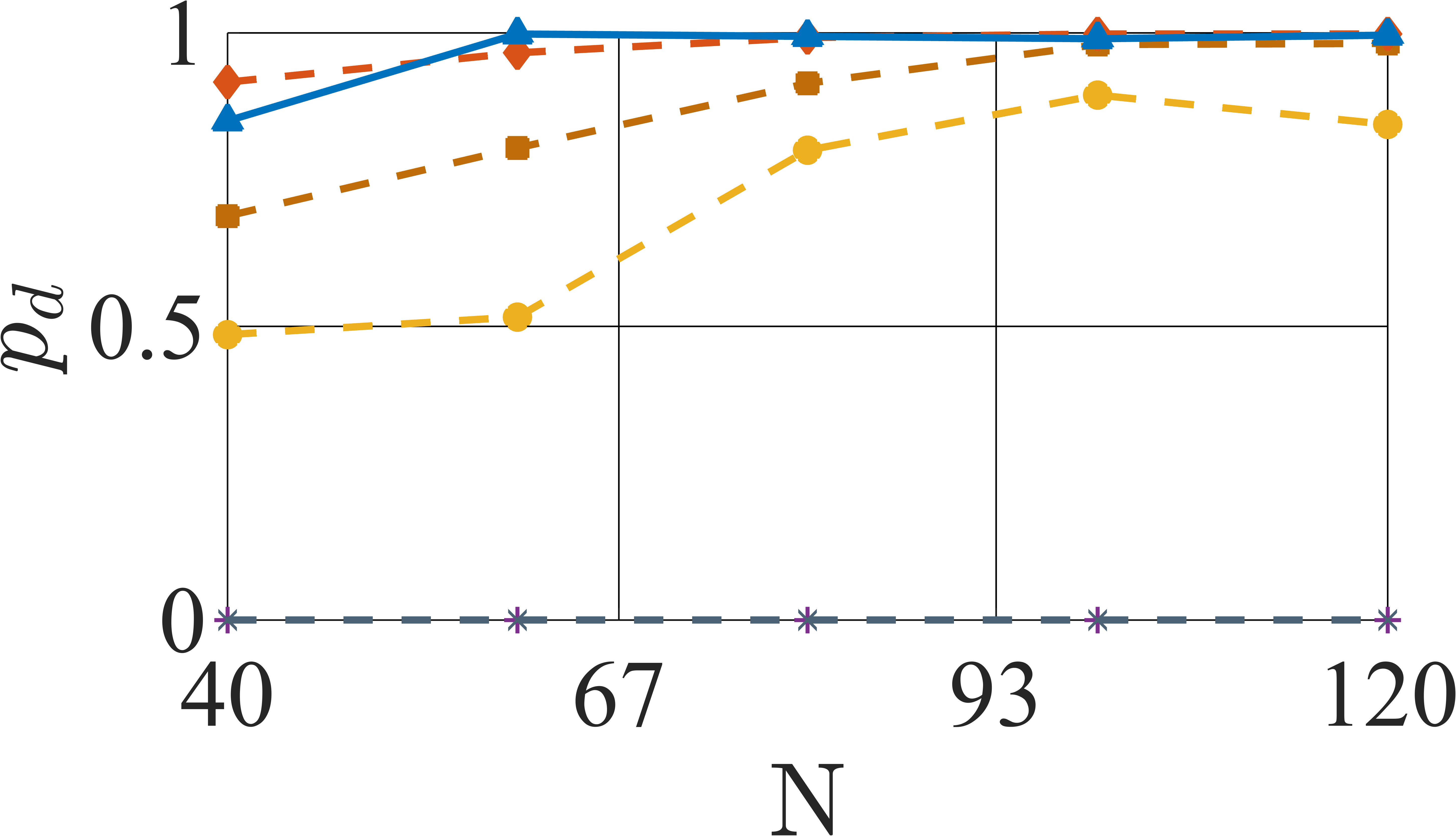}
        \caption{$M=2$, damped}
        \label{fig: pd_SAMPLE_2_dmp}
    \end{subfigure}

    \hfill
\begin{subfigure}[b]{0.49\columnwidth}
        \centering
        \includegraphics[width=\linewidth]{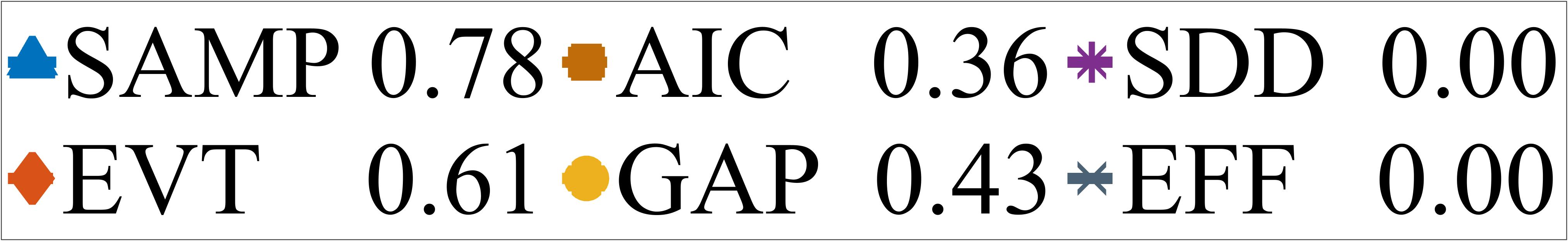}
    \end{subfigure}
    \begin{subfigure}[b]{0.49\columnwidth}
        \centering
        \includegraphics[width=\linewidth]{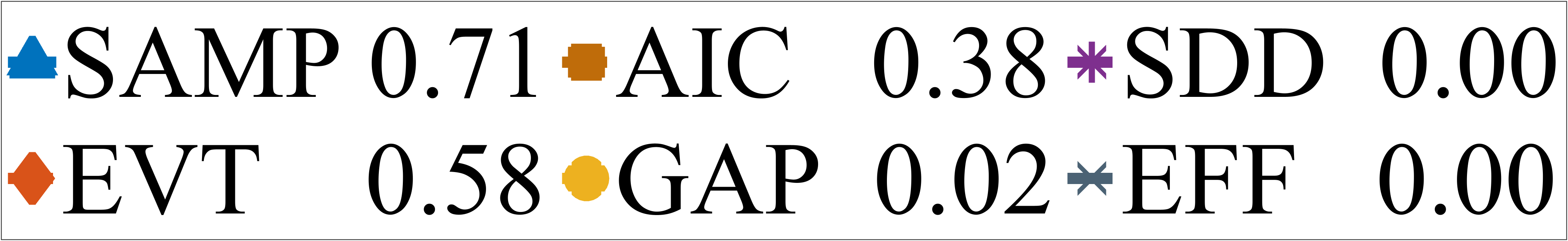}
    \end{subfigure}
    
    \begin{subfigure}[b]{0.49\columnwidth}
        \centering
        \includegraphics[width=\linewidth]{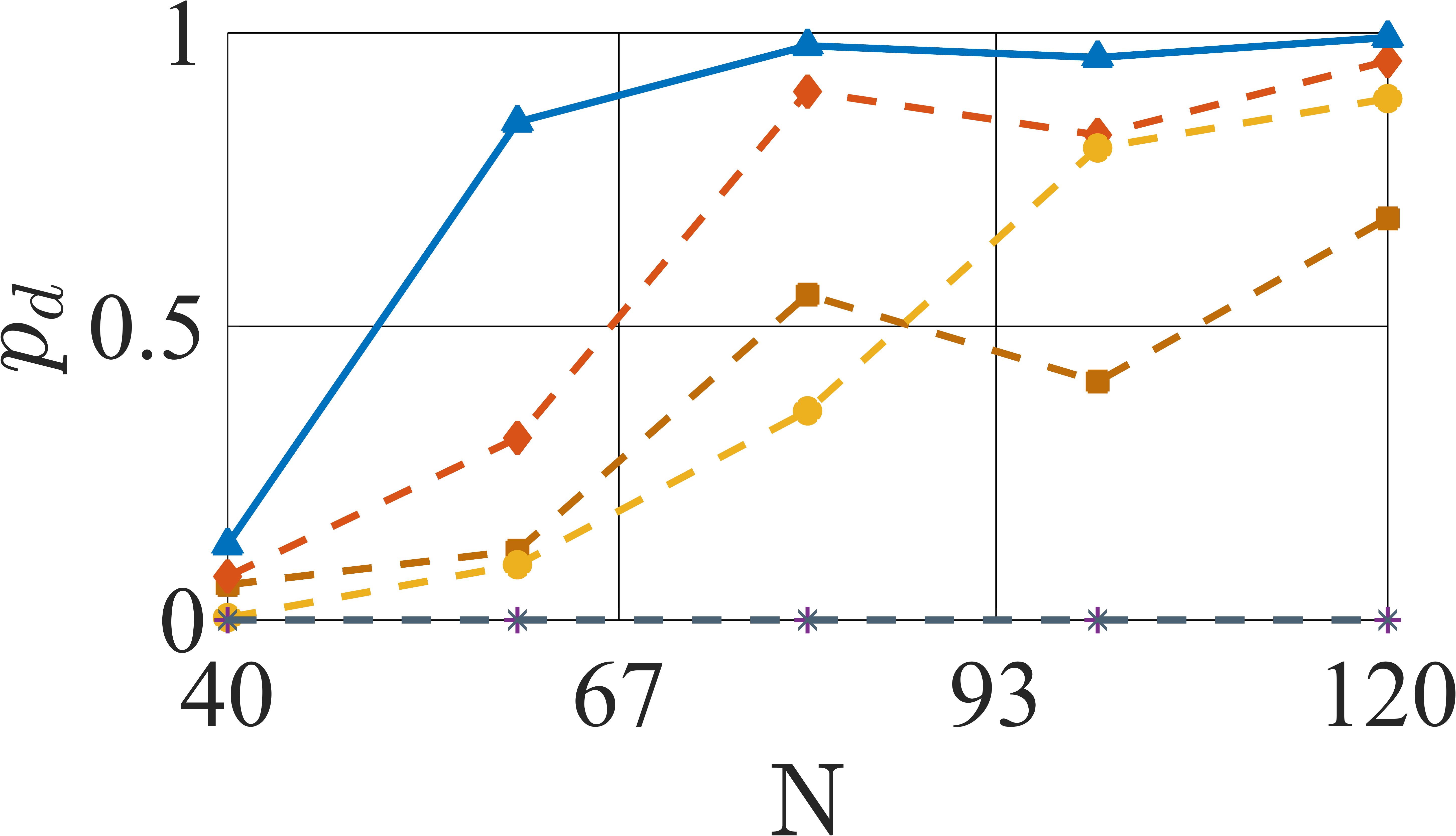}
        \caption{$M=4$, undamped}
        \label{fig: pd_SAMPLE_4_undmp}
    \end{subfigure}
    \begin{subfigure}[b]{0.49\columnwidth}
        \centering
        \includegraphics[width=\linewidth] {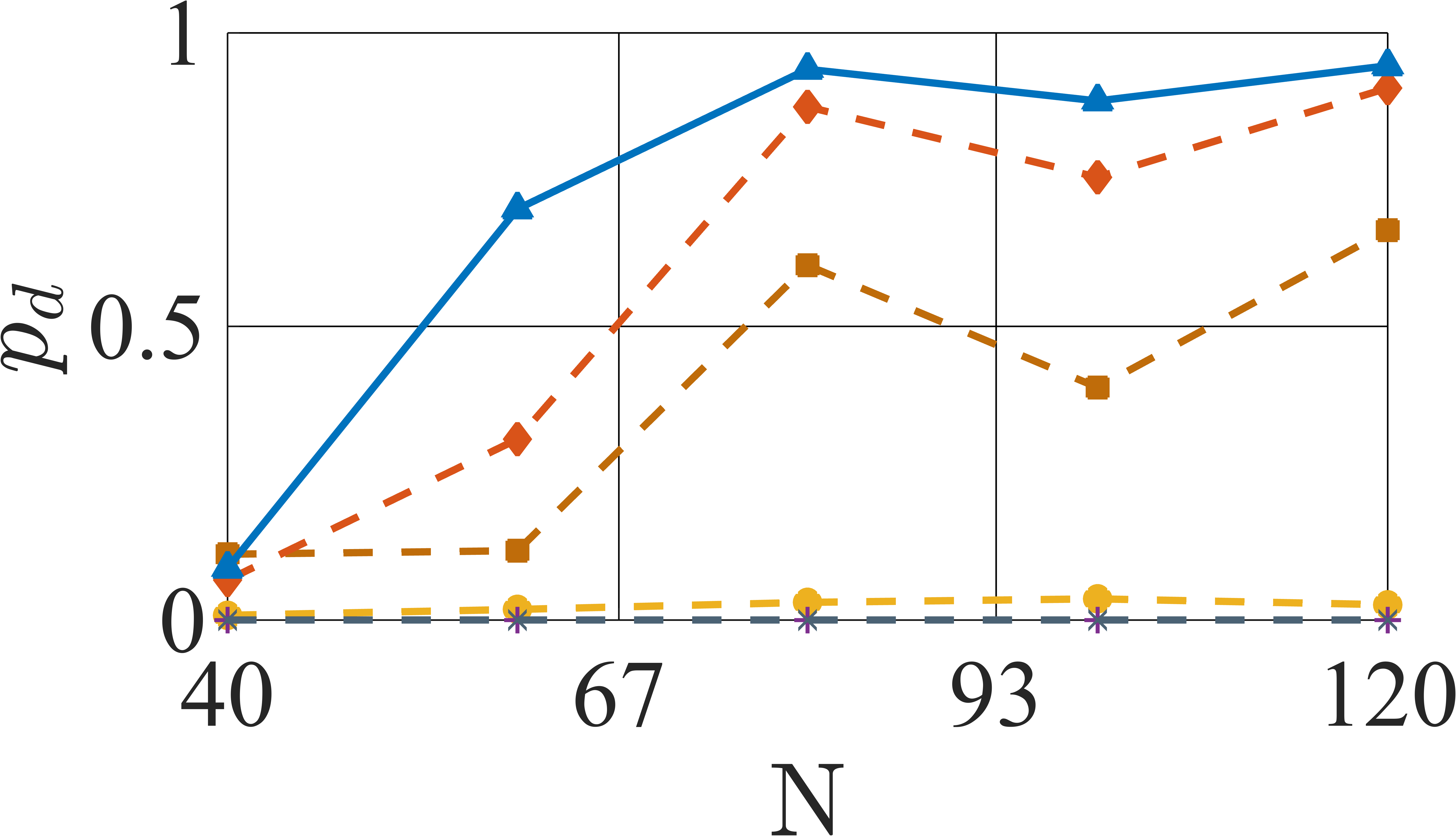}
        \caption{$M=4$, damped}
        \label{fig: pd_SAMPLE_4_dmp}
    \end{subfigure}
        \caption{Probability of correct model order estimation versus the number of samples $N$, at $\text{SNR}_{\text{dB}}=8$ dB. AUC scores are provided in the legend.}\vspace{-2mm}
    \label{fig: detection_probability_SAMPLE}
\end{figure}

\begin{figure}[t]
    \centering
           \begin{subfigure}[b]{0.49\columnwidth}
        \centering
        \includegraphics[width=\linewidth]{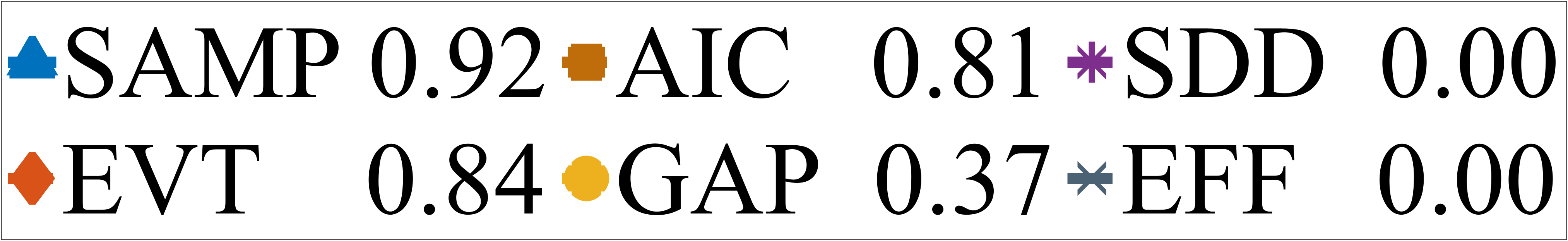}
    \end{subfigure}
    \begin{subfigure}[b]{0.49\columnwidth}
        \centering
        \includegraphics[width=\linewidth]{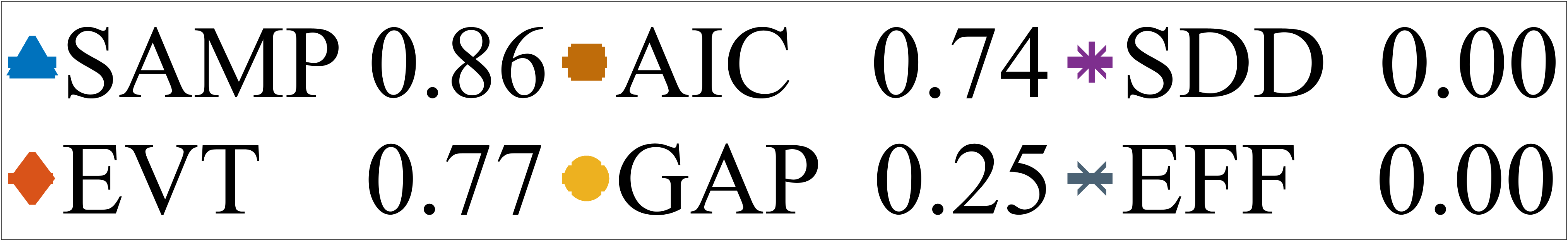}
    \end{subfigure}
    
    \begin{subfigure}[b]{0.49\columnwidth}
        \centering
        \includegraphics[width=\linewidth]{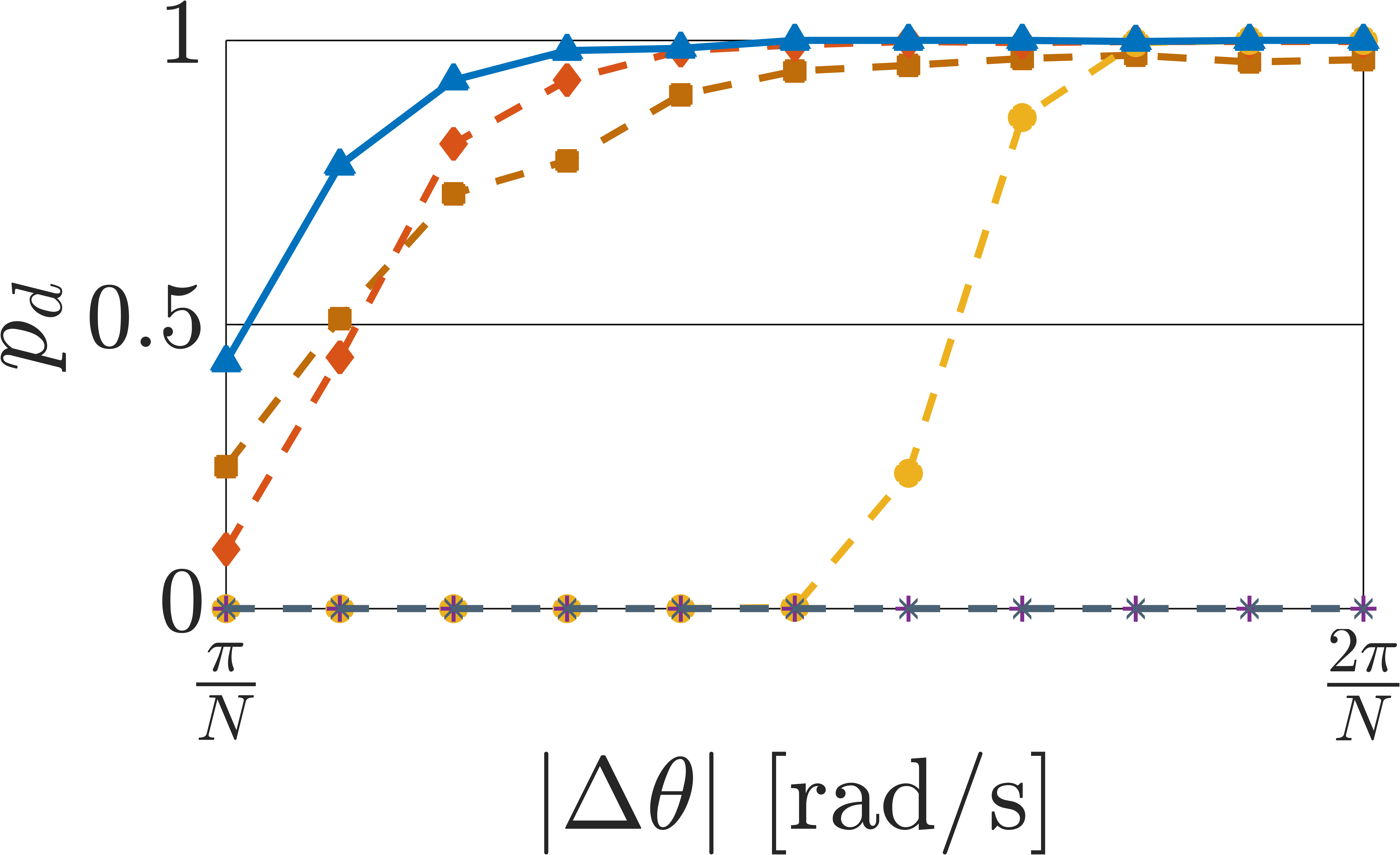}
        \caption{$M=2$, undamped}
        \label{fig: pd_DIFF_2_undmp}
    \end{subfigure}
    \begin{subfigure}[b]{0.49\columnwidth}
        \centering
        \includegraphics[width=\linewidth] {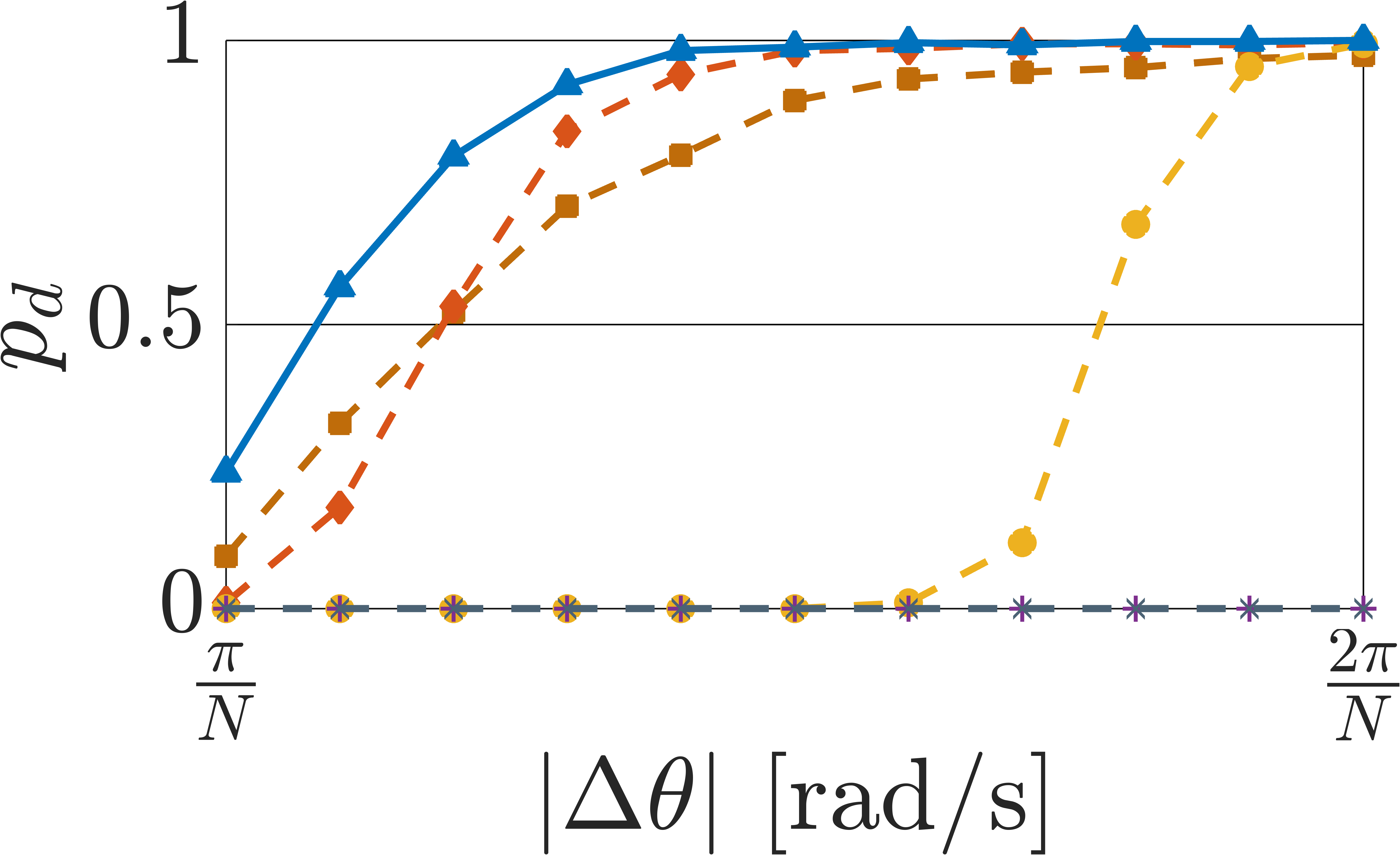}
        \caption{$M=2$, damped}
        \label{fig: pd_DIFF_2_dmp}
    \end{subfigure}
\hfill
    \begin{subfigure}[b]{0.49\columnwidth}
        \centering
        \includegraphics[width=\linewidth]{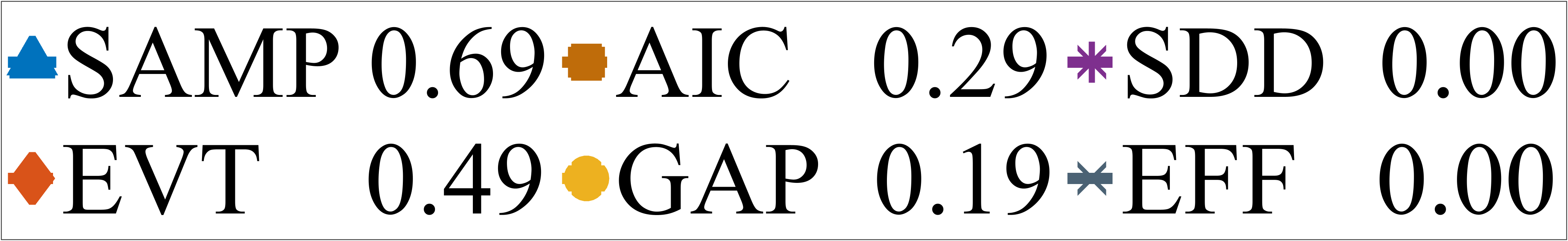}
    \end{subfigure}
    \begin{subfigure}[b]{0.49\columnwidth}
        \centering
        \includegraphics[width=\linewidth]{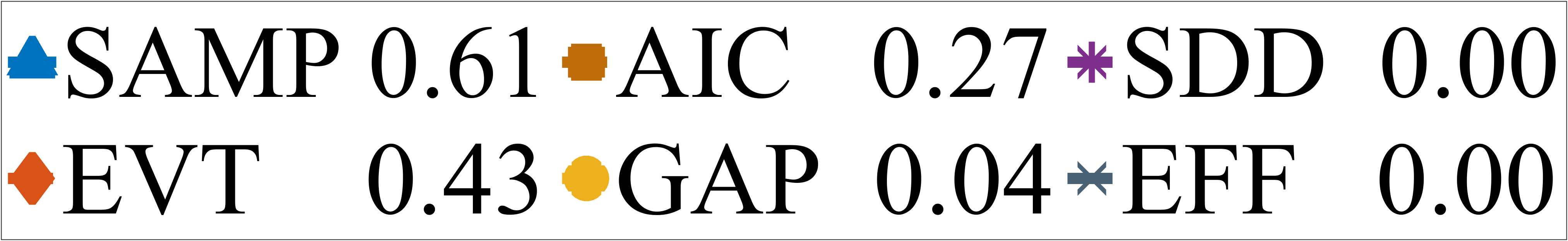}
    \end{subfigure}
    
    \begin{subfigure}[b]{0.49\columnwidth}
        \centering
        \includegraphics[width=\linewidth]{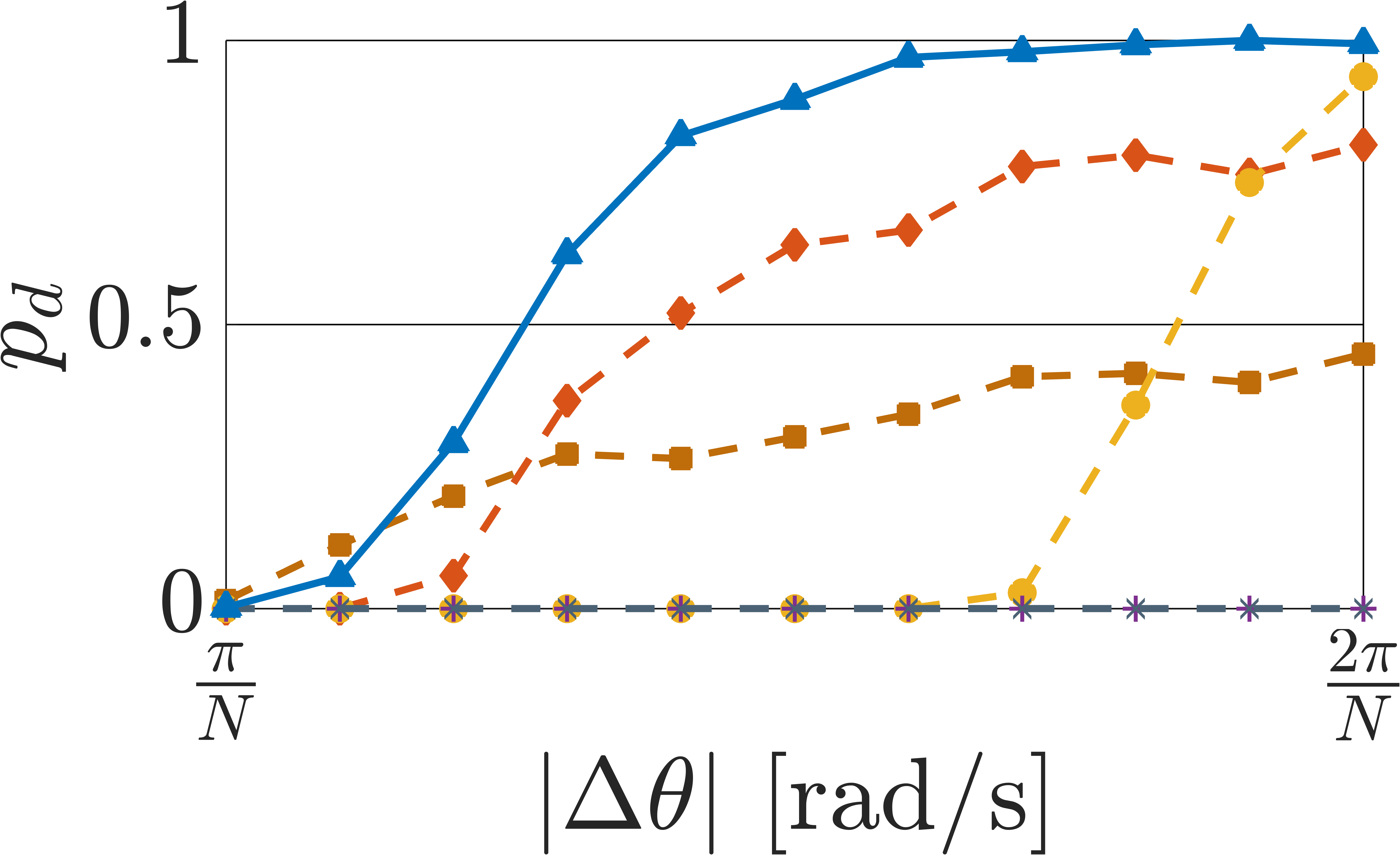}
        \caption{$M=4$, undamped}
        \label{fig: pd_DIFF_4_undmp}
    \end{subfigure}
    \begin{subfigure}[b]{0.49\columnwidth}
        \centering
        \includegraphics[width=\linewidth] {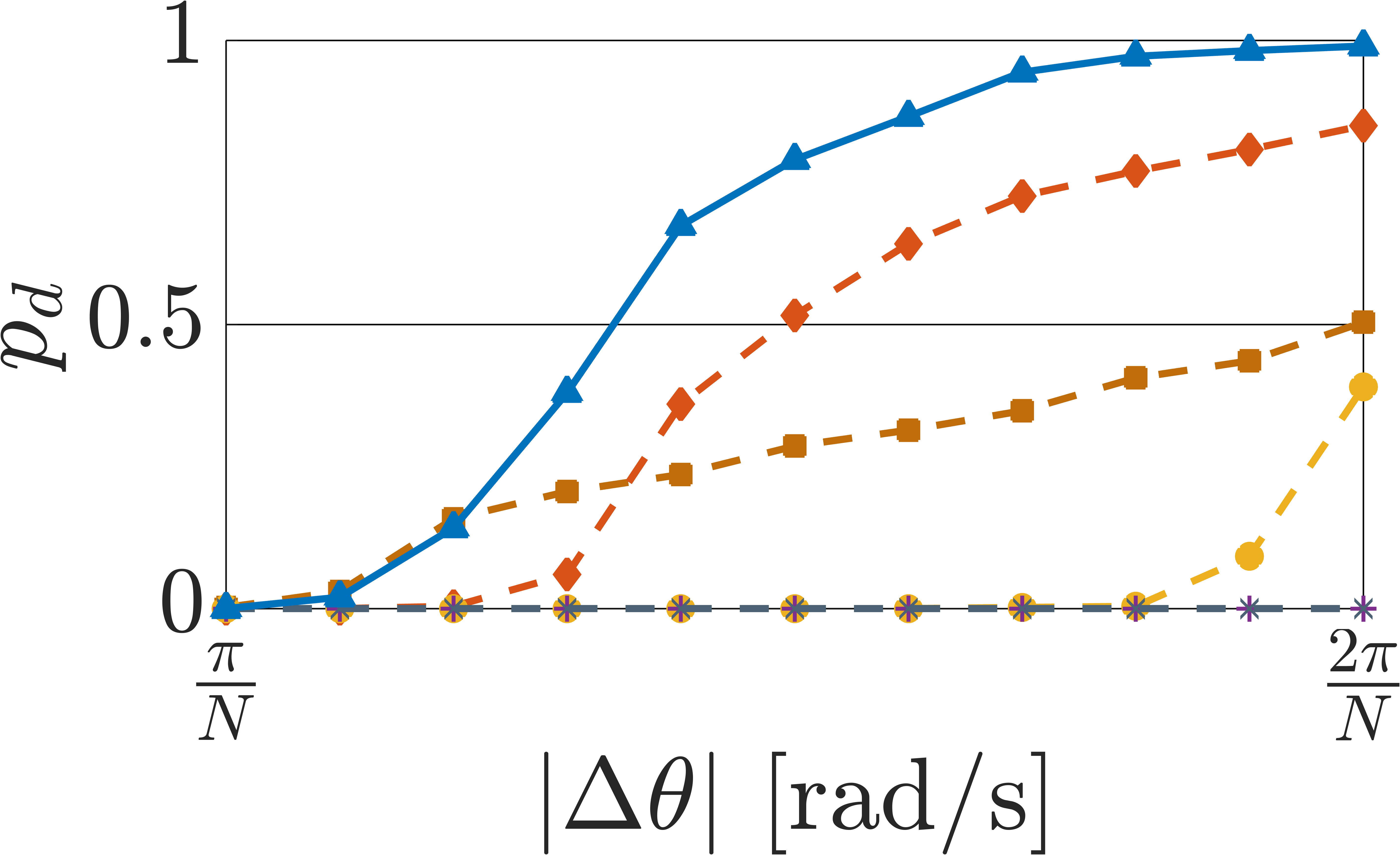}
        \caption{$M=4$, damped}
        \label{fig: pd_DIFF_4_dmp}
    \end{subfigure}
        \caption{Probability of correct model order estimation versus the frequency separation, at $\text{SNR}_{\text{dB}}=10$ dB. AUC scores are provided in the legend.}
    \label{fig: detection_probability_DIFF}
\end{figure}

\begin{table}[t]
    \centering
    \caption{AUC comparison of the different methods under normal and bi-normal noise distributions, for the undamped (undmp) and damped (dmp) cases.}
    \resizebox{0.8\columnwidth}{!}{%
    \begin{tabular}{l|cc|cc}
        \toprule
        \multirow{2}{*}{\textbf{Method}} &
        \multicolumn{2}{c|}{\textbf{Normal}} &
        \multicolumn{2}{c}{\textbf{Bi-normal}} \\
        \cmidrule(lr){2-3} \cmidrule(lr){4-5}
        & \textbf{undmp} & \textbf{dmp} & \textbf{undmp} & \textbf{dmp} \\
        \midrule
        SAMP & $\mathbf{0.82}$ & $\mathbf{0.80}$ & $\mathbf{0.78}$ & $\mathbf{0.76}$ \\
        EVT  & $0.65$ & $0.62$ & $0.56$ & $0.53$ \\
        AIC  & $0.60$ & $0.59$ & $0.50$ & $0.49$ \\
        GAP  & $0.51$ & $0.46$ & $0.52$ & $0.46$ \\
        SDD  & $0.00$ & $0.00$ & $0.00$ & $0.00$ \\
        EFF  & $0.00$ & $0.00$ & $0.00$ & $0.00$ \\
        \bottomrule
    \end{tabular}
    }
    \label{table:auc_comparison}
\end{table}\vspace{-2mm}

\subsection{Parameter estimation}

We evaluate the empirical bias and the Root Mean Squared Error (RMSE) of an estimate $\hat{\eta}$ of a parameter of interest $\eta$ which are defined by:
\begin{equation*}
    \text{bias}(\hat{\eta}) =  \frac{1}{N_{\text{exp}}} \sum_{i=1}^{N_{\text{exp}}} \hat{\eta}_i - \eta, \;
    \text{RMSE}(\hat{\eta}) = \sqrt{\frac{1}{N_{\text{exp}}}\sum_{i=1}^{N_{\text{exp}}} (\hat{\eta}_i - \eta)^2}
\end{equation*}\\
where $\{\hat{\eta}_i\}_{i=1}^{N_{\text{exp}}}$ is the estimate of $\eta$ in the $i$-th Monte-Carlo trial, and $N_{\text{exp}}=500$ is the number of independent Monte-Carlo trials. In cases of model order mismatch in the detection step, we assign a fixed penalty to the corresponding trial. This penalty is proportional to the range of $\eta$, and scaled according to the X-axis of the current evaluation. For example, if the RMSE is evaluated versus the number of samples $N$, the penalty is normalized by $N$. This ensures fair penalization across different ranges of the evaluation variable.

Fig. \ref{fig: freqs_SNR_RMSE} displays the average RMSE in dB of the estimates $\{\hat \theta_i\}_{i=1}^M$ versus the $\text{SNR}_{\text{dB}}$.
The complementary figures of the average empirical bias are presented in Fig. \ref{fig: freqs_SNR_BIAS} in Appendix \ref{app: Additional Numerical Results}.
We include the Cramer-Rao lower bound (CRB) as a reference for evaluating the average variance of $\{\hat \theta_i\}_{i=1}^M$ (see \cite{stoica2005spectral, lang1980frequency} for details). All the competing methods show nearly unbiased estimates for $\text{SNR}_{\text{dB}}$ values above a certain threshold, as shown in Fig. \ref{fig: freqs_SNR_BIAS} in Appendix \ref{app: Additional Numerical Results}. Therefore, following the approach presented in \cite{del1996comparison}, we use the CRB for an unbiased estimator as a reference for the competing methods. For unbiased estimators, the RMSE simplifies to the square root of the variance, since $\text{RMSE}^2(\hat{\theta}_1) = \text{var}(\hat{\theta}_1) + \text{bias}^2(\hat{\theta}_1)$.
{We note that while the CRB is used here as a lower bound on the RMSE for performance evaluation, it is a local bound and becomes tight only asymptotically. Global alternatives exist, such as the recently proposed Ziv–Zakai bounds for multiple-source parameter estimation \cite{zhang2022ziv}, which have been shown to remain tight over a broad range of SNR values.}

Figures \ref{fig:2_freqs_RMSE_undamped_SNR}-\ref{fig:4_freqs_RMSE_damped_SNR} show that the SAMP method achieves the lowest or comparable RMSE for SNR below a threshold. 
For $M=2$, SAMP and EVT are the two best performers, with SAMP showing slightly better results for SNR below $6$ dB.
For $M=4$, the separation between methods is larger: SAMP consistently yields the lowest RMSE up to $12$ dB, whereas EVT, AIC and GAP exhibit noticeably higher errors, especially at moderate SNR values. 
The SDD and EFF exhibit significantly higher RMSE across all regimes.
Overall, damping in Figs. \ref{fig:2_freqs_RMSE_damped_SNR} and \ref{fig:4_freqs_RMSE_damped_SNR}, degrades performance but does not change the ranking; SAMP is the most accurate for $M=4$ and is comparable to EVT for $M=2$.

The complementary results of the average RMSE in dB and the absolute value of the average empirical bias of the estimates $\{\hat \theta_i\}_{i=1}^M$ versus the number of samples $N$, and the frequency separation show similar trends and are presented in Fig. \ref{fig: freqs_SAMPLES_RMSE_BIAS} and Fig. \ref{fig: freqs_DIFF_RMSE_BIAS}, respectively, in Appendix \ref{app: Additional Numerical Results}.

%%----------------------------------------------------------------------%%

% --------------------------- RMSE and BIAS plots ---------------------------------

\begin{figure}[t]
    \centering
    \begin{subfigure}[b]{0.49\columnwidth}
        \centering
        \includegraphics[width=\linewidth]{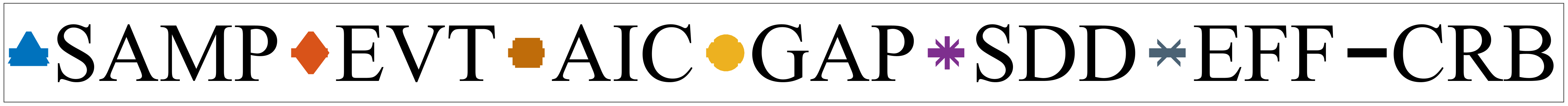}
    \end{subfigure}
    \begin{subfigure}[b]{0.49\columnwidth}
        \centering
        \includegraphics[width=\linewidth]{New_Figures/legend_fig.png}
    \end{subfigure}
    
    \begin{subfigure}[b]{0.49\columnwidth}
        \centering
        \includegraphics[width=\linewidth]{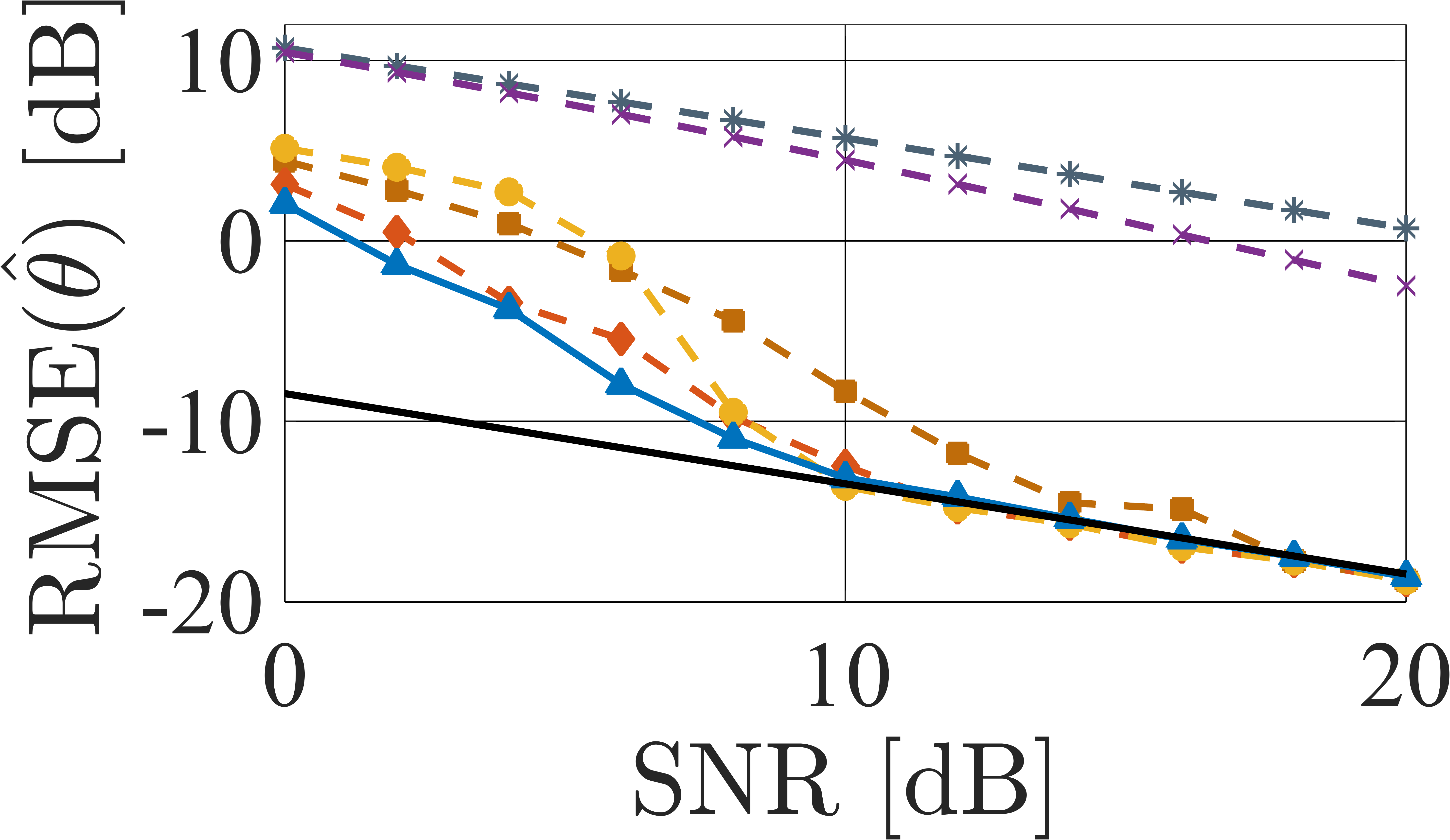}
        \caption{$M=2$, undamped}
        \label{fig:2_freqs_RMSE_undamped_SNR}
    \end{subfigure}
    \begin{subfigure}[b]{0.49\columnwidth}
        \centering
        \includegraphics[width=\linewidth] {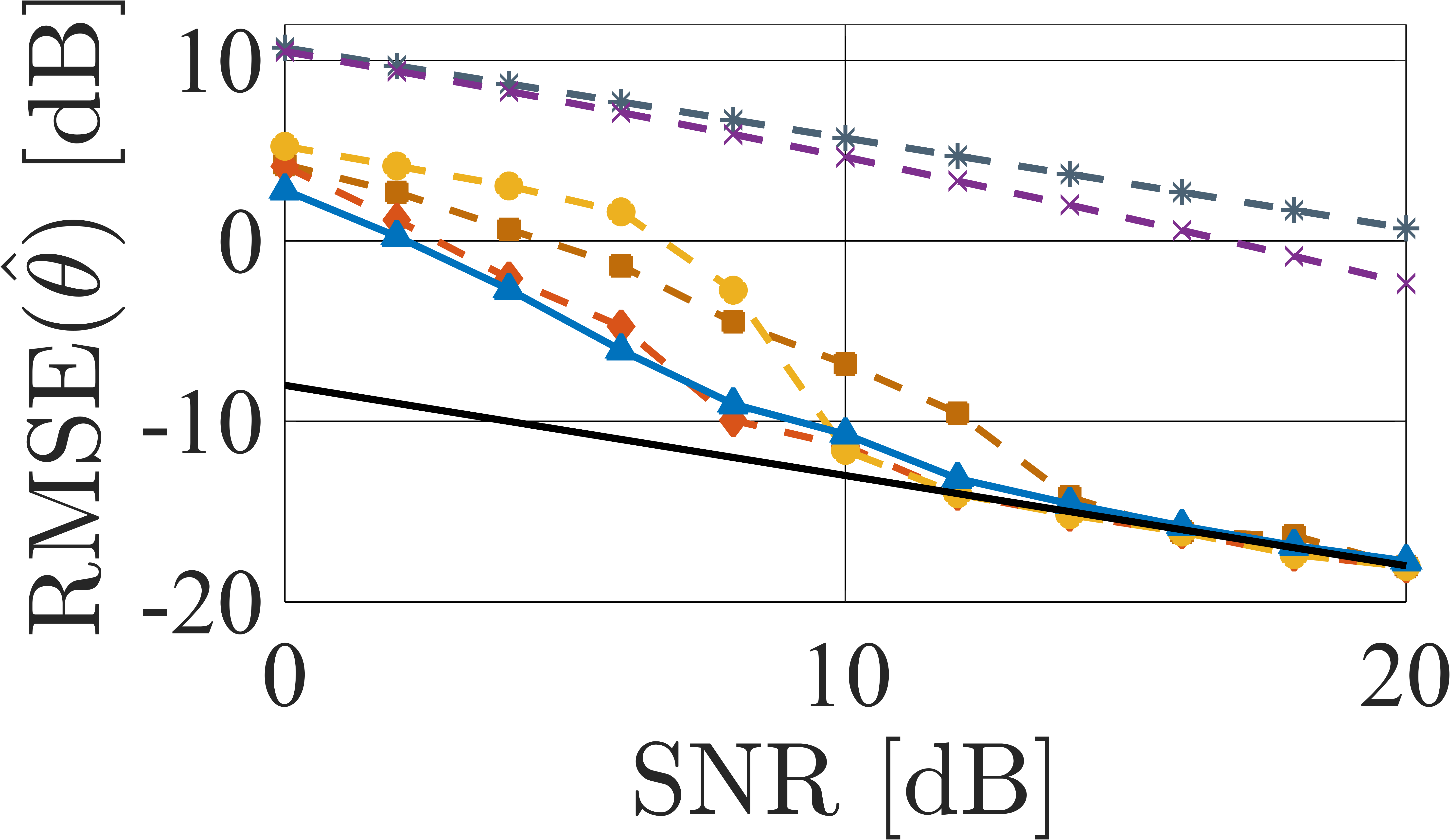}
        \caption{$M=2$, damped}
        \label{fig:2_freqs_RMSE_damped_SNR}
    \end{subfigure}

    \hfill
\begin{subfigure}[b]{0.49\columnwidth}
        \centering
        \includegraphics[width=\linewidth]{New_Figures/legend_fig.png}
    \end{subfigure}
    \begin{subfigure}[b]{0.49\columnwidth}
        \centering
        \includegraphics[width=\linewidth]{New_Figures/legend_fig.png}
    \end{subfigure}
    
    \begin{subfigure}[b]{0.49\columnwidth}
        \centering
        \includegraphics[width=\linewidth]{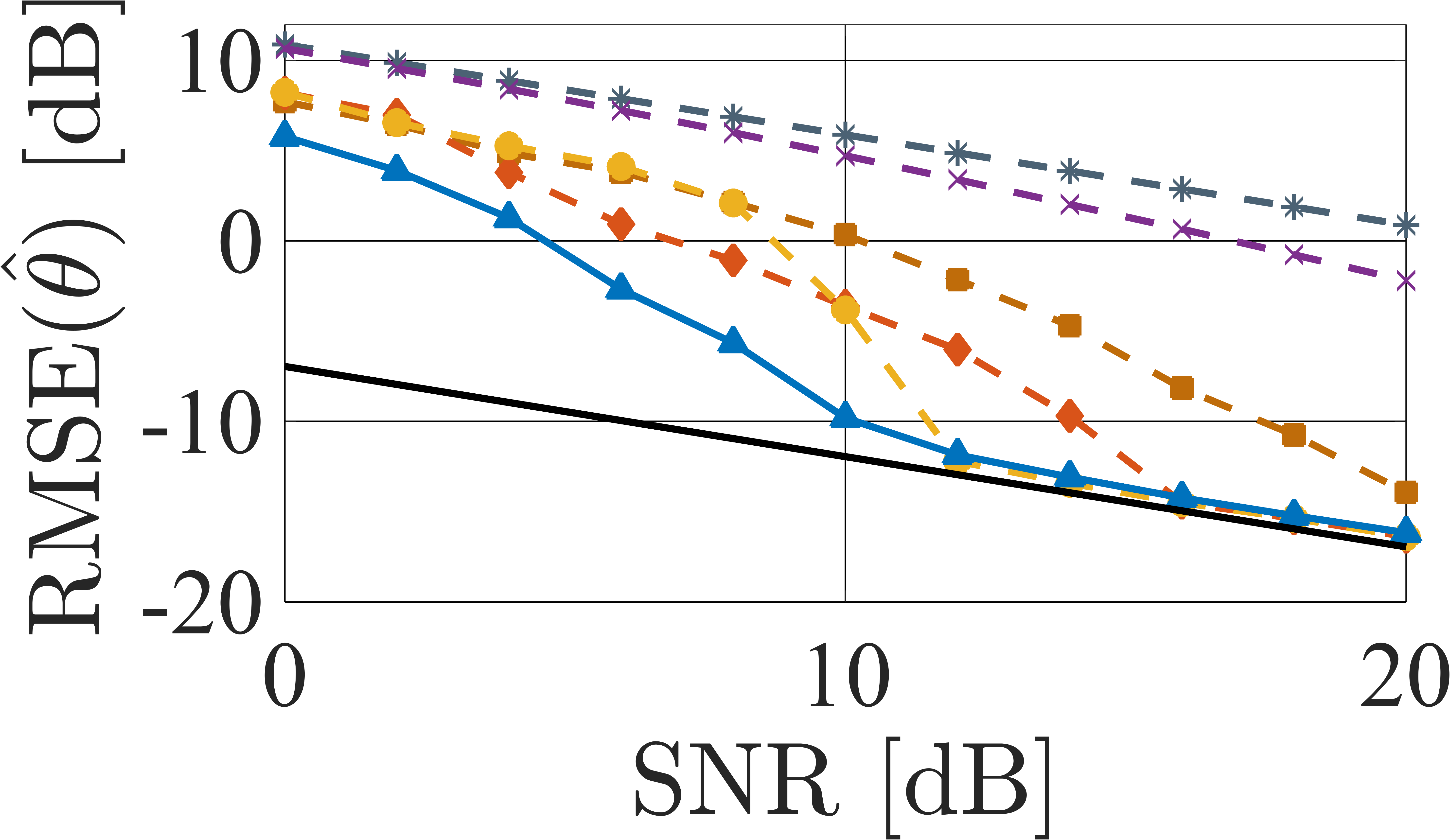}
        \caption{$M=4$, undamped}
        \label{fig:4_freqs_RMSE_undamped_SNR}
    \end{subfigure}
    \begin{subfigure}[b]{0.49\columnwidth}
        \centering
        \includegraphics[width=\linewidth] {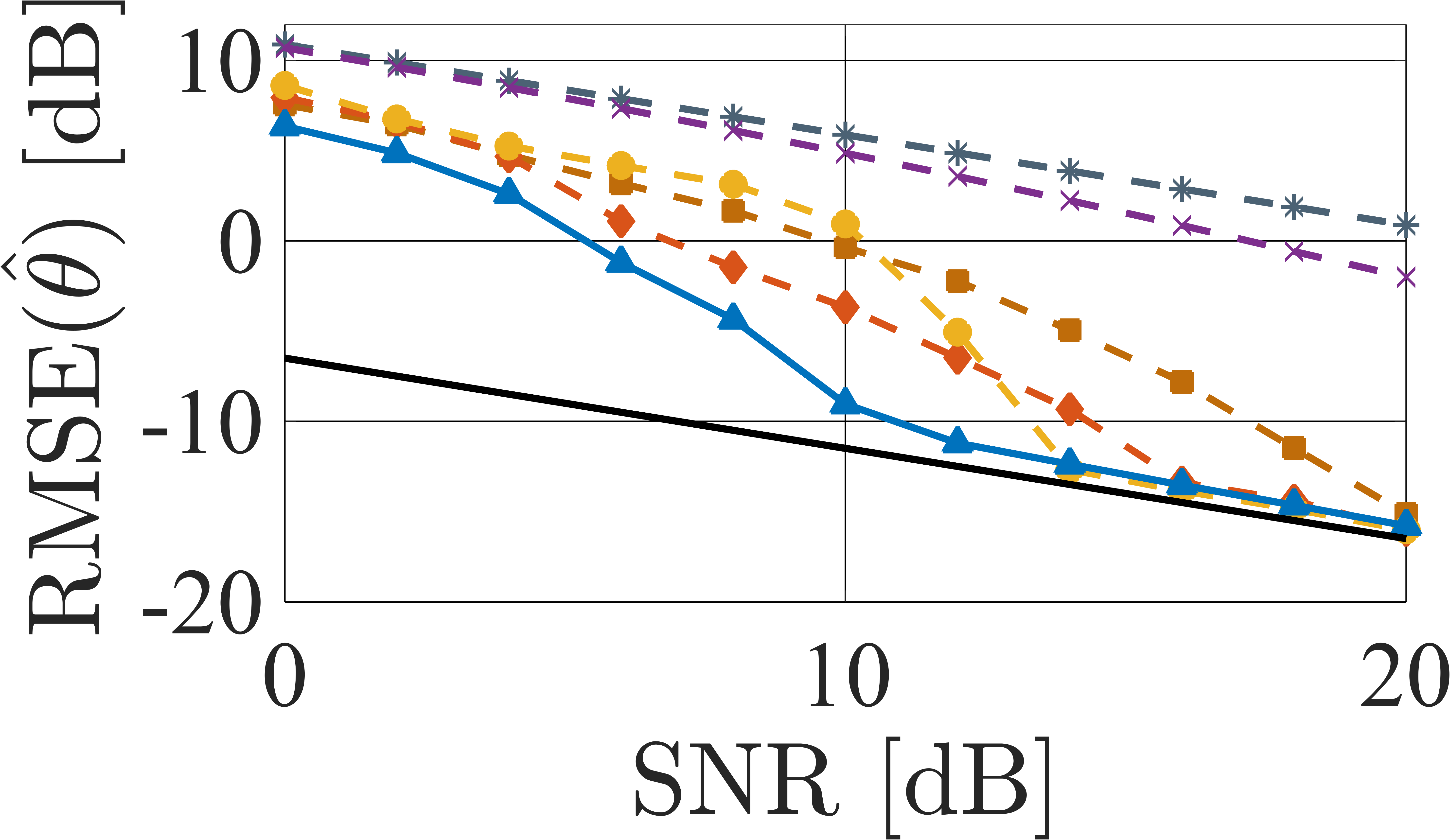}
        \caption{$M=4$, damped}
        \label{fig:4_freqs_RMSE_damped_SNR}
    \end{subfigure}
    \caption{Average RMSE of $\{ \hat\theta_i \}_{i=1}^M$ versus the $\text{SNR}_{\text{dB}}$.}\vspace{-3mm}
    \label{fig: freqs_SNR_RMSE}
\end{figure}

\vspace{-3mm}
\subsection{Computational Load} 
To compare the computational cost of the proposed SAMP algorithm with the baseline methods, we present the computation time as a function of the signal length in Fig. \ref{fig:MDL_AIC_MF_come_time_compare}. The results show that SAMP achieves lower computational time than the EVT and AIC methods, but higher than the GAP and SDD methods. 
{Specifically, SAMP requires, on average, approximately $2.4$ times less computation time than the ITE methods, but approximately $1.9$ times more computation time than the spectral thresholding methods.
}
This behavior is expected: ITE methods require evaluating the likelihood function for all candidate model orders $1 \leq \widehat{M} \leq L$, while spectral thresholding typically involves only $\mathcal{O}(L)$ operations, which is fewer than those required by SAMP, as detailed in Section \ref{sebsec: comp complex}. {Practically, this implies that, for a fixed computational budget, SAMP can process more samples, compared to ITE methods, and thus yields more accurate estimates.}

%% ------ Computational time plots

\begin{figure}[t]
    \centering
    \begin{subfigure}[b]{0.49\columnwidth}
        \centering
        \includegraphics[width=\linewidth]{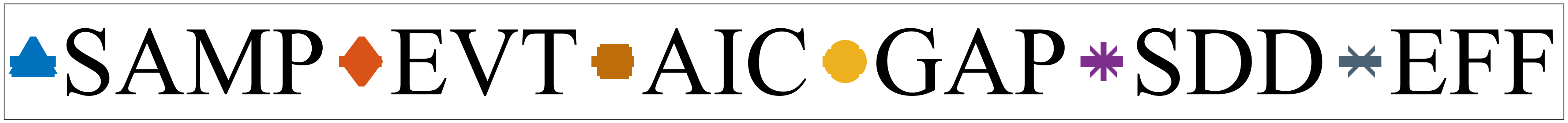}
    \end{subfigure}
    \begin{subfigure}[b]{0.49\columnwidth}
        \centering
        \includegraphics[width=\linewidth]{New_Figures/legend_fig_noCRB.png}
    \end{subfigure}
    
    \begin{subfigure}[b]{0.49\columnwidth}
        \centering
    \includegraphics[width=\linewidth]{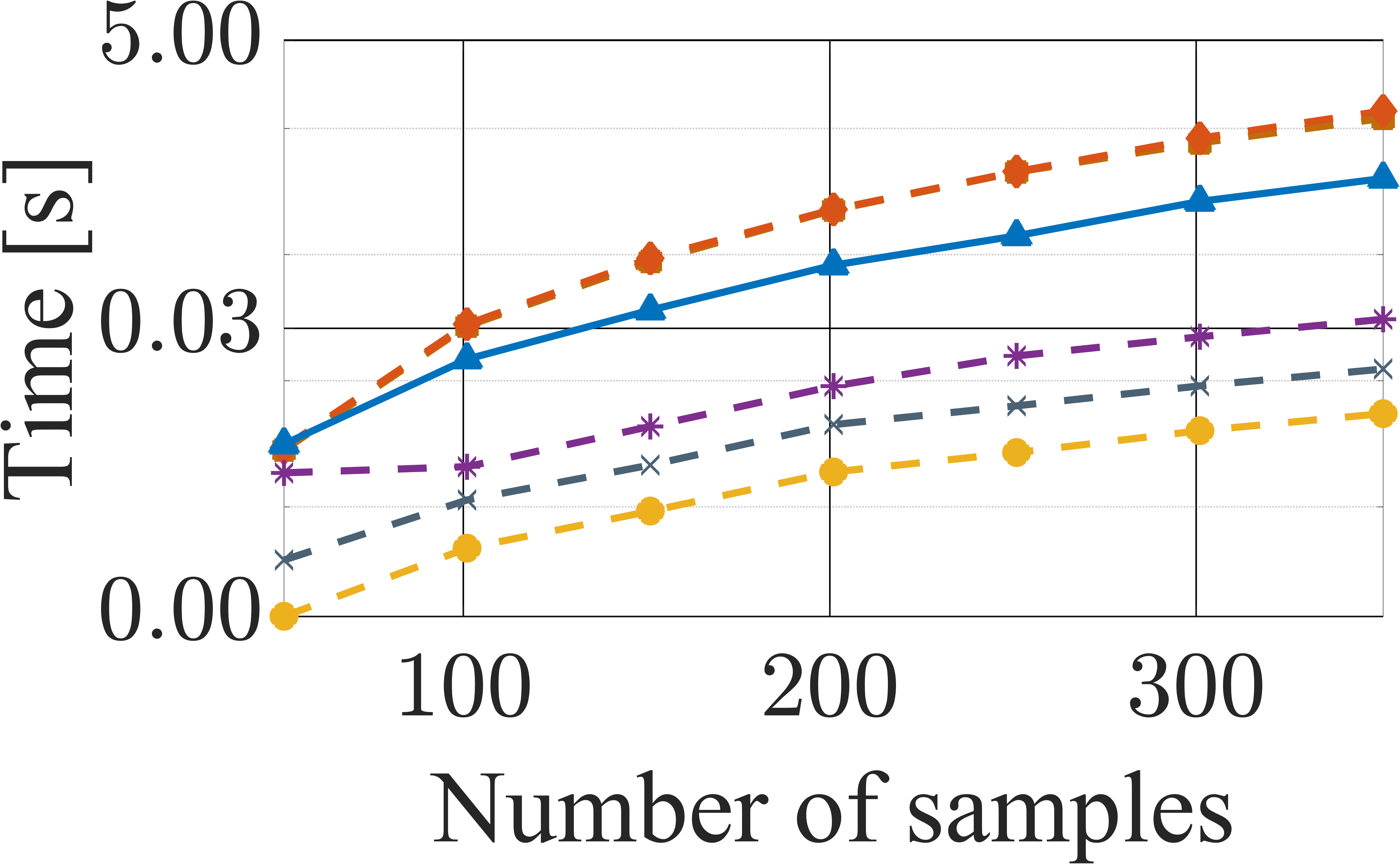}
    \subcaption{$M=2$}\label{fig:MDL_AIC_MF_come_time_compare_undamped}
    \end{subfigure}
    \begin{subfigure}[b]{0.49\columnwidth}
        \centering
        \includegraphics[width=\linewidth]{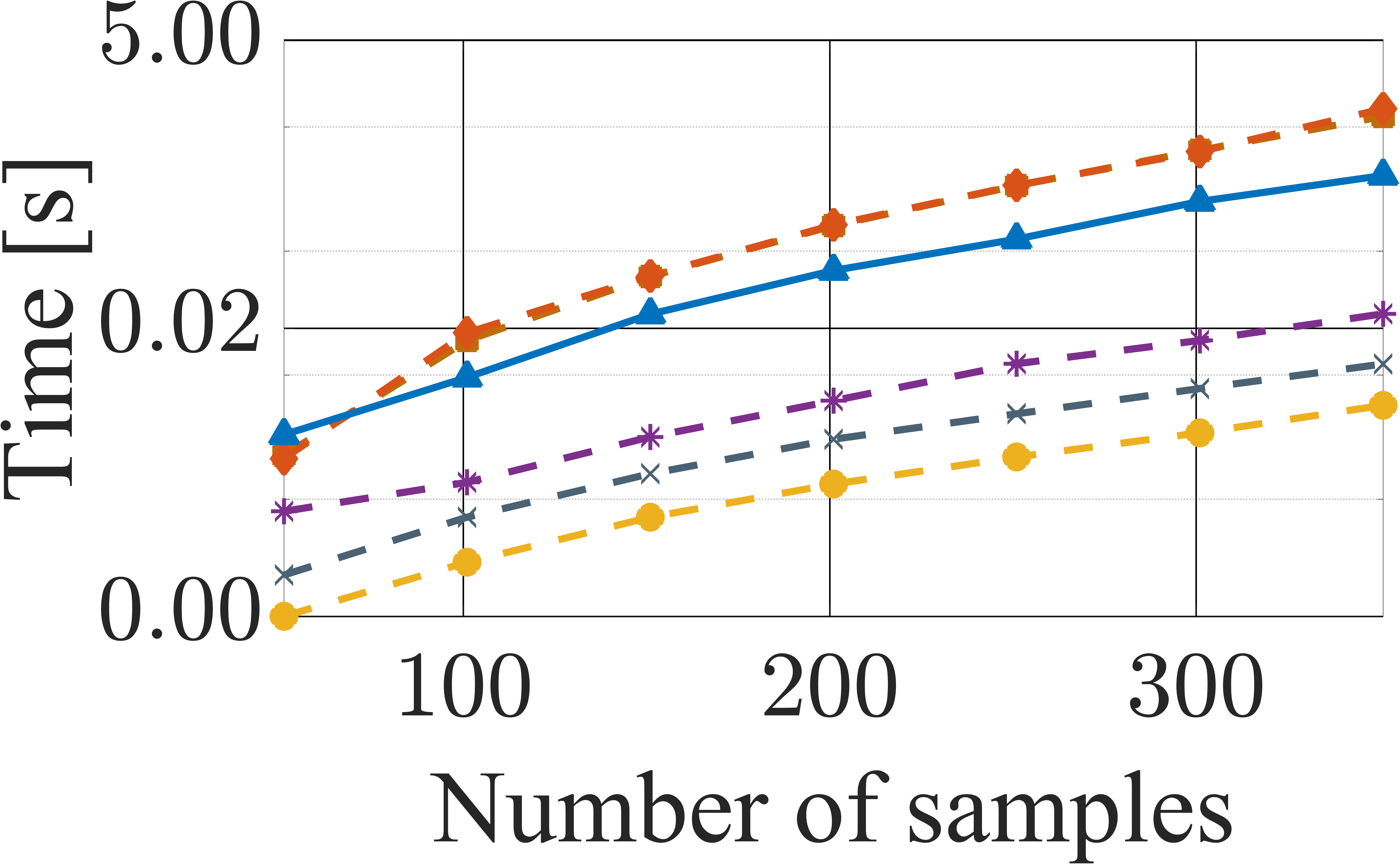}
        \subcaption{$M=4$}\label{fig:MDL_AIC_MF_come_time_compare_damped}
    \end{subfigure}
    \caption{Computational time (in log scale) versus number of samples for SAMP and baseline methods, at $\text{SNR}_{\text{dB}}=10$ dB.}\vspace{-2mm}
      \label{fig:MDL_AIC_MF_come_time_compare}
\end{figure}

\subsection{Amplitudes Estimation}

In \eqref{eq:new_amps} in Section \ref{sec: eff amp est}, we presented a computationally efficient approach for estimating the noisy amplitudes. We evaluate the performance of our method in comparison to the standard amplitude estimation technique presented in \eqref{eq:MP coeff estimation}. 
We set $r=L$ to concentrate only on the amplitude estimation,  circumventing any possible overestimation or underestimation of the model order by any of the methods.

{Fig. \ref{fig:amps compare} displays the amplitude estimation error for each method as a function of: (i) the computation time for different signal lengths in \ref{fig:amps compare 4 undamped INTERVAL}-\ref{fig:amps compare 4 damped INTERVAL}, and (ii) the SNR level in \ref{fig:amps compare 4 undamped SNR}-\ref{fig:amps compare 4 damped SNR}. {The error is defined as the average RMSE between each of the true amplitudes and the estimated amplitudes corresponding to the best $M$ estimated poles}. 
We report the case of $M=4$, as the complementary case of $M=2$ shows similar trends and is omitted for brevity.
We see in Fig.~\ref{fig:amps compare 4 undamped INTERVAL}-\ref{fig:amps compare 4 damped INTERVAL} that the proposed method obtains lower RMSE compared to the standard method, while requiring approximately $6$ times less computation time.
In addition, Fig.~\ref{fig:amps compare 4 undamped SNR}-\ref{fig:amps compare 4 damped SNR} show that the proposed approach achieves lower RMSE across nearly the entire simulated SNR range.
} 

\begin{figure}[t]
    \centering
    \begin{subfigure}[b]{0.49\columnwidth}
        \centering
        \includegraphics[width=\linewidth]{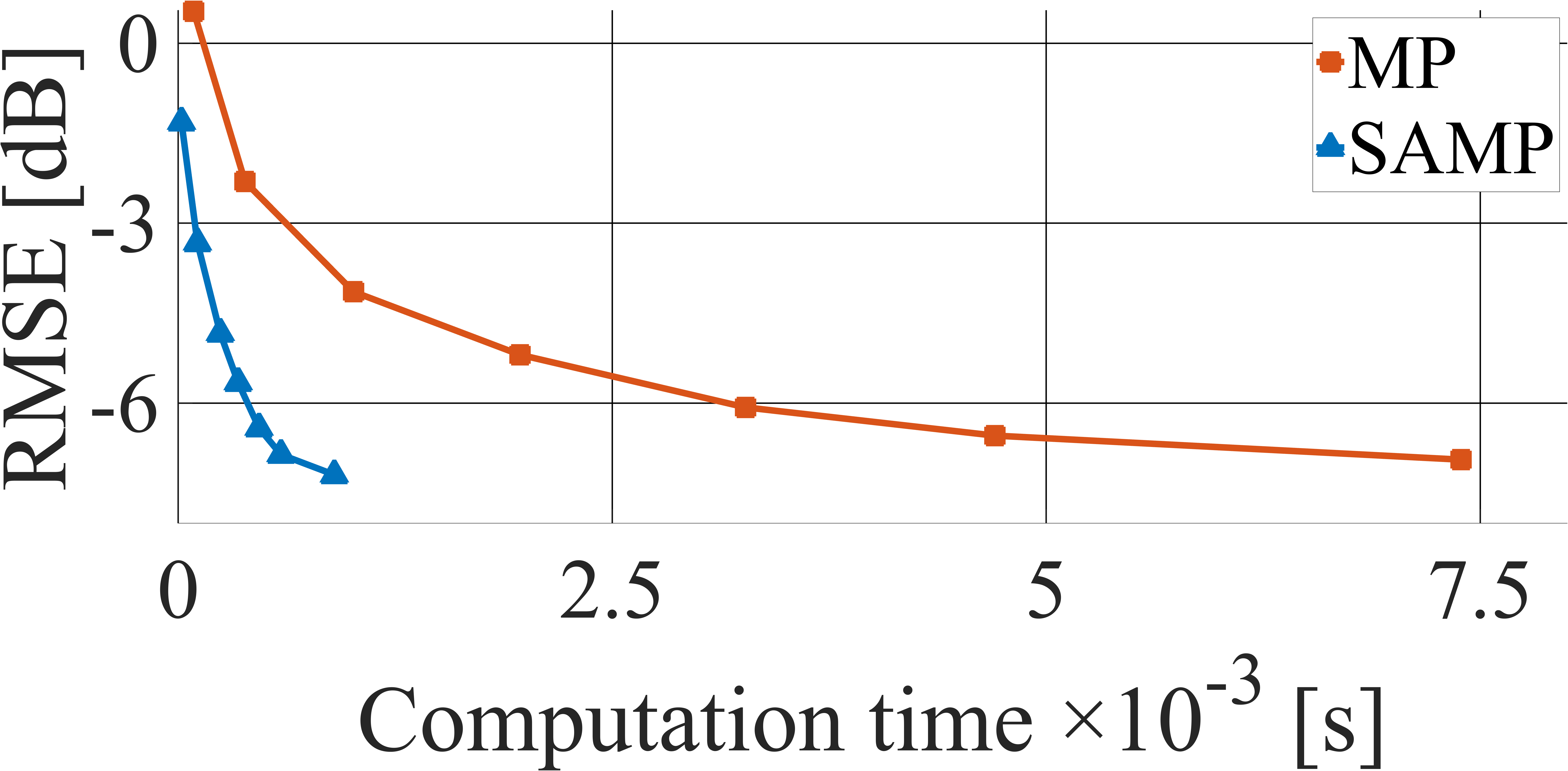}
        \caption{$M=4$, undamped}\label{fig:amps compare 4 undamped INTERVAL}
    \end{subfigure}
    \begin{subfigure}[b]{0.49\columnwidth}
        \centering
        \includegraphics[width=\linewidth]{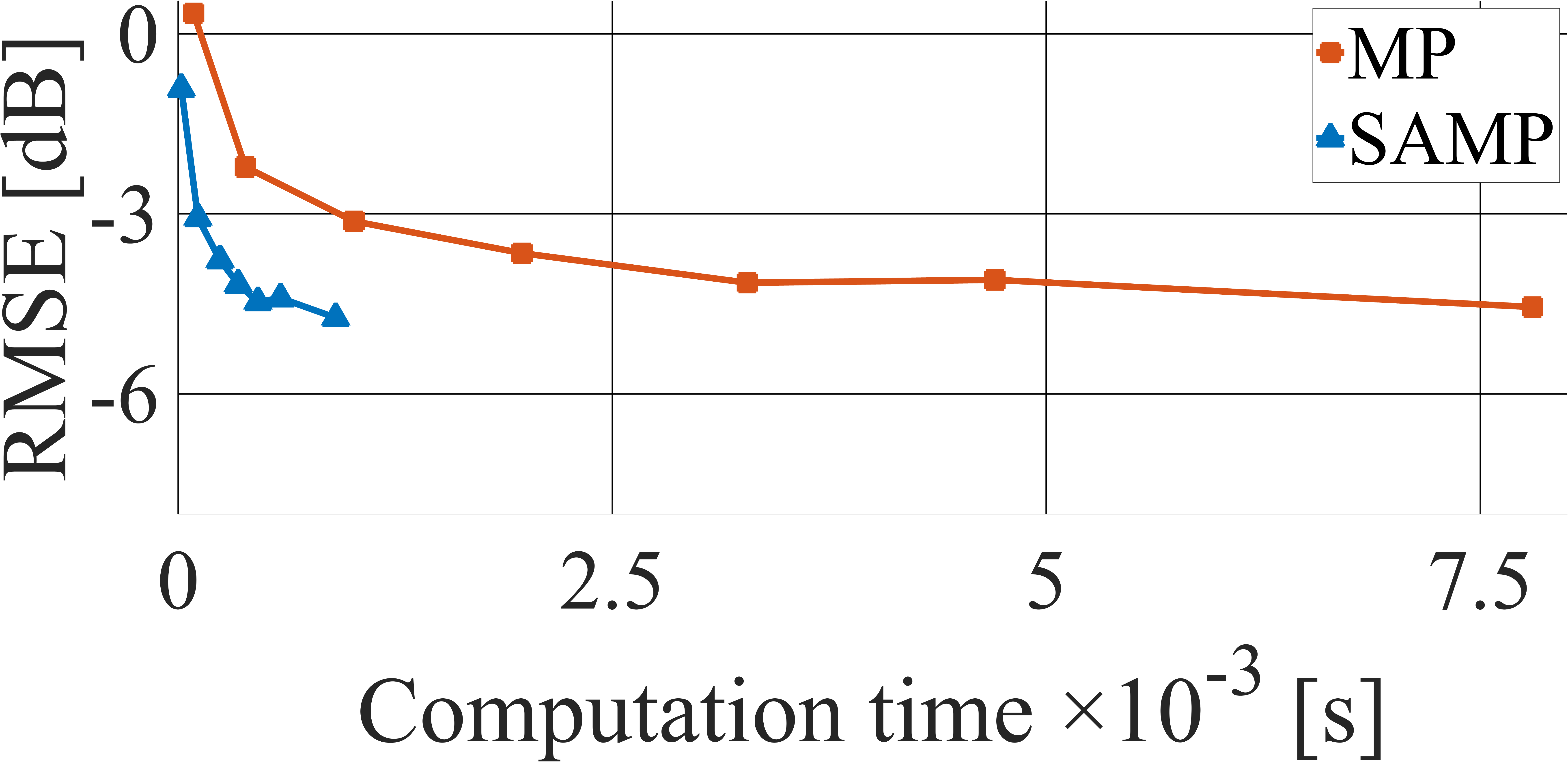}
        \caption{$M=4$, damped}\label{fig:amps compare 4 damped INTERVAL}
    \end{subfigure}
     \begin{subfigure}[b]{0.49\columnwidth}
        \centering
        \includegraphics[width=\linewidth]{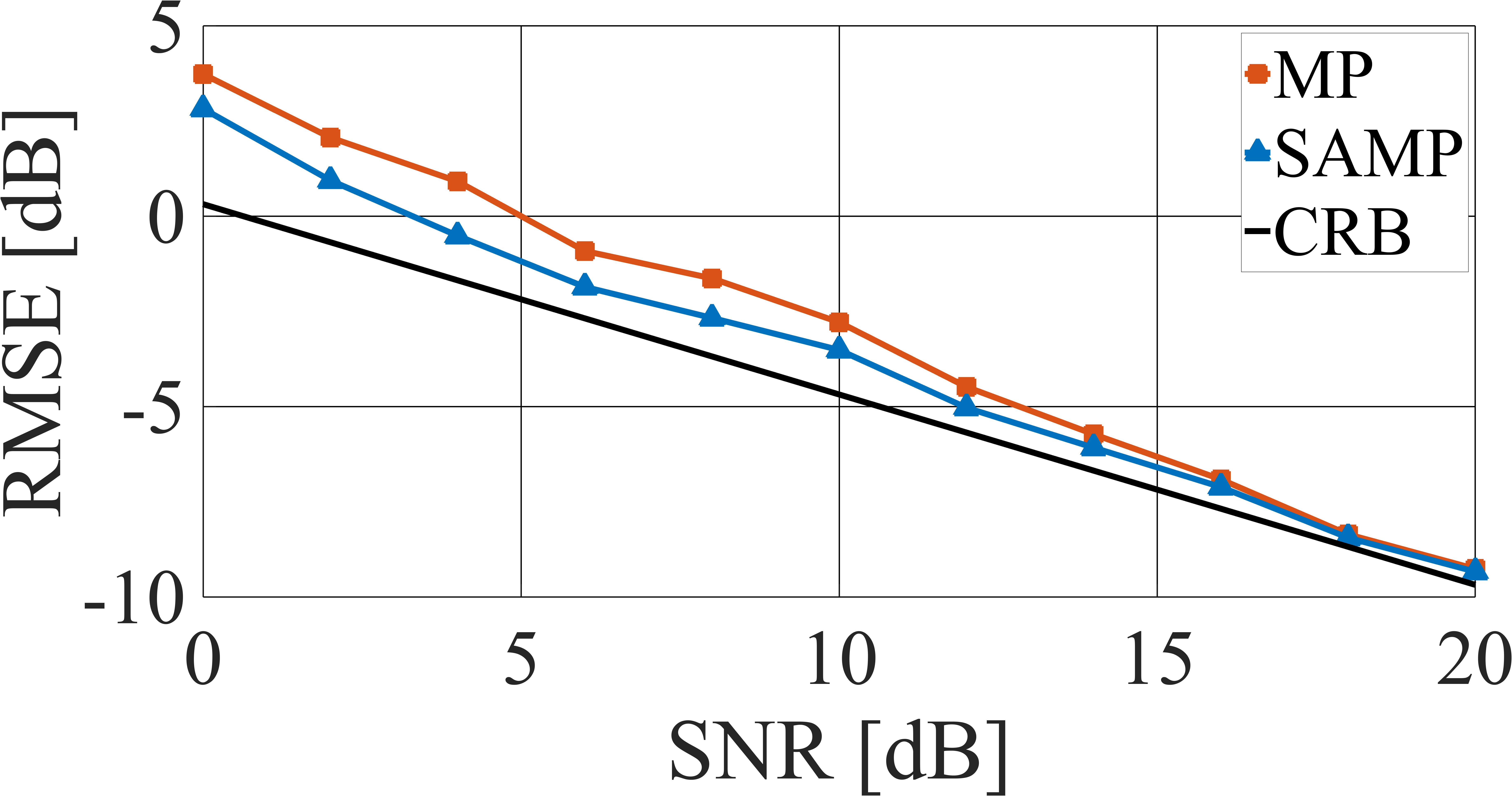}
        \caption{$M=4$, undamped}\label{fig:amps compare 4 undamped SNR}
    \end{subfigure}
    \begin{subfigure}[b]{0.49\columnwidth}
        \centering
        \includegraphics[width=\linewidth]{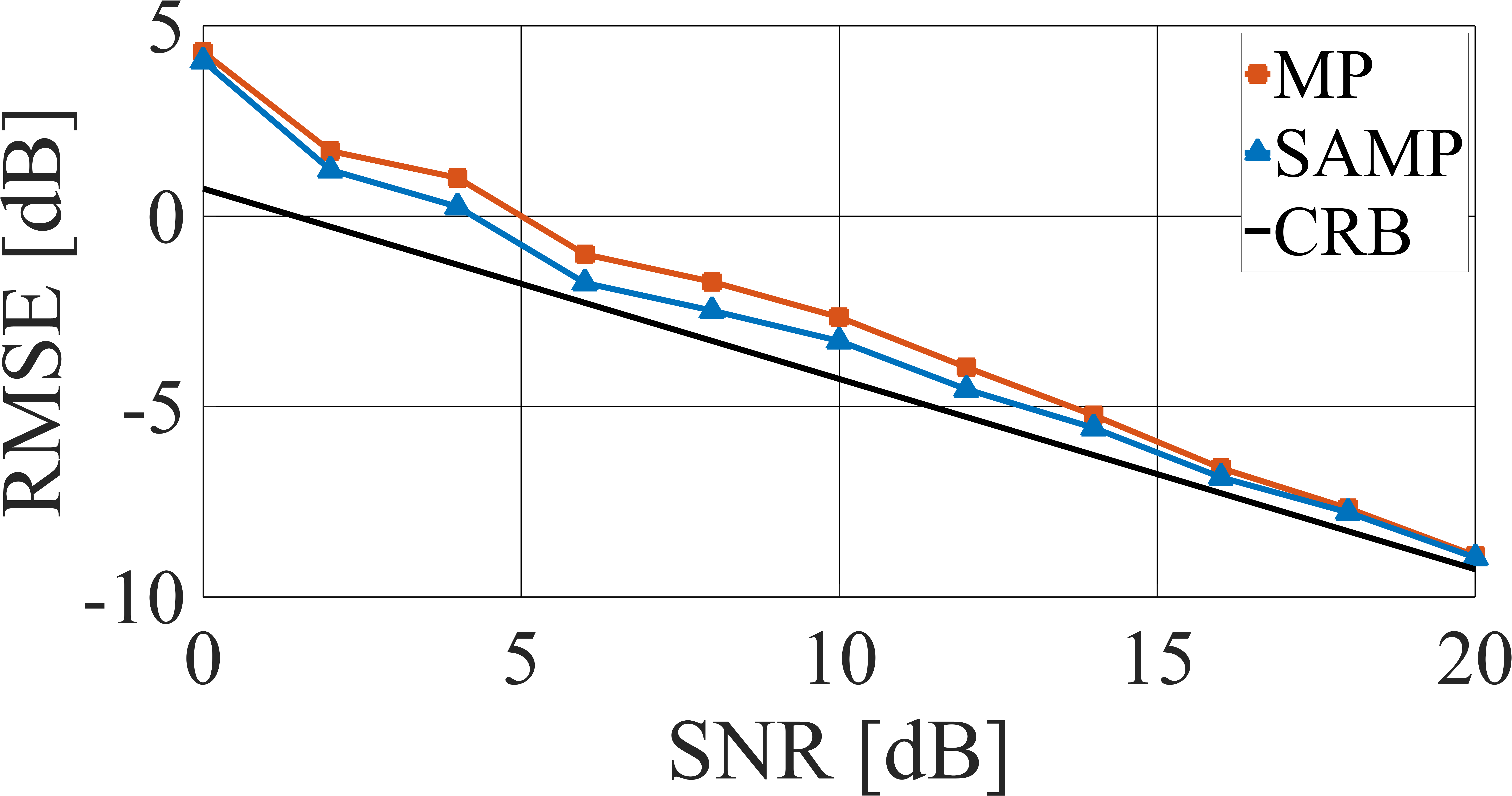}
        \caption{$M=4$, damped}\label{fig:amps compare 4 damped SNR}
    \end{subfigure}
      \caption{{Performance comparison between the proposed and the standard amplitude estimation methods. (a)-(b) show the RMSE vs computational time at $\text{SNR}_{\text{dB}}=10$ dB, and (c)-(d) show the RMSE vs SNR for $N=71$ samples. The frequency spacing is set to the Rayleigh limit.}}\vspace{-3mm}
      \label{fig:amps compare}
\end{figure}
\vspace{-2mm}
{
\subsection{Coarse-grid Granularity}\label{subsec: coarse-grid granularity}
{We solve the maximization problem in \eqref{eq: epsilon measure normalized}: $\underset{{z\in \mathcal{D}}}{\text{max}}\;P_i(z)$,
where $\mathcal{D} = \{z\in\mathbb{C}| z = e^{-\alpha + j\theta}, \: \alpha\in[0,1], \:\theta \in [-\pi,\pi)\}$, using a coarse grid search followed by biquadratic interpolation over a $3\times 3$ stencil, applied directly to $P_i(z)$, to mitigate grid-resolution effects and improves accuracy.} 
% The interpolation mitigates grid-resolution effects and improves accuracy by refining the maximizer to off-grid points.

{A full $2$D biquadratic interpolation is implemented by fitting a local biquadratic surface in $(\alpha,\theta)$ to the nine neighboring values of $P_i(z)$ around the coarse maximizer, including the mixed terms between $\alpha$ and $\theta$ \cite[Sec.~3.6]{press2007numerical}. A common simplified alternative performs two successive $1$D quadratic interpolations \cite[Chap.~5]{richards2005fundamentals}, first along $\theta$ and then along $\alpha$. In our simulations, both yielded similar results, while the latter was simpler and computationally more efficient. Therefore, we used the simplified version.}

A common practice is to choose the coarse grid such that the neighborhood used for local interpolation lies in a region where the objective is locally concave and well-approximated by a quadratic function \cite[Chap. 5]{richards2005fundamentals}. 
In our settings, we sample the rectangle $[0,1]\times[-\pi, \pi)$. Since damping attenuates the signal and reduces its effective length, thereby widening the main-lobe of the correlation kernel induced by the inner product in $P_i(z)$, we set the granularity of $\theta\in[-\pi,\pi)$ based on the undamped case $\alpha=0$, which is the most restrictive.
In this case, it is common to sample $[-\pi,\pi)$ uniformly with step size $\Delta\theta = \frac{\pi}{N-L}$, which is half the main-lobe width of the corresponding Dirichlet kernel. 
This choice ensures that the apparent peak and its two nearest neighbors fall within the same main-lobe \cite[Chap. 5]{richards2005fundamentals}.
A comparable guideline for sampling the damping axis $\alpha\in[0,1]$ is less standard; we therefore sample $[0,1]$ uniformly with step size $\Delta\alpha = \frac{1}{2(N-L)}$, following the same rationale.

To demonstrate the effect of grid resolution on the proposed algorithm, we conduct a sensitivity analysis with respect to the grid granularity, and evaluate its impact on model-order detection, parameter estimation, and computational load. Since we sample the rectangle $[0,1]\times[-\pi,\pi)$ uniformly with step sizes $\Delta\alpha$ and $\Delta\theta$, we vary the grid granularity by scaling these step sizes as
\[ 
(\Delta\alpha,\Delta\theta)\ \mapsto\ \beta(\Delta\alpha,\Delta\theta),
\]
where $\beta>0$ is a scalar termed the ``granularity factor''.

Fig.~\ref{fig:Pd and RMSE granularity sensitivity analysis} summarizes the performance of the proposed SAMP algorithm for granularity factors $\beta = [0.5,1,2,4,6,\dots,14]$, marked in the figure by curves in different colors, for the $M=4$ case. The complementary case of $M=2$ shows similar trends and is omitted for brevity. 
Figs.~\ref{fig: p_d vs SNR granularity 4 undamped}–\ref{fig: p_d vs SNR granularity 4 damped} are analogous to Figs.~\ref{fig: pd_SNR_4_undmp}–\ref{fig: pd_SNR_4_dmp}, but display $p_d$ versus SNR for different granularity factors. Similarly, Figs.~\ref{fig: RMSE vs SNR granularity 4 undamped}–\ref{fig: RMSE vs SNR granularity 4 damped} are analogous to Figs.~\ref{fig:4_freqs_RMSE_undamped_SNR}–\ref{fig:4_freqs_RMSE_damped_SNR}, but display the average RMSE of the estimates $\{\theta_i\}_{i=1}^M$ versus SNR, for different granularity factors.

The results show that SAMP remains stable for $\beta \lesssim 4$, but degrades rapidly for $\beta>4$. This suggests that, in practice, we can safely use a coarser grid, with step sizes $2(\Delta\alpha,\Delta\theta)$ for example, without a noticeable loss in performance.
We emphasize that all reported results correspond to off-grid simulations: (i) for $M=2$ we simulate two exponentials with $\theta_1=2 \; \text{rad/sample}$, and $\theta_2= \theta_1 + \frac{2\pi}{N}$. (ii) for $M=4$, we consider two clusters with $\theta_1=2 \; \text{rad/sample}$, $\theta_2= \theta_1 + \frac{2\pi}{N}$, and $\theta_3 = -\theta_1$, $\theta_4=-\theta_2$. These frequencies do not lie on the uniform grid with step size $\beta\Delta\theta$, for any of the simulated $\beta$.

Finally, Fig.~\ref{fig:SAMP_granularity_come_time_compare} is analogous to Fig.~\ref{fig:MDL_AIC_MF_come_time_compare}, but illustrates how grid granularity affects the computational load of SAMP. We see in the figure that even with a very coarse grid, SAMP has a non-negligible baseline runtime corresponding to the fixed-cost components of the algorithm (which are also shared by the standard MP). As the grid is refined, the runtime increases accordingly; nevertheless, even for a dense grid of $0.5(\Delta\alpha,\Delta\theta)$, SAMP attains comparable computation times relative to the EVT and AIC methods shown in Fig.~\ref{fig:MDL_AIC_MF_come_time_compare}.

\begin{figure}[t]
    \centering
    \begin{subfigure}[b]{0.49\columnwidth}
        \centering
        \includegraphics[width=\linewidth]{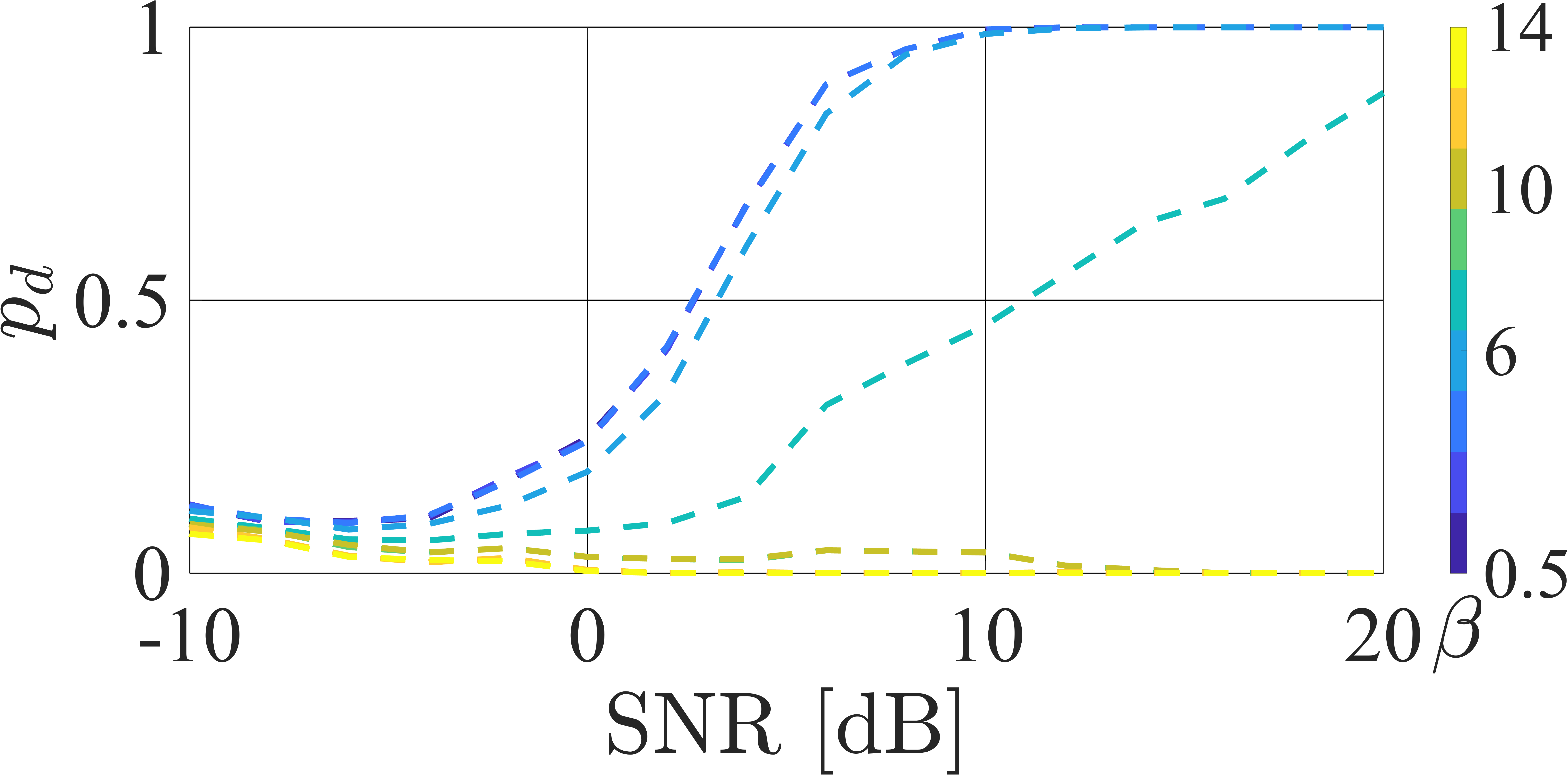}
        \caption{$M=4$, undamped}\label{fig: p_d vs SNR granularity 4 undamped}
    \end{subfigure}
    \begin{subfigure}[b]{0.49\columnwidth}
        \centering
        \includegraphics[width=\linewidth]{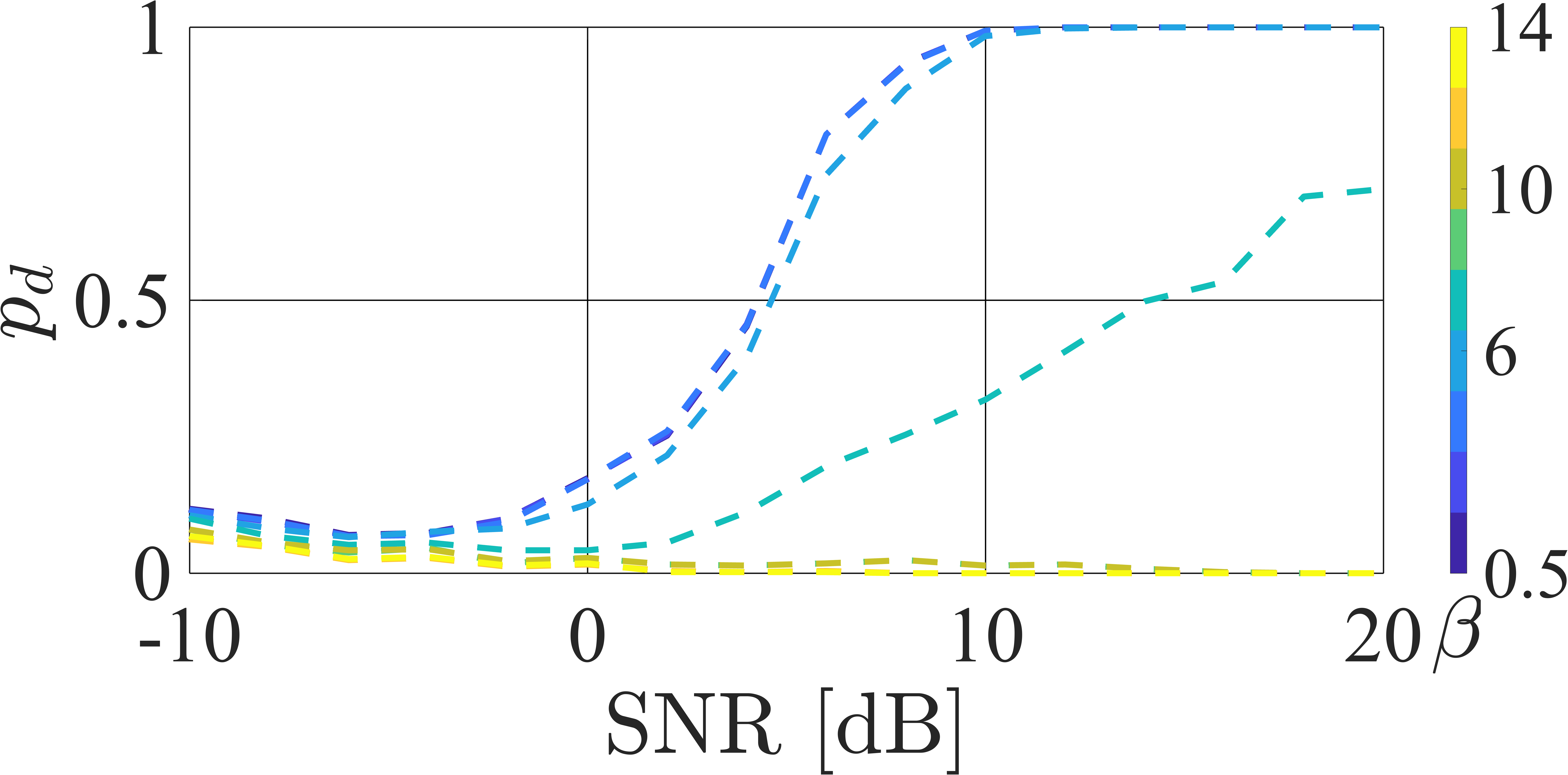}
        \caption{$M=4$, damped}\label{fig: p_d vs SNR granularity 4 damped}
    \end{subfigure}
     \begin{subfigure}[b]{0.49\columnwidth}
        \centering
        \includegraphics[width=\linewidth]{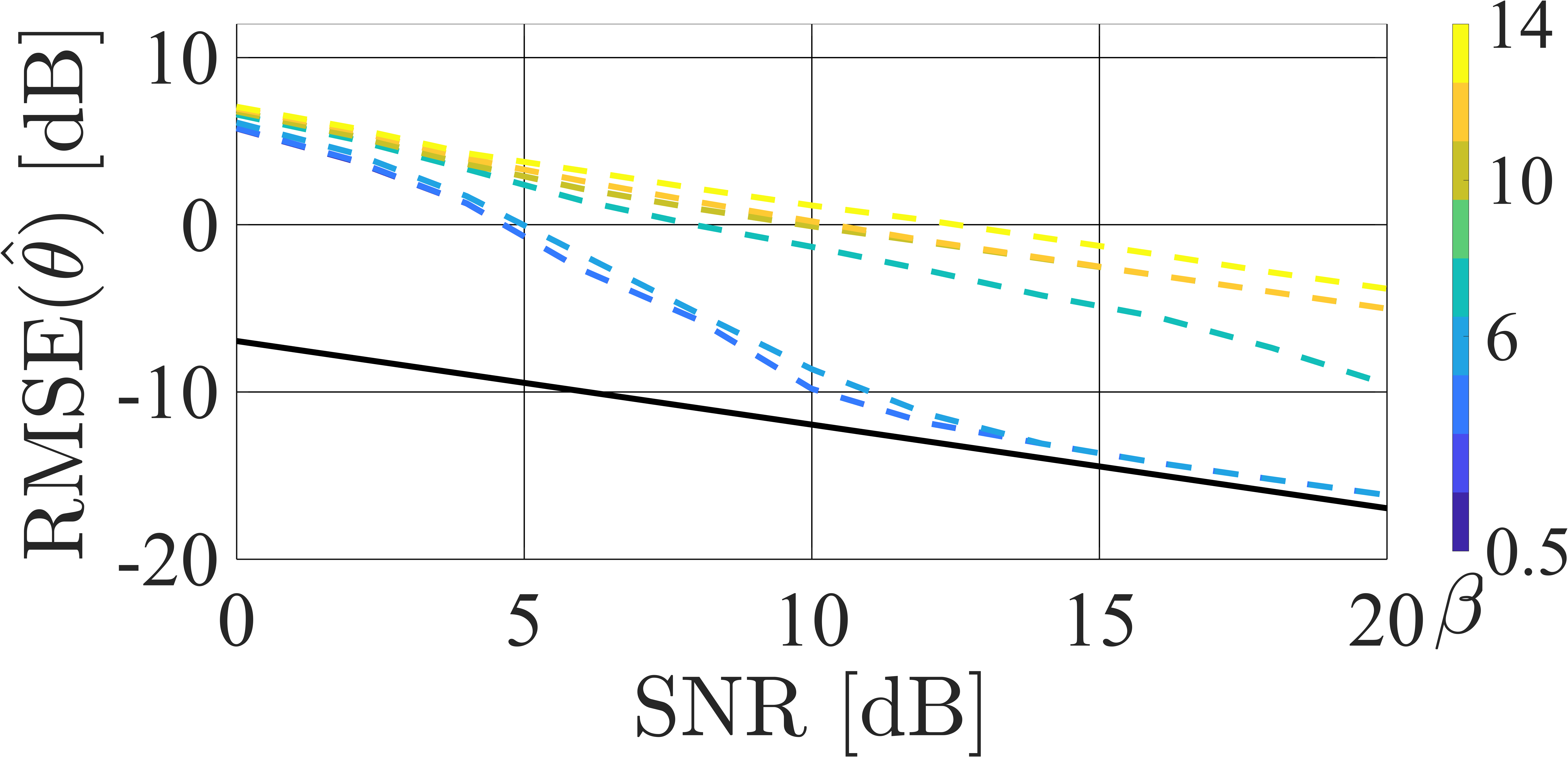}
        \caption{$M=4$, undamped}\label{fig: RMSE vs SNR granularity 4 undamped}
    \end{subfigure}
    \begin{subfigure}[b]{0.49\columnwidth}
        \centering
        \includegraphics[width=\linewidth]{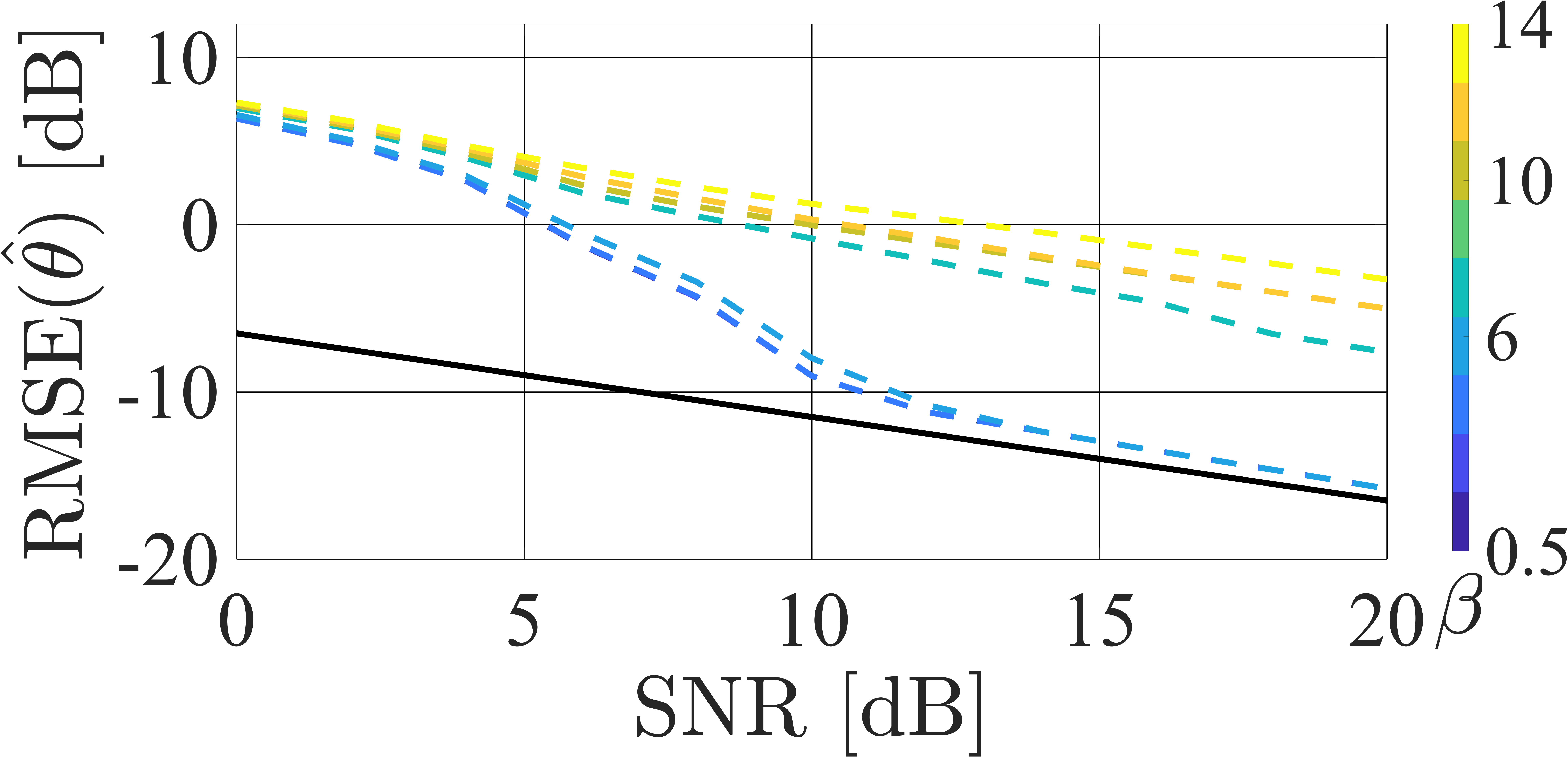}
        \caption{$M=4$, damped}\label{fig: RMSE vs SNR granularity 4 damped}
    \end{subfigure}
      \caption{{Probability of correct model order estimation versus $\text{SNR}_{\text{dB}}$ in (a)-(b), and average RMSE of $\{ \hat\theta_i \}_{i=1}^M$ versus the $\text{SNR}_{\text{dB}}$ in (c)-(d). Colors correspond to different values of $\beta$, which control the grid resolution (see colorbar).}}
      \label{fig:Pd and RMSE granularity sensitivity analysis}
\end{figure}

\begin{figure}[t]
    \centering
    \begin{subfigure}[b]{0.49\columnwidth}
        \centering
        \includegraphics[width=\linewidth]{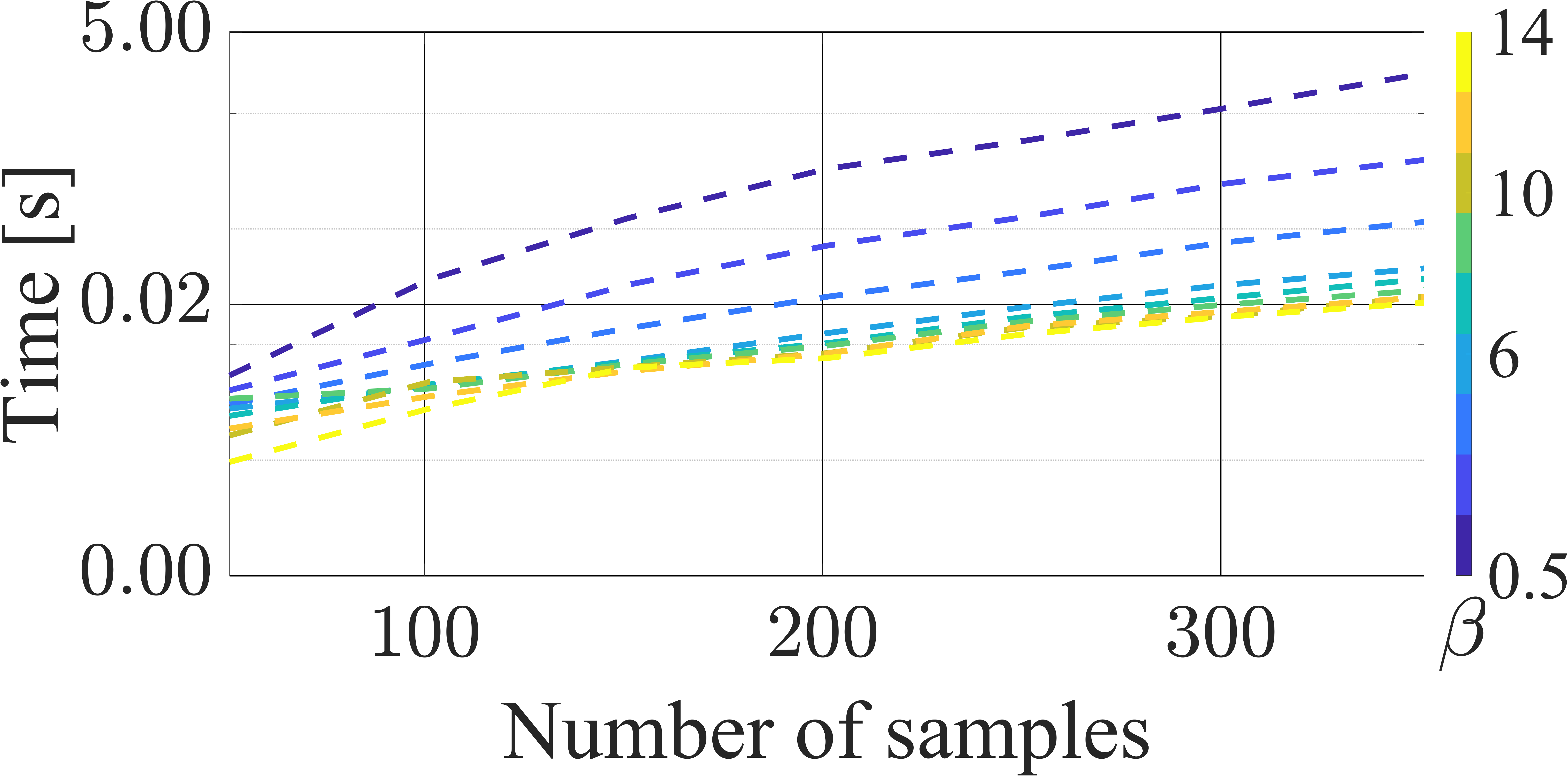}
        \subcaption{$M=2$}\label{fig:SAMP_granularity_time_compare_2_undamped}
    \end{subfigure}
    \begin{subfigure}[b]{0.49\columnwidth}
        \centering
        \includegraphics[width=\linewidth]{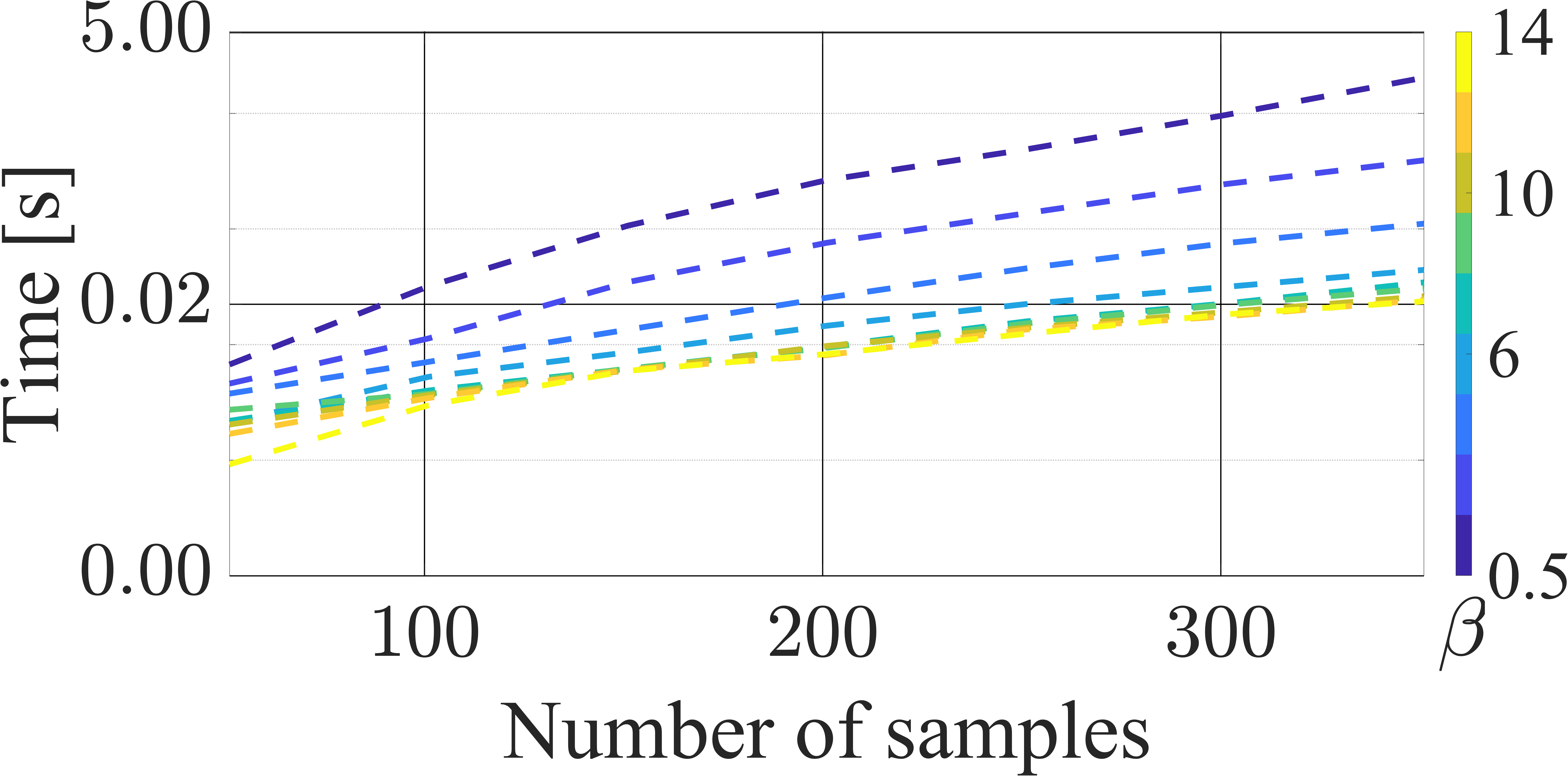}
        \subcaption{$M=4$}\label{fig:SAMP_granularity_time_compare_4_undamped}
    \end{subfigure}
    \caption{{Computational time (log scale) vs number of samples for SAMP at $\text{SNR}_{\text{dB}}=10$ dB. Colors correspond to different values of $\beta$, which control the grid resolution (see colorbar).}}\vspace{-2mm}
      \label{fig:SAMP_granularity_come_time_compare}
\end{figure}
}

\section{Conclusion} \label{sec: Conclusion}
We addressed the problem of detecting the number of complex exponentials and estimating their parameters from noisy measurements using the MP method. We introduced the concept of MP \textit{modes}, through an extension of the MP theorem to noisy conditions, revealing their informative structure. 
Specifically, we showed that the signal-related modes can be expressed as an obscured Vandermonde matrix, $\mathbf{Z}_L + \mathbf{E}_L$, and provided a detailed analysis of the noise-induced term $\mathbf{E}_L$. 

By leveraging the spectral structure of the signal-related modes and utilizing the noise-term characterization, we showed that temporal information can be effectively extracted, significantly improving model order selection, which often relies solely on singular values and thresholding. Utilizing the MP modes, we proposed the SAMP algorithm for detecting and estimating the signal-related components, and a computationally efficient method for estimating the signal's amplitudes.
A key advantage of using the MP modes lies in its ability to capture the temporal information encoded in the relationship between $\mathbf{Y}_0$ and $\mathbf{Y}_1$. In contrast, singular values are computed solely from $\mathbf{Y}_0$ and therefore miss this temporal information.

Our simulations indicate that the proposed method outperforms the classical thresholding-based competing methods by a large margin in detecting the model order, and nearly achieves the CRB for a SNR above a threshold, in estimating the signal's parameters. When compared to \textit{coupled detection-estimation} approaches, such as EVT and AIC, the proposed method achieves comparable or better results with lower computational time. Additionally, unlike EVT and AIC, the proposed method does not require prior knowledge of the noise distribution and was shown to be much less sensitive to noise model mismatch.  

We postulate that our approach is general, and it is not restricted to the MP method only. For example, several related methods generate Hankel matrices similar to the MP method and use a truncated SVD step to reduce noise. These include the SVD-based Prony method \cite{kumaresan1982estimating}, State Space method (SSM) \cite{kung1983state}, Dynamic Mode Decomposition (DMD) \cite{schmid2010dynamic}, and Singular Spectrum Analysis (SSA) \cite{broomhead1986extracting,elsner1996singular}. 
% The SVD-based Prony method produces frequency estimates comparable to the MP method in noiseless scenarios \cite{sarkar1995using}. SSM was shown to be equivalent to the MP method, to first-order, in signal pole estimation \cite{hua1991svd,jang2005quantitative}. DMD, a data-driven approach for high-dimensional dynamical systems, generates identical Hankel matrices to MP when using delayed coordinates for one-dimensional data \cite{tu2013dynamic}, with a partial comparison available in \cite{pogorelyuk2018clustering}. SSA also utilizes Hankel matrices and SVD but, to the best of our knowledge, lacks direct comparisons to the MP method.

In future research, we will extend these findings to multivariate signals and explore the concept of \textit{modes} in other super-resolution methods, such as ESPRIT \cite{roy1989esprit}.

% \section*{Acknowledgments}
% The authors would like to thank the associate editor and the anonymous reviewers for their insightful and constructive comments which helped to improve and clarify the paper. 

\printbibliography
\end{refsection}
\newpage

\twocolumn[
\begin{center}

    {\huge Structure-Aware Matrix Pencil Method\\  \vspace{0.2cm}
    \textbf{Supplementary Materials}}
    
    \vspace{0.5cm}
    
    {Yehonatan-Itay Segman, Alon Amar, and Ronen Talmon,~\IEEEmembership{Senior Member,~IEEE}  } 
    \vspace{0.5cm}
    
\end{center}
]

\setcounter{page}{1}
{
\appendices
\begin{NoHyper}
\begin{refsection}
\section{Simple eigenvector perturbation}\label{app:simple eigenvector pert}
\renewcommand\thefootnote{}
\footnote{
The authors are with the Viterbi Faculty of Electrical and Computer Engineering, Technion -- Israel Institute of Technology, Haifa 32000, Israel. (Corresponding author: Yehonatan-Itay Segman, email: yehonatans@campus.technion.ac.il).}
\renewcommand\thefootnote{\arabic{footnote}}
\addtocounter{footnote}{-1}
In this appendix, we derive a first-order approximation for a simple eigenvector under perturbation. Simple eigenvectors, associated with simple eigenvalues, are known to display stable behavior under perturbations \cite{magnus1985differentiating, greenbaum2020first}. Following \cite{magnus1985differentiating}, let $\mathbf{A} \in \mathbb{C}^{n \times n}$ be a complex square matrix, and let $\mathbf{v}_i$ denote a simple right eigenvector of $\mathbf{A}$ corresponding to a simple eigenvalue $\lambda_i$.
For a perturbed matrix $\mathbf{A} + \delta \mathbf{A}$, the author established the existence of a complex function $\lambda$ and a complex vector function $\mathbf{v}$ defined for all $\delta \mathbf{A}$ in a neighborhood of $\mathbf{A}$, satisfying $\lambda(\mathbf{A}) = \lambda_i$, $\mathbf{v}(\mathbf{A}) = \mathbf{v}_i$, and $(\mathbf{A} + \delta \mathbf{A}) \mathbf{v} = \lambda \mathbf{v}$. These functions are infinitely differentiable within this neighborhood, with their differentials at $\mathbf{A}$ also provided. While the proof is clear and elegant, we seek an alternative proof that will provide additional insights into the neighborhood in which eigenvector perturbation can be characterized.

Our approach follows the technique outlined in \cite[Theorem 4.1]{deif1995rigorous}, which assumes that the matrix $\mathbf{A}$ is non-defective, forming an eigenvector basis for $\mathbb{C}^{n}$. Let $\mathbf{V}$ and $\mathbf{U}$ denote the matrices of its right and left eigenvectors, respectively, where \(\mathbf{u}_k^{\mathrm{H}}\) is the \(k\)-th row of \(\mathbf{U}^{\mathrm{H}}\) and \(\mathbf{v}_i\) is the \(i\)-th column of \(\mathbf{V}\). The following lemma characterizes the approximate behavior of \(\mathbf{v}_i\) under perturbation, provided that a delicate balance between perturbation and eigenvalue separation is maintained.

\begin{lemma}\label{lemma:first order approximation of the vector of coefficients}
Let $\mathbf{v}_i$ be a simple right eigenvector of a non-defective matrix $\mathbf{A} \in \mathbb{C} ^{n \times n}$, and let $\lambda_i$ be its associated simple eigenvalue. Then there is a unique right eigenvector $\widetilde{\mathbf{v}}_i = \mathbf{v}_i + \delta \mathbf{v}_i$ of the matrix $\mathbf{A} + \delta\mathbf{A}$ and a corresponding eigenvalue $\Tilde{\lambda}_i = \lambda_i + \delta\lambda_i$ such that $\delta \lambda_i$ and $\delta \mathbf{v}_i$ can be approximated to first order by
\begin{align}
\delta \lambda_i &= \frac{\mathbf{u}_i^{\mathrm{H}}\delta \mathbf{A} \mathbf{v}_i}{\mathbf{u}_i^{\mathrm{H}} \mathbf{v}_i}, \\   
\delta \mathbf{v}_i &= \sum\limits_{\substack{k=1\\k \neq i}}^n \frac{\mathbf{u}_k^{\mathrm{H}} \delta \mathbf{A} \mathbf{v}_i}{\mathbf{u}_k^{\mathrm{H}} \mathbf{v}_k(\Tilde{\lambda}_i -\lambda_k)} \mathbf{v}_k.
\end{align} Provided that 
\begin{equation}\label{app_assum:spectral radius}
\rho(\mathbf{D}_i\mathbf{U}^{\mathrm{H}}\delta \mathbf{A}\mathbf{V}) <1,
\end{equation} and
\begin{equation}\label{app_assum:eigenvectors normalization}
\mathbf{u}_i^{\mathrm{H}}\mathbf{v}_i = 1 \; , \; \mathbf{u}_i^{\mathrm{H}}\widetilde{\mathbf{v}}_i = 1,
\end{equation}
where $\rho(\cdot)$ is the spectral radius, and the matrix $\mathbf{D}_i$ is a diagonal matrix defined by
\begin{equation} \label{eq:diagonal matrix Di}
\mathbf{D}_i = \text{diag}\Bigg(\frac{1}{\mathbf{u}_1^{\mathrm{H}} \mathbf{v}_1(\Tilde{\lambda}_i - \lambda_1)}, \cdots, 0_i, \cdots,\frac{1}{\mathbf{u}_n^{\mathrm{H}} \mathbf{v}_n(\Tilde{\lambda}_i - \lambda_n)}\Bigg).
\end{equation}    
\end{lemma}

\begin{proof}  
The fact that there is a unique perturbed eigenvalue $\Tilde{\lambda}_i$ such that 
\begin{equation} \label{eq:unique pert eigenvalue}
\delta \lambda_i = \frac{\mathbf{u}_i^{\mathrm{H}} \delta\mathbf{A} \mathbf{v}_i}{\mathbf{u}_i^{\mathrm{H}} \mathbf{v}_i} + O(\left\Vert\delta \mathbf{A}\right\Vert^2), 
\end{equation} is provided in \cite[chapter IV, Theorem 2.3]{stewart1990matrix}. Let $\widetilde{\mathbf{v}}_i$ be its corresponding perturbed eigenvector.
The non-defectiveness of the matrix $\mathbf{A}$ implies that we can expand $\delta \mathbf{v}_i = \widetilde{\mathbf{v}}_i - \mathbf{v}_i$ in the eigenvector basis of $\mathbf{A}$ as:
\begin{equation} \label{eq:delta_q expansion}
\delta \mathbf{v}_i = \sum\limits_{k=1}^n a_{ik}\mathbf{v}_k.
\end{equation}
Left multiplying \eqref{eq:delta_q expansion} by $\mathbf{u}_i^{\mathrm{H}}$ results in:
\begin{equation}
    a_{ii} = \frac{\mathbf{u}_i^{\mathrm{H}} \widetilde{\mathbf{v}}_i - \mathbf{u}_i^{\mathrm{H}} \mathbf{v}_i}{ \mathbf{u}_i^{\mathrm{H}} \mathbf{v}_i}.
\end{equation}
It follows, by assumption \eqref{app_assum:eigenvectors normalization}, that $a_{ii} = 0$. Thus, \eqref{eq:delta_q expansion} can be written as:
\begin{equation}\label{eq: delta_v_expansion}
    \delta \mathbf{v}_i = \sum\limits_{\substack{k=1\\k \neq i}}^n a_{ik}\mathbf{v}_k. 
\end{equation}
It is noteworthy that as the matrix $\mathbf{A}$ is non-defective, its right eigenvector matrix $\mathbf{V}$ is invertible. As eigenvectors are unique up to scalar multiplication, it follows that there exists a diagonal matrix $\mathbf{C} = \text{diag}(c_1,\ldots,c_n)$, where $\{c_k\}_{k=1}^n \neq 0$, such that $\mathbf{C} \mathbf{V}^{-1} = \mathbf{U}^{\mathrm{H}}$ or equivalently $\mathbf{C} = \mathbf{U}^{\mathrm{H}} \mathbf{V}$. Consequently, each element $c_k$ can be expressed as $c_k = \mathbf{u}_k^{\mathrm{H}} \mathbf{v}_k$, ensuring that all the terms involving the reciprocal of $\mathbf{u}_i^{\mathrm{H}} \mathbf{v}_i$ are well defined.

The perturbed eigenvalue problem can be stated as
\[
(\mathbf{A} +\delta\mathbf{A})(\mathbf{v}_i + \delta \mathbf{v}_i) = (\lambda_i + \delta \lambda_i)(\mathbf{v}_i + \delta \mathbf{v}_i).
\] 
Left multiplying by $\mathbf{u}_k^{\mathrm{H}}$ for $k \neq i$ and rearranging, we can express the vector of coefficients $\mathbf{a}_i$, from \eqref{eq: delta_v_expansion}, by the following relation:
\begin{equation} \label{eq:vector of coefficients}
    \left(\mathbf{I} - \mathbf{D}_i\mathbf{U}^{\mathrm{H}}\delta \mathbf{A}\mathbf{V}\right)\mathbf{a}_i = \mathbf{D}_i\mathbf{U}^{\mathrm{H}}\delta \mathbf{A}\mathbf{v}_i,
\end{equation} 
where the matrix $\mathbf{D}_i$ is defined in \eqref{eq:diagonal matrix Di}. 

By Assumption \eqref{app_assum:spectral radius}, the spectral radius $\rho(\mathbf{D}_i\mathbf{U}^{\mathrm{H}}\delta \mathbf{A}\mathbf{V}) <1$, and we can explicitly express the coefficients vector $\mathbf{a}_i$ by inverting the matrix $\left(\mathbf{I}-\mathbf{D}_i\mathbf{U}^{\mathrm{H}}\delta \mathbf{A}\mathbf{V}\right)$.
Finally, as $\rho(\mathbf{D}_i\mathbf{U}^{\mathrm{H}}\delta \mathbf{A}\mathbf{V}) <1$, using the Neumann series representation of the matrix $\left(\mathbf{I}-\mathbf{D}_i\mathbf{U}^{\mathrm{H}} \delta \mathbf{A}\mathbf{V}\right)^{-1}$, we get the following first-order approximation for the coefficients vector: 
\begin{equation}
    \mathbf{a}_i = \mathbf{D}_i\mathbf{U}^{\mathrm{H}}\delta \mathbf{A}\mathbf{v}_i.
\end{equation}
This leads us to the following expression for $a_{ik}$:
\begin{equation}
a_{ik}=
\begin{cases}
0 &, \text{if}\ k=i, \\
\frac{\mathbf{u}_k^{\mathrm{H}} \delta \mathbf{A} \mathbf{v}_i}{\mathbf{u}_k^{\mathrm{H}} \mathbf{v}_k(\Tilde{\lambda}_i -\lambda_k}) &, \text{otherwise}
\end{cases}
\end{equation}
which completes the proof. 
\end{proof}

\section{Extending the Matrix Pencil Theorem: Accounting for Noise Effects}
\label{app:Extending the MP Theorem: Accounting for Noise}

In this section, we extend the MP theorem \cite[Theorem 2.1]{hua1990matrix} to the noisy case and analyze the noise-induced effects.
By doing so, we bridge the gap between the existing theory, described in Section \ref{subsec: MP noiseless case}, and practice, described in Section \ref{subsec:MP method - model order detection}. 

\begin{theorem}
\label{theo:MSMP} 
Assume the underlying signal is given by \eqref{eq: noisy discrete signal}. Let $\mathbf{Y}_0$ and $ \mathbf{Y}_1$ be the Hankel matrices defined in \eqref{eq:full noisy Hankel matrix}. Then, $\mathbf{Y}_0$ and $ \mathbf{Y}_1$ admit the following factorization:
    \begin{align}
         \mathbf{Y}_0 &= \widetilde{\mathbf{Z}}_L \widetilde{ \mathbf{B} } \widetilde{\mathbf{Z}}_R,\\ 
         \mathbf{Y}_1 &= \widetilde{\mathbf{Z}}_L\widetilde{ \mathbf{B} } \widetilde{\mathbf{\Lambda}} \widetilde{\mathbf{Z}}_R, 
    \end{align} where
\begin{enumerate}

    \item{The diagonal matrix $\widetilde{\mathbf{\Lambda}}$ consists of the MP eigenvalues, acts as a forward propagation in time, and can be recast as: 
    \begin{equation} \label{eq:Lambda decomp}
    \widetilde{\mathbf{\Lambda}} =  
    \begin{bmatrix} 
       \widetilde{\mathbf{Z}} & 0 \\ 
      0 & \widetilde{\mathbf{0}} 
    \end{bmatrix}_{L\times L},
    \end{equation} where 
    \begin{enumerate}
        \item{The diagonal matrix $ \widetilde{\mathbf{Z}} \in\mathbb{C}^{M \times M}$ comprises the signal-related eigenvalues.}
        % , and each of its diagonal elements $\{\lambda_i \in \text{diag}(\mathbf{\Lambda}^s)\}_{i=1}^M$ corresponds exactly to one of the signal poles $\{z_i\}_{i=1}^M$.} 
        \item{The diagonal matrix $\widetilde{\mathbf{0}} \in\mathbb{C}^{\left(L-M\right) \times \left(L-M\right)}$ consists of the noise-related eigenvalues.}
    \end{enumerate} }
    
    \item{The diagonal matrix $\widetilde{ \mathbf{B} }$ can be cast as 
    \begin{equation}
        \widetilde{ \mathbf{B} } =  
    \begin{bmatrix} 
      \mathbf{B} & 0 \\ 
      0 & \mathbf{\mathbf{I}}_{L-M} 
    \end{bmatrix}_{L\times L},
    \end{equation} where the diagonal matrix $\mathbf{B} \in\mathbb{C}^{M \times M}$ consists of the signal amplitudes as defined in \eqref{eq:clean Vandermonde factorization}. }

    \item{$\widetilde{\mathbf{Z}}_R$, which is defined in \eqref{def:right mp mode} as the pseudo-inverse of the right eigenvector matrix of the noisy MP, can be recast as:
     \begin{equation} 
    \label{eq:right MP mode decomposition}  
    \widetilde{\mathbf{Z}}_R = 
    \begin{bmatrix}
    \mathbf{Z}_R +  \mathbf{E}_R \\
    \mathbf{C}_R
    \end{bmatrix}_{L \times L},
    \end{equation} 
    where the top $M$ rows of $\widetilde{\mathbf{Z}}_R$ are the sum of the Vandermonde matrix $\mathbf{Z}_R$ and a noise-related component, denoted by $\mathbf{E}_R$. The remaining $L-M$ rows are noise-related spurious rows, denoted by $\mathbf{C}_R.$ }
    
    \item {$\widetilde{\mathbf{Z}}_L$, which is defined in \eqref{def:left mp mode} as the pseudo-inverse of the left eigenvector matrix of the noisy MP, can be recast as 
    \begin{equation} 
    % \label{eq:left MP mode decomposition}   
    \widetilde{\mathbf{Z}}_L = 
    \begin{bmatrix} 
      \mathbf{Z}_L + \mathbf{E}_L, \; \mathbf{C}_L \\ 
    \end{bmatrix}_{(N-L) \times L},
    \end{equation}
    where the leftmost $M$ columns of $\widetilde{\mathbf{Z}}_L$ are the sum of the Vandermonde matrix $\mathbf{Z}_L$ and a noise-related component, denoted by $\mathbf{E}_L$. The remaining $L-M$ columns are noise-related spurious columns, denoted by $\mathbf{C}_L.$}

\end{enumerate}
\end{theorem}

\begin{proof}
Using the definition of $\widetilde{\mathbf{Z}}_R^{\dagger}$ and $\widetilde{\mathbf{Z}}_L^{\dagger}$ from \eqref{eq:right noisy eigenvectors} and \eqref{eq:left noisy eigenvectors}, and the SVD of $\mathbf{Y}_0 = \mathbf{U\Sigma V}^{\mathrm{H}}$ from \eqref{eq:SVD of Y_0}, it is easy to verify that $\widetilde{\mathbf{Z}}_L^{\dagger} \mathbf{Y}_0 \widetilde{\mathbf{Z}}_R^{\dagger} = \mathbf{I}_L$, and $\widetilde{\mathbf{Z}}_L^{\dagger} \mathbf{Y}_1 \widetilde{\mathbf{Z}}_R^{\dagger} = \mathbf{\Lambda}$, where $\mathbf{I}_L$ is the $L\times L$ identity matrix. From these equalities we find a new representation for the noisy Hankel matrices:
\begin{align} \label{eq:Y0 factorization}
\mathbf{Y}_0 &= \widetilde{\mathbf{Z}}_L \widetilde{\mathbf{Z}}_R,\\
\mathbf{Y}_1 &= \widetilde{\mathbf{Z}}_L \mathbf{\Lambda} \widetilde{\mathbf{Z}}_R. \label{eq:Y1 factorization}
\end{align}
By employing the Vandermonde factorization of $\mathbf{X}_0$, we can express the Hankel matrix $\mathbf{Y}_0$ as
\begin{equation}\label{eq:noise_factor_right}
\mathbf{Y}_0 = \mathbf{X}_0 + \mathbf{W}_0 = \mathbf{Z}_L\mathbf{B}\mathbf{Z}_R + \mathbf{W}_0.    
\end{equation} Let us re-frame equation \eqref{eq:noise_factor_right} as follows: 
\begin{equation}\label{eq:noise_factor_right_symm}
\mathbf{Y}_0 = (\mathbf{Z}_L \mathbf{B}^{\frac{1}{2}})(\mathbf{B}^{\frac{1}{2}} \mathbf{Z}_R) + \mathbf{W}_0,
\end{equation} as this formulation will prove to be useful in subsequent discussions.
Combining \eqref{eq:Y0 factorization} and \eqref{eq:noise_factor_right_symm}, we get:
\begin{equation} \label{eq:noisy left mode}
\widetilde{\mathbf{Z}}_L  = (\mathbf{Z}_L\mathbf{B}^{\frac{1}{2}})(\mathbf{B}^{\frac{1}{2}}\mathbf{Z}_R)\widetilde{\mathbf{Z}}_R^{\dagger} + \mathbf{W}_0 \widetilde{\mathbf{Z}}_R^{\dagger}.
\end{equation}
Now, let us analyze the term $\mathbf{B}^{\frac{1}{2}} \mathbf{Z}_R\widetilde{\mathbf{Z}}_R^{\dagger}$. As established in \cite[Theorem 2.1]{hua1990matrix}, the columns of $\mathbf{Z}_R^{\dagger}$ (and hence the columns of $\mathbf{Z}_R^{\dagger} \mathbf{B}^{-\frac{1}{2}}$) are right eigenvectors of the noiseless MP $(\mathbf{X}_0, \mathbf{X}_1)$, or equivalently, right simple eigenvectors of $\mathbf{X}_0^{\dagger}\mathbf{X}_1$ corresponding to its non-zero and simple eigenvalues. Furthermore, by definition, the columns of $\widetilde{\mathbf{Z}}_R^{\dagger}$ are right eigenvectors of the noisy MP $\mathbf{Y}_1 - \lambda\mathbf{Y}_0$, or equivalently, right eigenvectors of $\mathbf{Y}_0^{\dagger}\mathbf{Y}_1$.

As each $\{z_i\}_{i=1}^M$ is a simple eigenvalue of $\mathbf{X}_0^{\dagger}\mathbf{X}_1$, and $\mathbf{Y}_0^{\dagger}\mathbf{Y}_1$ is a perturbation of $\mathbf{X}_0^{\dagger}\mathbf{X}_1$, according to Theorem 2.3, chapter IV in \cite{stewart1990matrix}, we can uniquely match each simple eigenvector of $\mathbf{X}_0^{\dagger}\mathbf{X}_1$ with a simple eigenvector of $\mathbf{Y}_0^{\dagger}\mathbf{Y}_1$ \footnote{Theorem 2.3, chapter IV in \cite{stewart1990matrix} guarantees that there is a unique perturbed eigenvalue $\Tilde{z}_i$, such that 
\[
\Tilde{z}_i = z_i + \frac{\mathbf{p}_i^{'*}\delta(\mathbf{X}_0^{\dagger}\mathbf{X}_1)\mathbf{q}'_i}{\mathbf{p}_i^{'*}\mathbf{q}'_i} + O(\left\Vert\delta(\mathbf{X}_0^{\dagger}\mathbf{X}_1)\right\Vert^2)
\] 
where $\mathbf{q}'_i$ and $\mathbf{p}_i^{'*}$ are the corresponding right and left eigenvectors of the matrix $\mathbf{X}_0^{\dagger}\mathbf{X}_1$ and $\mathbf{Y}_0^{\dagger}\mathbf{Y}_1 = \mathbf{X}_0^{\dagger}\mathbf{X}_1 + \delta(\mathbf{X}_0^{\dagger}\mathbf{X}_1)$. }.
Without loss of generality, let the leftmost $M$ columns of $\widetilde{\mathbf{Z}}_R^{\dagger}$ consist of these perturbed eigenvectors, which are the $M$ columns of ${\mathbf{Z}}_R^{\dagger}$. Consequently, we can express this relationship as follows:

\begin{equation} \label{eq:noise right eigenvectors decomp} 
\widetilde{\mathbf{Z}}_R^{\dagger}=
\begin{bmatrix}
  \mathbf{Z}_R^{\dagger} \mathbf{B}^{-\frac{1}{2}} + \delta {\mathbf{Z}}_R^{\dagger}, \;  \mathbf{C}_1 
\end{bmatrix}_{L \times L}
\end{equation}
where $\mathbf{C}_1 \in \mathbb{C}^{L \times (L-M)}$ contains the remaining $L-M$ eigenvectors of $\mathbf{Y}_0^{\dagger}\mathbf{Y}_1$ and $\delta {\mathbf{Z}}_R^{\dagger} \in \mathbb{C}^{L \times M}$ represents the additive discrepancy in the scaled, right eigenvectors matrix $\mathbf{Z}_R^{\dagger} \mathbf{B}^{-\frac{1}{2}}$ due to the perturbation of the noiseless matrix $\mathbf{X}_0^{\dagger}\mathbf{X}_1$. Substituting \eqref{eq:noise right eigenvectors decomp} into \eqref{eq:noisy left mode}, we get the following structure of the left mode:

\begin{equation}
\begin{aligned}
\widetilde{\mathbf{Z}}_L  = (\mathbf{Z}_L\mathbf{B}^{\frac{1}{2}})(\mathbf{B}^{\frac{1}{2}}\mathbf{Z}_R)
&\begin{bmatrix}
  \mathbf{Z}_R^{\dagger} \mathbf{B}^{-\frac{1}{2}} + \delta {\mathbf{Z}}_R^{\dagger}, \;  \mathbf{C}_1 
\end{bmatrix} + \\ 
\mathbf{W}_0 
&\begin{bmatrix}
  \mathbf{Z}_R^{\dagger} \mathbf{B}^{-\frac{1}{2}} + \delta {\mathbf{Z}}_R^{\dagger}, \;  \mathbf{C}_1 
\end{bmatrix}.
\end{aligned}
\end{equation}
This can be rearranged to the following form: 
\begin{equation}
\label{eq:left mode unnormalized decomposition}
    \widetilde{\mathbf{Z}}_L = 
    \begin{bmatrix}
    \mathbf{Z}_L \mathbf{B}^{\frac{1}{2}} + \mathbf{X}_0\delta {\mathbf{Z}}_R^{\dagger} +  \mathbf{W}_0 \mathbf{Z}_R^{\dagger} \mathbf{B}^{-\frac{1}{2}} + \mathbf{W}_0\delta {\mathbf{Z}}_R^{\dagger} , \; \mathbf{Y}_0 \mathbf{C}_1
\end{bmatrix}.
\end{equation}
To obtain the desired normalized form, define the diagonal matrix
\[
\widetilde{ \mathbf{B} } = 
\begin{bmatrix} 
  \mathbf{B} & 0 \\ 
  0 & \mathbf{\mathbf{I}}_{L-M}
\end{bmatrix}_{L \times L}, 
\]
to get  
\begin{equation} \label{eq:normalized left mode} 
\widetilde{\mathbf{Z}}_L\widetilde{ \mathbf{B} }^{-\frac{1}{2}}=
\begin{bmatrix}
  \mathbf{Z}_L + \mathbf{E}_L, \; \mathbf{C}_L 
\end{bmatrix},
\end{equation}
where the noise-related component $\mathbf{E}_L$ and the spurious matrix $\mathbf{C}_L$ are defined by 
\begin{align}
\mathbf{E}_L &= \mathbf{X}_0\delta {\mathbf{Z}}_R^{\dagger}\mathbf{B}^{-\frac{1}{2}} + \mathbf{W}_0 \mathbf{Z}_R^{\dagger}\mathbf{B}^{-1}+ \mathbf{W}_0\delta {\mathbf{Z}}_R^{\dagger}\mathbf{B}^{-\frac{1}{2}}, \label{eq:E_L def}\\
\mathbf{C}_L &= \mathbf{Y}_0\mathbf{C}_1.
\end{align} For simplicity, we maintain the notation $\widetilde{\mathbf{Z}}_L$ for the normalized form, i.e., $\widetilde{\mathbf{Z}}_L = \widetilde{\mathbf{Z}}_L\widetilde{ \mathbf{B} }^{-\frac{1}{2}}$.

Following the same technique as in \eqref{eq:noise right eigenvectors decomp}, the left eigenvectors matrix can be recast as:
\begin{equation} \label{eq:noise left eigenvectors decomp}
    \widetilde{\mathbf{Z}}_L^{\dagger} =
    \begin{bmatrix}
    \mathbf{B}^{-\frac{1}{2}}\mathbf{Z}_L^{\dagger} + \delta {\mathbf{Z}}_L^{\dagger} \\ 
    \mathbf{C}_2  
    \end{bmatrix}_{L \times (N-L)},
\end{equation} 
where $\mathbf{C}_2 \in \mathbb{C}^{(L-M) \times (N-L)}$ and $\delta {\mathbf{Z}}_L^{\dagger} \in \mathbb{C}^{M \times (N-L)}$, to obtain the desired structure of:
\begin{equation}\label{eq:normalized right mode}
    \widetilde{ \mathbf{B} }^{-\frac{1}{2}}\widetilde{\mathbf{Z}}_R = 
    \begin{bmatrix}
    \mathbf{Z}_R + \mathbf{E}_R \\ 
    \mathbf{C}_R  
    \end{bmatrix},
\end{equation} where,
\begin{align}
\mathbf{E}_R &=  \mathbf{B}^{-\frac{1}{2}} \delta {\mathbf{Z}}_L^{\dagger}\mathbf{X}_0 +   \mathbf{B}^{-1}\mathbf{Z}_L^{\dagger}\mathbf{W}_0 + \mathbf{B}^{-\frac{1}{2}}\delta {\mathbf{Z}}_L^{\dagger}\mathbf{W}_0,\\
\mathbf{C}_R &= \mathbf{C}_2\mathbf{Y}_0.
\end{align} For simplicity, we maintain the notation $\widetilde{\mathbf{Z}}_R$ for the normalized form, i.e., $\widetilde{\mathbf{Z}}_R = \widetilde{ \mathbf{B} }^{-\frac{1}{2}}\widetilde{\mathbf{Z}}_R$.
Combining the normalized forms of $\widetilde{\mathbf{Z}}_L$ and $\widetilde{\mathbf{Z}}_R$ with the decomposition of the Hankel matrices $\mathbf{Y}_0$ and $\mathbf{Y}_1$ from \eqref{eq:Y0 factorization} and \eqref{eq:Y1 factorization} leads to the desired normalized factorization of $\mathbf{Y}_0= \widetilde{\mathbf{Z}}_L \widetilde{ \mathbf{B} } \widetilde{\mathbf{Z}}_R$ and $\mathbf{Y}_1=  \widetilde{\mathbf{Z}}_L\widetilde{ \mathbf{B} } \widetilde{\mathbf{\Lambda}} \widetilde{\mathbf{Z}}_R$.

The partitioning of $\widetilde{\mathbf{\Lambda}}$ to signal-related and noise-related components is derived from the assumption on the order of vectors in $\widetilde{\mathbf{Z}}_R^{\dagger}$ and $\widetilde{\mathbf{Z}}_L^{\dagger}$. We assumed, without loss of generality, that the leftmost $M$ columns of $\widetilde{\mathbf{Z}}_R^{\dagger}$ are the perturbed versions of the $M$ simple right eigenvectors of $\mathbf{X}_0^{\dagger}\mathbf{X}_1$ corresponding to the signal poles. A parallel assumption is made on the top $M$ rows of $\widetilde{\mathbf{Z}}_L^{\dagger}$. Such ordering, of the right and left eigenvectors of the matrix $\mathbf{X}_0^{\dagger}\mathbf{X}_1$, induces the same ordering on the corresponding eigenvalues matrix $\widetilde{\mathbf{\Lambda}}$. Hence the decomposition of $\widetilde{\mathbf{\Lambda}}$ in \eqref{eq:Lambda decomp}.
% Furthermore, $\mathbf{Z}_L$ or $\mathbf{Z}_R$ has exactly $M$ columns or rows, respectively, corresponding to the $M$ signal poles. It follows that each element $\lambda\in\mathbf{\Lambda}^s$ uniquely corresponds to one of the signal poles. 
This completes the proof.
\end{proof}

\section{Proof of Proposition \ref{prop:FOA of E_L}} \label{appen:first order approx of E_L}
% In this appendix, we provide a first-order approximation for the additive error term $\mathbf{E}_L$ presented in Theorem \ref{theo:MSMP}.
\begin{proof} 
The left additive error term was shown to be:
\begin{equation} \label{appen:left additive error}
    \mathbf{E}_L = \mathbf{X}_0\delta {\mathbf{Z}}_R^{\dagger}\mathbf{B}^{-\frac{1}{2}} + \mathbf{W}_0 \mathbf{Z}_R^{\dagger}\mathbf{B}^{-1}+ \mathbf{W}_0\delta {\mathbf{Z}}_R^{\dagger}\mathbf{B}^{-\frac{1}{2}}.
\end{equation} 
We begin by examining the term $\delta {\mathbf{Z}}_R^{\dagger}\in \mathbb{C}^{L \times M}$, which represents the additive discrepancy in the scaled right eigenvectors matrix of the noiseless MP pair $(\mathbf{X}_0, \mathbf{X}_1)$. 
Recalling that the matrix $\mathbf{X}_0^{\dagger}\mathbf{X}_1$ is non-defective, with its complete set of $L$ independent right and left eigenvectors denoted as $\{\mathbf{q}'_k\}_{k=1}^L$ and $\{\mathbf{p}'_k\}_{k=1}^L$, respectively. Let $\{\delta {\mathbf{q}}_i\}_{i=1}^M$ represent the $i$-th column of $\delta {\mathbf{Z}}_R^{\dagger}$. We can express $\delta {\mathbf{q}}_i$ as:
\begin{equation}
\delta {\mathbf{q}}_i =\sum\limits_{k=1}^L a_{ik}\mathbf{q}'_k = \sum\limits_{k=1}^M a_{ik}\mathbf{q}'_k +\sum\limits_{k=M+1}^L a_{ik}\mathbf{q}'_k.
\end{equation} 
Through the Vandermonde factorization of $\mathbf{X}_0$ and $\mathbf{X}_1$, we get $\mathbf{X}_0^{\dagger}\mathbf{X}_1 = \mathbf{Z}_R^{\dagger}\mathbf{Z}\mathbf{Z}_R$. 
Consequently,  for $1 \leq k \leq M$ $\mathbf{q}'_k$ is the $k$-th column of $\mathbf{Z}_R^{\dagger}$, $\mathbf{p}'^{\mathrm{H}}_k$ is the $k$-th row of $\mathbf{Z}_R$ and $\mathbf{p}'^{\mathrm{H}}_k \mathbf{q}'_k = 1$. It follows from Theorem \cite[Theorem 2.1]{hua1990matrix} that $\mathbf{q}'_k = \mathbf{q}_k$ for $1 \leq k \leq M$.
However, for $M < k \leq L$, $\mathbf{q}'_k$ lies in the null space of $\mathbf{X}_0^{\dagger}\mathbf{X}_1$, implying $\mathbf{X}_0\mathbf{q}'_k=0$. Substituting this decomposition back into \eqref{appen:left additive error}, we obtain:

\begin{equation} 
    \mathbf{E}_L^i = \frac{\mathbf{X}_0}{\sqrt{b_i}} \sum\limits_{k=1}^M a_{ik}\mathbf{q}_k
    +  \frac{\mathbf{W}_0\mathbf{q}_i}{b_i}
     +\frac{\mathbf{W}_0}{\sqrt{b_i}} \sum\limits_{k=1}^L a_{ik}\mathbf{q}'_k.
\end{equation} Using again the Vandermonde factorization of $\mathbf{X}_0 = \mathbf{Z}_L \mathbf{B}\mathbf{Z}_R$ and the fact that $\mathbf{q}_k$ is the $k$-th column of $\mathbf{Z}_R^{\dagger}$, we get:
\begin{equation} \label{eq:E_L decomp}
    \mathbf{E}_L^i =\sum\limits_{k=1}^M \frac{a_{ik}b_k}{\sqrt{b_i}} \begin{bmatrix}
           1 \\
           z_k \\
           \vdots \\
           z_k^{N-L-1}
         \end{bmatrix}
    + \frac{\mathbf{W}_0\mathbf{q}_i}{b_i}
    + \frac{\mathbf{W}_0}{\sqrt{b_i}} \sum\limits_{k=1}^L a_{ik}\mathbf{q}'_k.
\end{equation}

We now turn to get a first-order approximation of the coefficients $a_{ik}$ using Lemma \ref{lemma:first order approximation of the vector of coefficients} applied to the matrix $\mathbf{X}_0^{\dagger}\mathbf{X}_1$, its perturbation $\mathbf{Y}_0^{\dagger}\mathbf{Y}_1$ = $\mathbf{X}_0^{\dagger}\mathbf{X}_1$ + $\delta(\mathbf{X}_0^{\dagger}\mathbf{X}_1)$, and to the simple, scaled, right eigenvectors $\mathbf{Z}_R^{\dagger} \mathbf{B}^{-\frac{1}{2}}$ of $\mathbf{X}_0^{\dagger}\mathbf{X}_1$.

Lemma \ref{lemma:first order approximation of the vector of coefficients} assumes a specific normalization for the right, simple eigenvectors, and their perturbation. As noted above, for each simple eigenvector $\{\mathbf{q}'_k\}_{k=1}^M$ it holds that $\mathbf{p}'^{\mathrm{H}}_k \mathbf{q}'_k = 1$. Additionally, by our Assumption \ref{assum:eigenvectors normalization}, we have $\mathbf{p}'^{\mathrm{H}}_k \widetilde{\mathbf{q}}'_k = 1$.
Finally, by Assumption \eqref{eq:spectral radius}, $\rho(\mathbf{D}_i\mathbf{T}_L \delta (\mathbf{X}_0^{\dagger}\mathbf{X}_1)\mathbf{T}_R)<1$, and hence we can apply Lemma \ref{lemma:first order approximation of the vector of coefficients} to get the first order approximation of the coefficients $a_{ik}$:
\[
a_{ik}=
\begin{cases}
0  &,k=i\\
\frac{\mathbf{p}'^{\mathrm{H}}_k \delta (\mathbf{X}_0^{\dagger}\mathbf{X}_1) \mathbf{q}'_i}{\sqrt{b_i}(\Tilde{\lambda}_i -\lambda_k)} &, \text{otherwise.}
\end{cases}
\] 
% where $\mathbf{q}_i$ is the $i$-th column of $\mathbf{T}$ and $\mathbf{u}_k^{\mathrm{H}}$ is the $k$-th row of $\mathbf{T}^{-1}$.
Note that $\lambda_k$ are the eigenvalues of $\mathbf{X}_0^{\dagger}\mathbf{X}_1$, which are assumed to be indexed such that: 
\begin{equation*}
{\lambda_k} =
\begin{cases}
z_k,& \text{if}\  1\leq k \leq M\\
0,&{\text{otherwise,}}
\end{cases}
\end{equation*}
and $\Tilde{\lambda}_i$ is the unique perturbation of $\lambda_i$ as described in Lemma \ref{lemma:first order approximation of the vector of coefficients}. As noted above, $\mathbf{q}'_i = \mathbf{q}_i$ for $1 \leq i \leq M$ and hence:
\[
a_{ik}=
\begin{cases}
0  &,k=i\\
\frac{\mathbf{p}'^{\mathrm{H}}_k \delta (\mathbf{X}_0^{\dagger}\mathbf{X}_1) \mathbf{q}_i}{\sqrt{b_i}(\Tilde{\lambda}_i -\lambda_k)} &, \text{otherwise.}
\end{cases}
\] 

Following the approach presented in \cite[section IV]{hua1990matrix}, the term $\mathbf{p}'^{\mathrm{H}}_k \delta (\mathbf{X}_0^{\dagger}\mathbf{X}_1) \mathbf{q}_i$, can be approximated to first order by
\begin{equation} \label{app:coef approx}
\mathbf{p}'^{\mathrm{H}}_k \delta (\mathbf{X}_0^{\dagger}\mathbf{X}_1) \mathbf{q}_i = \mathbf{p}'^{\mathrm{H}}_k (\delta (\mathbf{Y}_0^{\dagger})\mathbf{X}_1) \mathbf{q}_i + \mathbf{p}'^{\mathrm{H}}_k (\mathbf{X}_0^{\dagger}\mathbf{W}_1) \mathbf{q}_i.
\end{equation} 
For $1 \leq k \leq M$ it is also shown in  \cite{hua1990matrix} that a first order approximation of \eqref{app:coef approx} is given by:
\begin{equation}\label{eq:coeff explcit pprox}
\mathbf{p}'^{\mathrm{H}}_k \delta (\mathbf{X}_0^{\dagger}\mathbf{X}_1) \mathbf{q}_i = \frac{\mathbf{p}_k^{\mathrm{H}} (\mathbf{W}_1-z_i\mathbf{W}_0) \mathbf{q}_i}{b_k},
\end{equation} 
where $\mathbf{p}_k^{\mathrm{H}}$ is the $k$-th  row of $\mathbf{Z}_L^{\dagger}$, as described in \cite[Theorem 2.1]{hua1990matrix}.
On the other hand, for $M < k \leq L$, the rightmost term of (\ref{app:coef approx}) annihilates and we are left with $\mathbf{p}'^{\mathrm{H}}_k (\delta (\mathbf{Y}_0^{\dagger})\mathbf{X}_1) \mathbf{q}_i$.
Using the decomposition theorem for pseudo-inverses, given in \cite{wedin1973perturbation}, we can write:
\begin{equation}\label{eq:pseudo-inve decomp}
\begin{aligned}
\delta \mathbf{Y}_0^{\dagger} &= -\mathbf{Y}_0^{\dagger}\mathbf{W}_0\mathbf{X}_0^{\dagger} + (\mathbf{Y}_0^{\mathrm{H}}\mathbf{Y}_0)^{\dagger}\mathbf{W}_0^{\mathrm{H}}P_{N(\mathbf{X}_0^{\mathrm{H}})}\\
&+ P_{N(\mathbf{Y}_0)}\mathbf{W}_0^{\mathrm{H}}(\mathbf{X}_0\mathbf{X}_0^{\mathrm{H}})^{\dagger},
\end{aligned} 
\end{equation} where $P_X$ is the orthogonal projection onto a subspace $X$. 

Combining \eqref{eq:coeff explcit pprox}  and \eqref{eq:pseudo-inve decomp} back into \eqref{eq:E_L decomp}, a first order approximation of $\mathbf{E}_L^i$ is given by:
\begin{equation}
\mathbf{E}_L^i  = \sum\limits_{\substack{k=1\\k \neq i}}^M \frac{\mathbf{p}_k^{\mathrm{H}} (\mathbf{W}_1-z_i\mathbf{W}_0) \mathbf{q}_i}{b_i(\Tilde{\lambda}_i -z_k)} \begin{bmatrix}
           1 \\
           z_k \\
           \vdots \\
           z_k^{N-L-1}
         \end{bmatrix} + \frac{\mathbf{W}_0\mathbf{q}_i}{b_i}.
\end{equation}
\end{proof}

\begin{remark}\label{remark: sep_pert balance}
    Assumption \eqref{eq:spectral radius} intricately implies a delicate balance between the $i$-th pole separation and perturbation. Two observations support this. (i) Recognizing that matrix $\mathbf{D}_i$ contains information about the separation and perturbation of the $i$-th pole as:
    \[
    (\Tilde{\lambda}_i -\lambda_k) = (\Tilde{\lambda}_i - \lambda_i) - (\lambda_k - \lambda_i).
    \] 
    (ii) Referring to \cite[Theorem 3.3, Chapter IV]{stewart1990matrix}, the following inequality holds:
    \begin{equation} 
        \max_i \min_k |\Tilde{\lambda}_i - \lambda_k| \leq \left\Vert\mathbf{T}_R^{-1}\delta (\mathbf{X}_0^{\dagger}\mathbf{X}_1)\mathbf{T}_R\right\Vert,
    \end{equation} 
    and hence the inequality:
    \begin{equation} \label{eq:spectral variation inequality}
        1 \leq \frac{\left\Vert\mathbf{T}_R^{-1}\delta (\mathbf{X}_0^{\dagger}\mathbf{X}_1)\mathbf{T}_R \right\Vert}{\min\limits_k |\Tilde{\lambda}_i - \lambda_k|}
    \end{equation} holds for all indices $1 \leq i \leq M$. 

    Equations \eqref{eq:spectral radius} and \eqref{eq:spectral variation inequality} are related through the connection between the spectral radius and spectral norm, but differ in two aspects: (i) the $i$-th coordinate of the matrix $\mathbf{D}_i$ is set to zero \footnote{This motivates the eigenvector normalization in Assumption \ref{assum:eigenvectors normalization}.}, and (ii) the matrix norm is used, providing a natural upper bound on the spectral radius.
\end{remark}

\section{Proof of Proposition \ref{prop:bounds of noise-related terms}} \label{app:bounds of noise-related terms}

% In this appendix, we analyze the noise-related terms introduced in \eqref{eq:noise-terms defs}, and provide the bounds in \eqref{eq:gamma_ik upper bound} and \eqref{eq:xi_i upper bound}, under the conditions of Proposition \ref{prop:FOA of E_L} and Assumption \ref{eq:perturb-separation assumption of prop 2}.
\begin{proof}
Recall that: 
\begin{equation}\label{eq:noise-terms defs in proof}
\gamma_{i,m} = \frac{\mathbf{u}_m^{\mathrm{T}} \mathbf{Q}_i \mathbf{w}}{(\widetilde{\lambda}_i-z_m)b_i},\;\; \boldsymbol\xi_i = \frac{\mathbf{I}_0\mathbf{Q}_i \mathbf{w}}{b_i},
\end{equation} 
and for simplicity, denote: 
\begin{equation}\label{eq:gamma_bar}
\bar{\gamma}_{i,m} = \frac{\mathbf{u}_m^{\mathrm{T}} \mathbf{Q}_i \mathbf{w}}{b_i}.
\end{equation}
Equation \eqref{eq:noise-terms defs in proof} shows that $\bar{\gamma}_{i,m}$ and the entries of $\boldsymbol\xi_i$ are a linear combination of the complex, i.i.d Gaussian random variables $\{\mathbf{w}(n)\}_{n=1}^n$ with zero mean and variance $\sigma_w^2$. Hence, $\bar{\gamma}_{i,m}$ and the entries of $\boldsymbol\xi_i$ are also complex, Gaussian random variables with zero mean. 
Inspecting the circular structure of the matrix $\mathbf{Q}_i$, we can recast the term $\mathbf{u}_m \mathbf{Q}_i$ as the convolution of the vectors $\mathbf{q}_i$ and $\mathbf{u}_m$. Using \eqref{eq:gamma_bar}, we derive the following statistics of $\bar{\gamma}_{i,m}$: 

\begin{align}
    \mathbb{E}[\bar{\gamma}_{i,m}] &= 0,\\
    \text{Var}(\bar{\gamma}_{i,m}) &=  \frac{\left\Vert \mathbf{u}_m * \mathbf{q}_i \right\Vert_2^2}{\text{SNR}_i}. \label{eq:gamma variance}
\end{align} 

It is worth recalling that the vector $\mathbf{u}_m$ was defined in \eqref{eq:u_m def} by
\[
    \mathbf{u}_m^{\mathrm{T}} = \begin{bmatrix} 0,\; \mathbf{p}_m^{\mathrm{H}} \end{bmatrix} -z_i \begin{bmatrix} \mathbf{p}_m^{\mathrm{H}}, \;0 \end{bmatrix},
\] where $\mathbf{p}_m^{\mathrm{H}}$ and $\mathbf{q}_i$ depend solely on the signal poles and the parameters $N$ and $L$, as they are also given by the $m$-th row of $\mathbf{Z}_L^{\dagger}$ and $i$-th column of $\mathbf{Z}_R^{\dagger}$, respectively. 

Since $\bar{\gamma}_{i,m}$ is a complex Gaussian random variable, $|\bar{\gamma}_{i,m}|$ follows a Rayleigh distribution with scale parameter $\sigma_{\bar{\gamma}_{i,m}}=\text{var}(\bar{\gamma}_{i,m})$. By setting $t = \sigma_{\bar{\gamma}_{i,m}} \sqrt{2\text{log}(\frac{1}{\varepsilon})}$, we ensure that  $|\bar{\gamma}_{i,m}| \leq t$ with a probability of at least $1-\varepsilon$.
Utilizing \eqref{eq:gamma variance}, it follows that: 
\[
|\bar{\gamma}_{i,m}| \leq  \sqrt{\frac{2\text{log}(\frac{1}{\varepsilon})}{\text{SNR}_i}}\left\Vert \mathbf{u}_m *\mathbf{q}_i\right\Vert_2, 
\] with a probability of at least $1- \varepsilon$.

To derive the bound on $\gamma_{i,m}$, we use our assumption that $\left|z_i-z_m\right| \geq 2 | \delta z_i |$ for $m \neq i$, to get:
\[
|\widetilde{\lambda}_i-z_m| = |z_i + \delta z_i - z_m| \geq |z_i - z_m| -  | \delta z_i| \geq \frac{|z_i-z_m|}{2}.
\]
Finally, we derive that:
\[
|\gamma_{i,m}| \leq  2\sqrt{\frac{2\text{log}(\frac{1}{\varepsilon})}{\text{SNR}_i}} \frac{\left\Vert \mathbf{u}_m *\mathbf{q}_i\right\Vert_2}{|z_i-z_m|}, 
\] with a probability of at least $1- \varepsilon$.

The entries of $\boldsymbol\xi_i$ can be bounded similarly. Writing the $k$-th entry of $\boldsymbol\xi_i$ by $\boldsymbol\xi_i(k) = \mathbf{e}_k^{\mathrm{T}}\boldsymbol\xi_i$, where $\mathbf{e}_k$ is the $k$-th standard unit vector we get: 
\[
\boldsymbol\xi_i(k) = \frac{\mathbf{e}_k^{\mathrm{T}}\mathbf{I}_0\mathbf{Q}_i \mathbf{w}}{b_i},
\] and we derive the following statistics of $\boldsymbol\xi_i(k)$:

\begin{align}
    \mathbb{E}[\boldsymbol\xi_i(k)] &= 0,\\
    \text{Var}(\boldsymbol\xi_i(k)) &=  \frac{\left\Vert \mathbf{e}_k^{\mathrm{T}} \mathbf{I}_0 \mathbf{Q}_i \right\Vert_2^2}{\text{SNR}_i}. \label{eq:xi variance}
\end{align} 

Re-casting $\mathbf{e}_k^{\mathrm{T}} \mathbf{I}_0 \mathbf{Q}_i$ as a convolution of two vectors and using Young's convolution inequality, we achieve the following bound:
\begin{equation}
    \text{Var}(\boldsymbol\xi_i(k)) \leq \frac{1}{\text{SNR}_i}.
\end{equation}
It follows that: 
\[
\left| \boldsymbol\xi_i(k) \right| \leq  \sqrt{\frac{2\text{log}(\frac{1}{\varepsilon})}{\text{SNR}_i}}, 
\] with a probability of at least $1- \varepsilon$.

\end{proof}

% \section{Proof of Proposition \ref{prop:bound of DFT of xi}} \label{app:bound of DFT of xi}
% % In this Appendix, we provide a probabilistic bound on the DFT of the noise factor in \eqref{eq:stringent conditions for single peak in DFT} for sufficiently high $\text{SNR}_i$.

% \begin{proof}
% By definition \eqref{eq:noise-terms defs}, the entries of $\boldsymbol\xi_i$ are Gaussian with zero mean and variance smaller than $\frac{1}{\text{SNR}_i}$ (see the proof of Proposition \ref{prop:bounds of noise-related terms}). It follows that the entries of $\mathcal{F}(\boldsymbol\xi_i)$ are also complex Gaussian with zero mean and variance smaller than $\frac{N-L}{\text{SNR}_i}$. Consequently, $\left|\mathcal{F}(\boldsymbol\xi_i)[k]\right|$ follows Rayleigh distribution, and we obtain the following bound:
% \begin{equation}\label{eq:upper bound on noise term in DFT}
%     \left|\mathcal{F}(\boldsymbol\xi_i)[k]\right| \leq \sqrt{\frac{2\text{log}(\frac{1}{\varepsilon})(N-L)}{\text{SNR}_i}},
% \end{equation} with a probability of at least $1- \varepsilon$. Hence, if the $\text{SNR}_i$ is sufficiently high:
% \[
% \frac{2\text{log}(\frac{1}{\varepsilon})}{N-L} \ll \text{SNR}_i,
% \] 
% the condition on the DFT of the noise factor in \eqref{eq:stringent conditions for single peak in DFT} is met, i.e.,
% \[
% \left|\mathcal{F}(\boldsymbol\xi_i)[k]\right| \ll N-L,
% \] with a probability of at least $1-\varepsilon$.    
% \end{proof}

\section{Local perspectives for signal modes detection}\label{app: Local perspectives for signal modes detection}
To exploit the obscured multiplicative Vandermonde structure in \eqref{eq:signal_mode}, we could examine the following ratio between consecutive entries of $\widetilde{\mathbf{Z}}_L^i$:
\begin{equation}\label{eq:ratios of cosecutive rows in signal mode}
    \hat{\mathbf{z}}_i(k) := \frac{ \widetilde{\mathbf{Z}}_L^i(k+1)}{ \widetilde{\mathbf{Z}}_L^i(k)}= \frac{z_i^{k} + \mathbf{E}_L^i(k+1)}{z_i^{k-1} + \mathbf{E}_L^i(k)},
\end{equation} for $1 \leq k \leq N-L-1$.

Informally, when $|\mathbf{E}_L^i(k)|$ and $|\mathbf{E}_L^i(k+1)|$ are sufficiently small, the ratio $\hat{\mathbf{z}}_i(k)$ in \eqref{eq:ratios of cosecutive rows in signal mode} can be approximated by:
\begin{equation}
\begin{aligned}\label{eq:approx ratios of consecutive rows in signal mode}
    \hat{\mathbf{z}}_i(k) &= \frac{z_i^{k} + \mathbf{E}_L^i(k+1)}{z_i^{k-1} + \mathbf{E}_L^i(k)} 
    = z_i \left(\frac{ 1 + \frac{ \mathbf{E}_L^i(k+1)}{z_i^{k}} } { 1 + \frac{\mathbf{E}_L^i(k)}{z_i^{k-1}} }\right)\\
    &\approx z_i \left(1 + \frac{\mathbf{E}_L^i(k+1)}{z_i^{k}} -\frac{\mathbf{E}_L^i(k)}{z_i^{k-1}}\right)
\end{aligned} 
\end{equation} using standard linear approximation techniques.

Equation \eqref{eq:approx ratios of consecutive rows in signal mode} shows that the ratios between consecutive entries of the $i$-th signal mode are centered around a \emph{constant} value, which is the $i$-th signal pole $z_i$. 
%
% Consequently, to distinguish between signal and noise modes, we could use the following measure:
% \begin{equation}\label{eq:ratios epsilon measure}
%     \epsilon_i = \text{Var}(\hat{\mathbf{z}}_i) = \frac{1}{N-L-2}\sum\limits_{k=1}^{N-L-1}\left| \hat{\mathbf{z}}_i(k) - \mu_{\hat{\mathbf{z}}_i}\right|^2
% \end{equation}
% where $1 \leq i \leq L$ and:
% \begin{equation}\label{eq:empirical avarage def}
%      \mu_{\hat{\mathbf{z}}_i} :=  \frac{1}{N-L-1}\sum\limits_{k=1}^{N-L-1} \hat{\mathbf{z}}_i(k).
% \end{equation}
% Small values of $\epsilon_i$ indicate that $\widetilde{\mathbf{Z}}_L^i$ is a signal mode. Conversely, as noise modes lack a particular structure, large values of $\epsilon_i$ indicate that $\widetilde{\mathbf{Z}}_L^i$ is a noise mode.
%
% We note that there are other alternative implementations for discerning between signal and noise modes, leveraging the near multiplicative structure of the signal modes.
In the context of dynamic mode decomposition (DMD) for data-driven dynamical systems analysis, Bronstein et al. \cite{bronstein2022spatiotemporal} proposed to use the features: $\epsilon_i = |\widetilde{\lambda}_i-\mu_{\hat{\mathbf{z}}_i}|$, which can now be supported by the first-order approximation of $\widetilde{\lambda}_i$, given by \cite{hua1990matrix}:
\[
\widetilde{\lambda}_i \cong  z_i +\frac{\mathbf{u}_i \mathbf{Q}_i \mathbf{w} }{b_i}.
\]

{
\section{Other alternatives for amplitude estimation}\label{app: amlitude estimation alternatives}

In Section~\ref{sec: eff amp est}, we introduced one approach that exploits the mode structure to efficiently estimate the amplitudes.
Here, we present two additional alternatives and compare their performance with the method of Section~\ref{sec: eff amp est}.

By Theorem \ref{theo:MSMP}, Appendix \ref{app:Extending the MP Theorem: Accounting for Noise}, the un-normalized signal modes defined in \eqref{def:right mp mode} and \eqref{def:left mp mode} can be recast as:
\begin{align*}
\widetilde{\mathbf{Z}}_{L,un}^{sig} &= 
    \mathbf{Z}_L\mathbf{B}_L + \mathbf{E}_L\mathbf{B}_L,\\
    \widetilde{\mathbf{Z}}_{R,un}^{sig} &= \mathbf{B}_R\mathbf{Z}_R +  \mathbf{B}_R\mathbf{E}_R,
\end{align*}
where $\mathbf{B}_L,\mathbf{B}_R \in \mathbb{C}^{M \times M}$ are diagonal matrices satisfying $\mathbf{B}_L \mathbf{B}_R = \mathbf{B}$.
Consequently,
\begin{equation}
    \widetilde{\mathbf{Z}}_{L,un}^{sig} \widetilde{\mathbf{Z}}_{R,un}^{sig} =  \mathbf{Z}_L\mathbf{B} \mathbf{Z}_R + \mathbf{Z}_L\mathbf{B}\mathbf{E}_R + \mathbf{E}_L\mathbf{B}\mathbf{Z}_R + \mathbf{E}_L\mathbf{B}\mathbf{E}_R,
\end{equation}
which is viewed as a perturbed version of $\mathbf{Z}_L\mathbf{B}\mathbf{Z}_R$, as $\Vert\mathbf{E}_L\Vert, \Vert\mathbf{E}_R\Vert \to 0$ for $\sigma_w^2\to 0$ by Corollary 1. This motivates the constrained LS problem:
\[
\widehat{\mathbf{b}} = \arg\min_{\mathbf{b}\in\mathbb{C}^{M}}
\left\Vert\widetilde{\mathbf{Z}}_{L,un}^{sig}\widetilde{\mathbf{Z}}_{R,un}^{sig}
-\mathbf{Z}_L\text{diag}(\mathbf{b})\mathbf{Z}_R\right\Vert_F^2,
\]
where $\mathbf B=\text{diag}(\mathbf b)$, whose closed-form solution is:
\[
    \widehat{\mathbf b} = \left(\mathbf Z_R^{T}\diamond \mathbf Z_L\right)^{\dagger}\mathrm{vec}(\widetilde{\mathbf{Z}}_{L,un}^{sig} \widetilde{\mathbf{Z}}_{R,un}^{sig}),
\]
where $\diamond$ denotes the columnwise Kronecker product. 

In practice, $\mathbf{Z}_L$, and $\mathbf{Z}_R$, can only be approximated, and a straightforward option is to approximate them by constructing a Vandermonde matrix from the estimated poles:
\[
    \widehat{\mathbf{Z}}_L = \mathbf{V}_{N-L}(\widetilde{\mathbf{\Lambda}}), \quad \widehat{\mathbf{Z}}_R = \mathbf{V}^{\mathrm{T}}_{L}(\widetilde{\mathbf{\Lambda}})
\]
where $\mathbf{V}_{N-L}(\widetilde{\mathbf{\Lambda}}) \in\mathbb{C}^{(N-L)\times L}$ and $\mathbf{V}_{L}(\widetilde{\mathbf{\Lambda}})\in\mathbb{C}^{L\times L}$ are Vandermonde matrices of order $N-L$ and $L$, respectively, built from the $L$ estimated poles $\text{diag}(\{\widetilde{\lambda}_i\}_{i=1}^L) = \widetilde{\mathbf{\Lambda}}$. The corresponding amplitude-estimator is then approximated by:
\begin{equation}\label{eq: approx constrained solution}
    \widehat{\mathbf b} \approx \left(\mathbf{V}_{L}(\widetilde{\mathbf{\Lambda}})\diamond \mathbf{V}_{N-L}(\widetilde{\mathbf{\Lambda}})\right)^{\dagger}\mathrm{vec}(\widetilde{\mathbf{Z}}_{L,un} \widetilde{\mathbf{Z}}_{R,un}).
\end{equation}
This estimator requires forming $\mathbf{V}_{L}(\widetilde{\mathbf{\Lambda}})\diamond \mathbf{V}_{N-L}(\widetilde{\mathbf{\Lambda}})$ and computing its pseudo-inverse, in addition to matrix-vector multiplication, which requires $\mathcal{O}((N-L)(L^3+3L^2))$.

A more efficient, yet suboptimal solution is to consider the following unconstrained LS problem:
\[
\widehat{\mathbf{B}} = \arg\min_{\mathbf{B}\in\mathbb{C}^{M\times M}}
\left\Vert\widetilde{\mathbf{Z}}_{L,un}^{sig}\widetilde{\mathbf{Z}}_{R,un}^{sig}
-\mathbf{Z}_L\mathbf{B}\mathbf{Z}_R\right\Vert_F^2,
\]
whose closed-form minimizer is:
\begin{equation}\label{eq: LS minimizer}
\widehat{\mathbf{B}} = \mathbf{Z}_L^{\dagger}\left(\widetilde{\mathbf{Z}}_{L,un}^{sig}\widetilde{\mathbf{Z}}_{R,un}^{sig}\right)\mathbf{Z}_R^{\dagger},
\end{equation}
where we later enforced diagonality by retaining only the diagonal part of $\widehat{\mathbf{B}}$.
Similarly to \eqref{eq: approx constrained solution}, the amplitude estimator corresponding to \eqref{eq: LS minimizer} is then approximated by: 
\begin{equation} \label{eq: amp est alternative 1}
    \widehat{\mathbf{B}} \approx \mathbf{V}_{N-L}(\widetilde{\mathbf{\Lambda}})^{\dagger} \widetilde{\mathbf{Z}}_{L,un}\widetilde{\mathbf{Z}}_{R,un}\mathbf{V}^{\mathrm{T}}_{L}(\widetilde{\mathbf{\Lambda}})^{\dagger}.
\end{equation}
% This option fully utilizes all the columns and rows of the modes, in addition to the entire set of estimated poles. 
This estimator involves the computation and pseudo-inverse of $\mathbf{V}_{N-L}(\widetilde{\mathbf{\Lambda}})$ and $\mathbf{V}_{L}(\widetilde{\mathbf{\Lambda}})$, in addition to multiplication of four matrices, which requires a total of $\mathcal{O}(5L^3 + L^2 + (N-L)(3L^2+L))$ operations.
% which require $\mathcal{O}((N-L)(L^2+L))$, and $\mathcal{O}(L^3 + L^2)$ operations, respectively. Additionally, the multiplication of the four involved matrices require $\mathcal{O}(4L^3 + 2L^2(N-L))$ operations. Hence, a total of $\mathcal{O}(5L^3 + L^2 + (N-L)(3L^2+L))$ operations is required.

A different approach to utilize the modes structure is to compute their element-wise product:
\begin{equation}\label{eq: element-wise multiplication of the modes matrices to est amps}
    \begin{aligned}
            &\widetilde{\mathbf{Z}}_{L,un}^i \odot {\widetilde{\mathbf{Z}}_{R,un}^i} =\\ &b_i\left( \mathbf{V}_L(z_i^2) + \mathbf{V}_L(z_i) \odot \mathbf{E}_R^i + \mathbf{E}_L^i \odot \mathbf{V}_L(z_i) + \mathbf{E}^i_L \odot \mathbf{E}^i_R\right),
    \end{aligned}
\end{equation}
where $\odot$ is the element-wise multiplication, and $ \widetilde{\mathbf{Z}}_{L,un}^i$ is truncated to length $L$ to fit the length of $\widetilde{\mathbf{Z}}_{R,un}^i$. 
By Corollary 1, the noise terms vanish as $\sigma_w^2 \to 0$, and thus 
\[
\widetilde{\mathbf{Z}}_{L,un}^i \odot {\widetilde{\mathbf{Z}}_{R,un}^i} \cong b_i \mathbf{V}_L(z_i^2).
\]
Consequently, $b_i$ can be estimated by correcting $\widetilde{\mathbf{Z}}_{L,un}^i \odot {\widetilde{\mathbf{Z}}_{R,un}^i}$ with $\mathbf{V}_L(z_i^{-2})$. In practice, $z_i$ is approximated by the $i$-th MP pole $\widetilde{\lambda}_i$, and the $i$-th amplitude estimation is achieved by averaging the resulting vector:
\begin{equation}\label{eq: amp est alternative 2}
    \hat{b}_i = \text{mean}(\widetilde{\mathbf{Z}}_{L,un}^i \odot \widetilde{\mathbf{Z}}_{R,un}^i \odot \mathbf{V}_L(\widetilde{\lambda}_i^{-2})).
\end{equation} 
This estimator involves the element-wise multiplication of three vectors of size $L$, and the computation of the vector $\mathbf{V}_L(\widetilde{\lambda}_i^{-2})$. This requires $\mathcal{O}(4L)$ operations for each $1\leq i \leq L$, and hence a total of $\mathcal{O}(4L^2)$ operations.

\subsubsection*{\textbf{Quantitative Results}}
We compare four amplitude-estimation methods in Figures~\ref{fig: amps est comparison with alternatives INTERVAL} and \ref{fig: amps est comparison with alternatives SNR}: (i) the proposed method from Section~V (``SAMP''); (ii) the first alternative estimator in \eqref{eq: amp est alternative 1} (``alt~1''); (iii) the second alternative estimator in \eqref{eq: amp est alternative 2} (``alt~2''); and (iv) the standard MP estimator described in Section~\ref{subsec: classic MP params est}. 
In addition, we report genie-aided benchmarks for ``alt~1'' and ``alt~2'', in which the true poles $\{z_i\}_{i=1}^M$ are used in place of the estimated poles $\{\widetilde{\lambda}_i\}_{i=1}^M$. This isolates the effect of pole-estimation errors from that of noise-induced terms in the noisy modes. 
These genie-aided benchmarks serve as oracle baselines and are denoted by ``orc~1'' and ``orc~2'', respectively.
In our simulations, the estimator in \eqref{eq: approx constrained solution} achieves comparable accuracy to the MP and SAMP estimators, but at a substantially higher runtime. We therefore omit it for brevity.

Fig. \ref{fig: amps est comparison with alternatives INTERVAL} displays the amplitude estimation error for each method as a function of the computation time for different signal lengths. The error is defined as the average NMSE between each of the true amplitudes and the estimated amplitudes corresponding to the best $M$ estimated poles. 
We see in the figure that SAMP and alt~2 achieve comparable, and the lowest, NMSE values, while requiring approximately $6$ times and $3.5$ times less computation time, respectively, compared to the standard MP estimator.
Furthermore, alt~1 has the highest computational cost, in alignment with the theoretical analysis above, and also yields the worst NMSE. Moreover, comparison with its genie-aided counterpart (orc~1) emphasizes that using the estimated poles provides additional information, but at the expense of introducing pole-estimation errors. This statement is further supported by comparing alt~2 with its genie-aided counterpart (orc~2).

Fig. \ref{fig: amps est comparison with alternatives SNR} displays the amplitude estimation error for each method as a function of the SNR. The error is defined as the average NMSE between each of the true amplitudes and the estimated amplitudes corresponding to the best $M$ estimated poles. 
The results indicate that SAMP attains the lowest, or comparable, RMSE across all simulation settings. 
Moreover, while alt~2 is competitive with SAMP in Fig.~\ref{fig: amps est comparison with alternatives INTERVAL}, it degrades rapidly as SNR decreases, in contrast to SAMP. Finally, the performance gaps between alt~1 and alt~2 and their genie-aided counterparts indicate again that using estimated poles provides additional information, but at the expense of introducing pole-estimation errors.}\vspace{-2mm}

\begin{figure}[t]
    \centering
    \begin{subfigure}[b]{0.49\columnwidth}
        \centering
        \includegraphics[width=\linewidth]{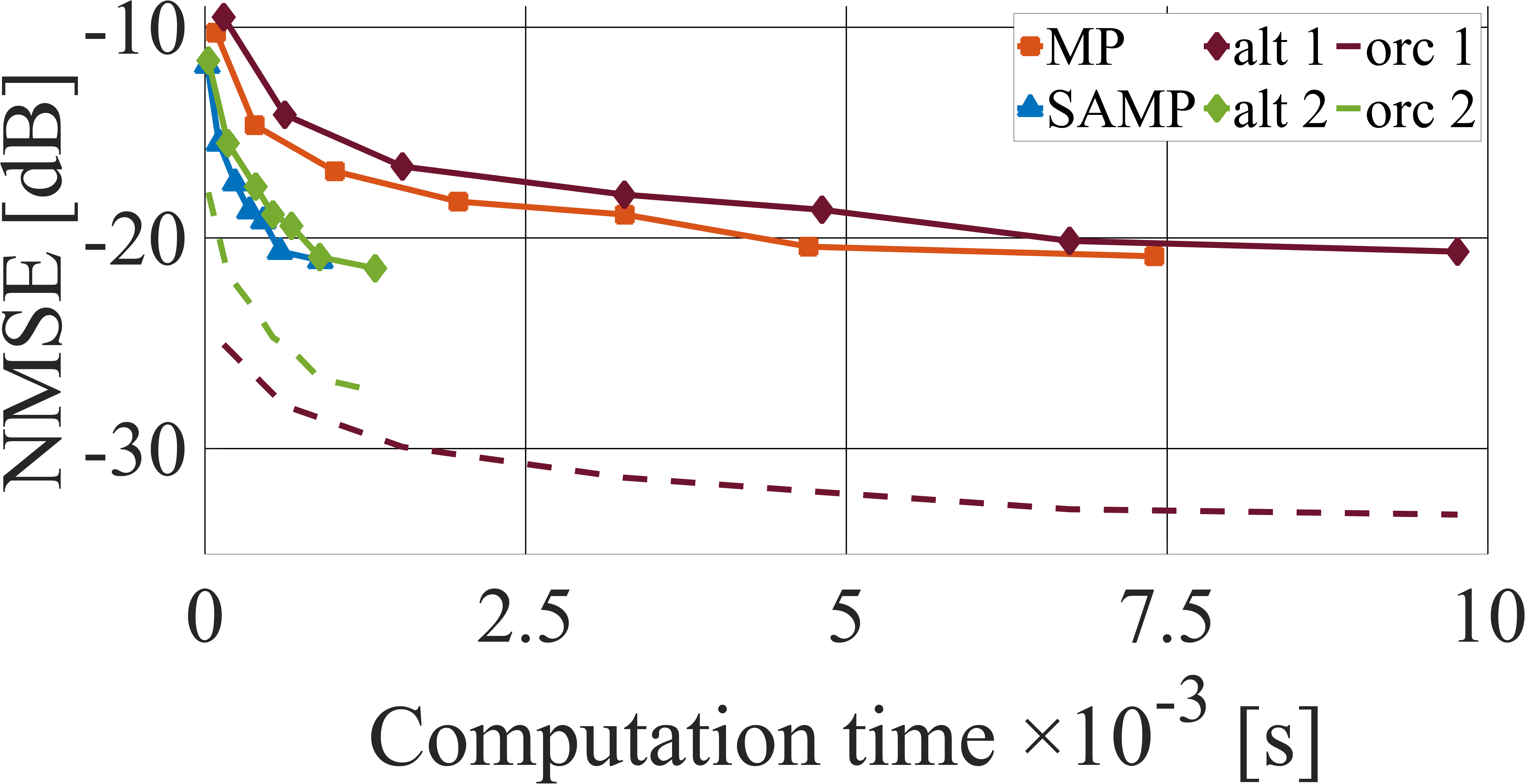}
        \caption{$M=2$, undamped}
    \end{subfigure}
    \begin{subfigure}[b]{0.49\columnwidth}
        \centering
        \includegraphics[width=\linewidth]{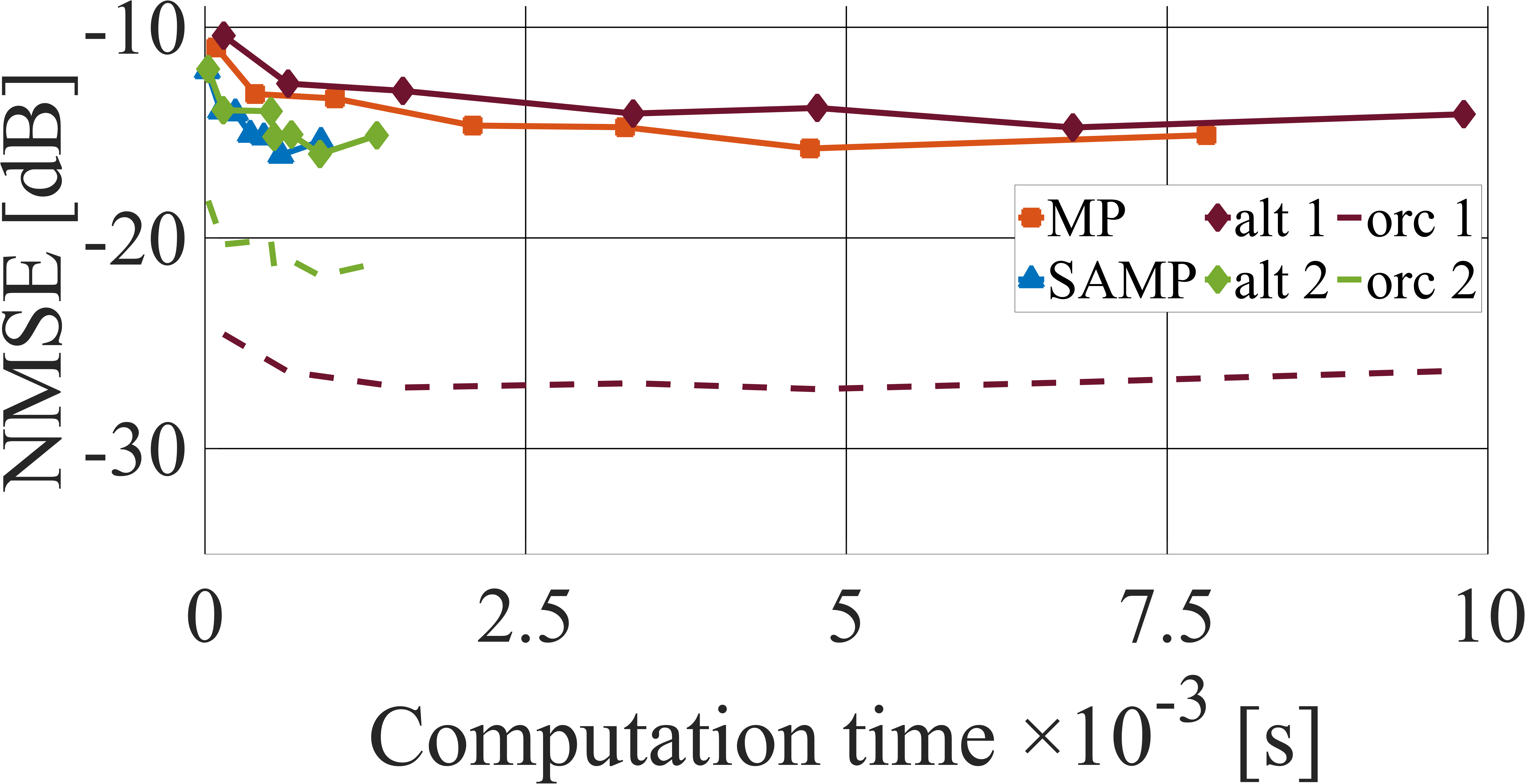}
        \caption{$M=2$, damped}
    \end{subfigure}
\hfill
    \begin{subfigure}[b]{0.49\columnwidth}
        \centering
        \includegraphics[width=\linewidth]{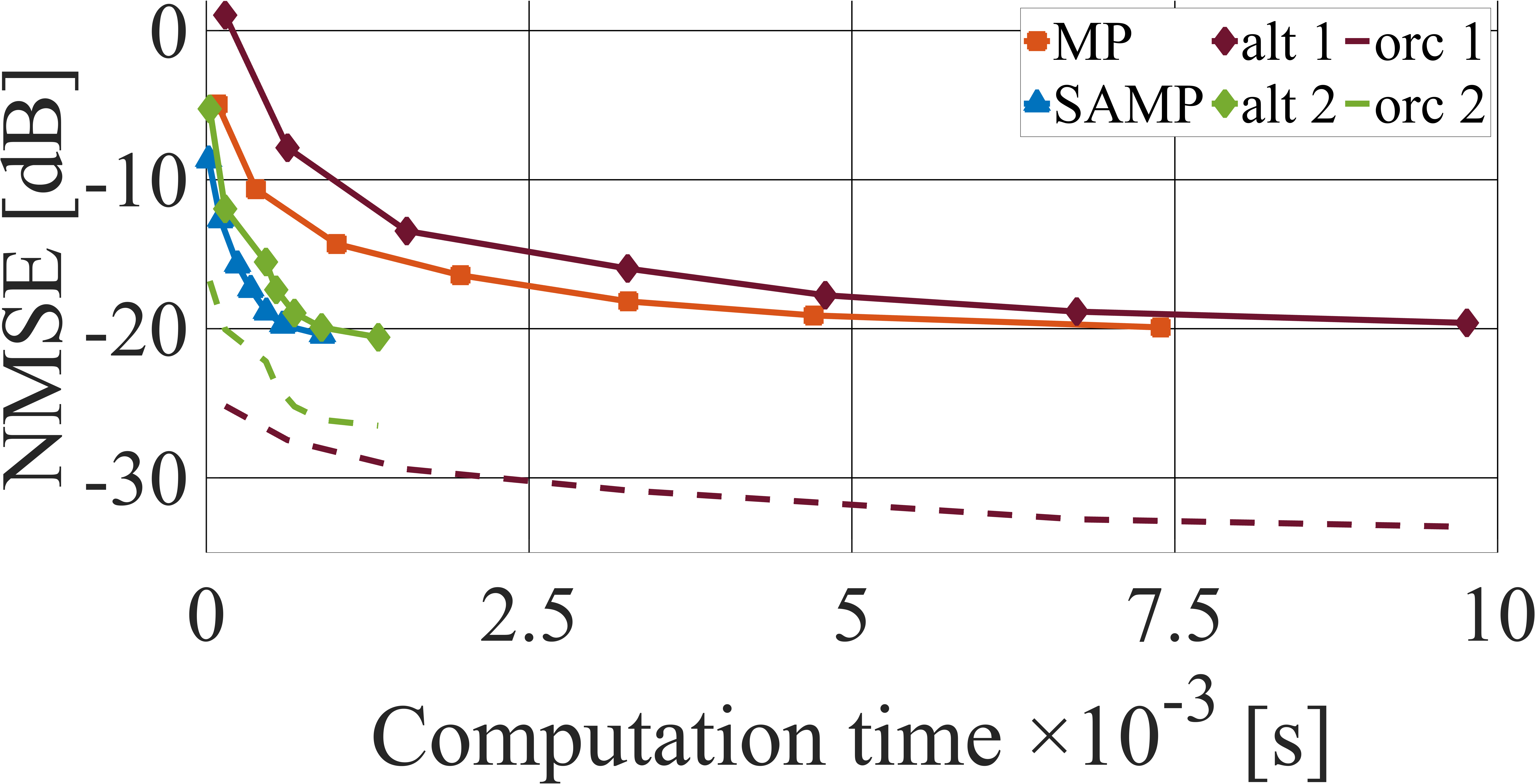}
        \caption{$M=4$, undamped}
    \end{subfigure}
    \begin{subfigure}[b]{0.49\columnwidth}
        \centering
        \includegraphics[width=\linewidth]{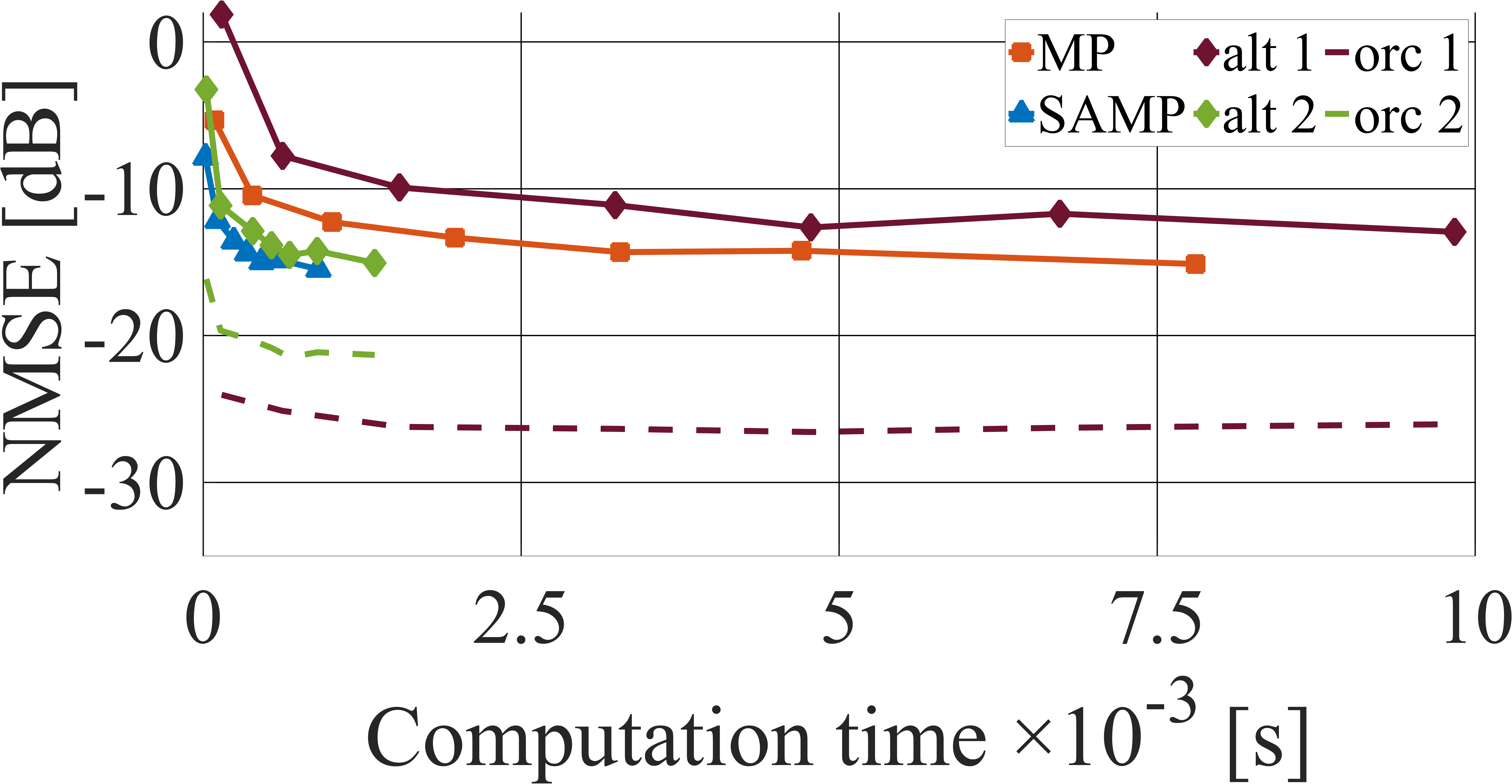}
        \caption{$M=4$, damped}
    \end{subfigure}
        \caption{{NMSE vs computational time comparison between the proposed, the alternatives, and the standard amplitude estimators. The oracle estimators for the alternative estimators are also included for reference. The frequency spacing is the Rayleigh limit, and the $\text{SNR}_{\text{dB}}$ is $10$ dB.}}
    \label{fig: amps est comparison with alternatives INTERVAL}
\end{figure}

\begin{figure}[t]
    \centering
    \begin{subfigure}[b]{0.49\columnwidth}
        \centering
        \includegraphics[width=\linewidth]{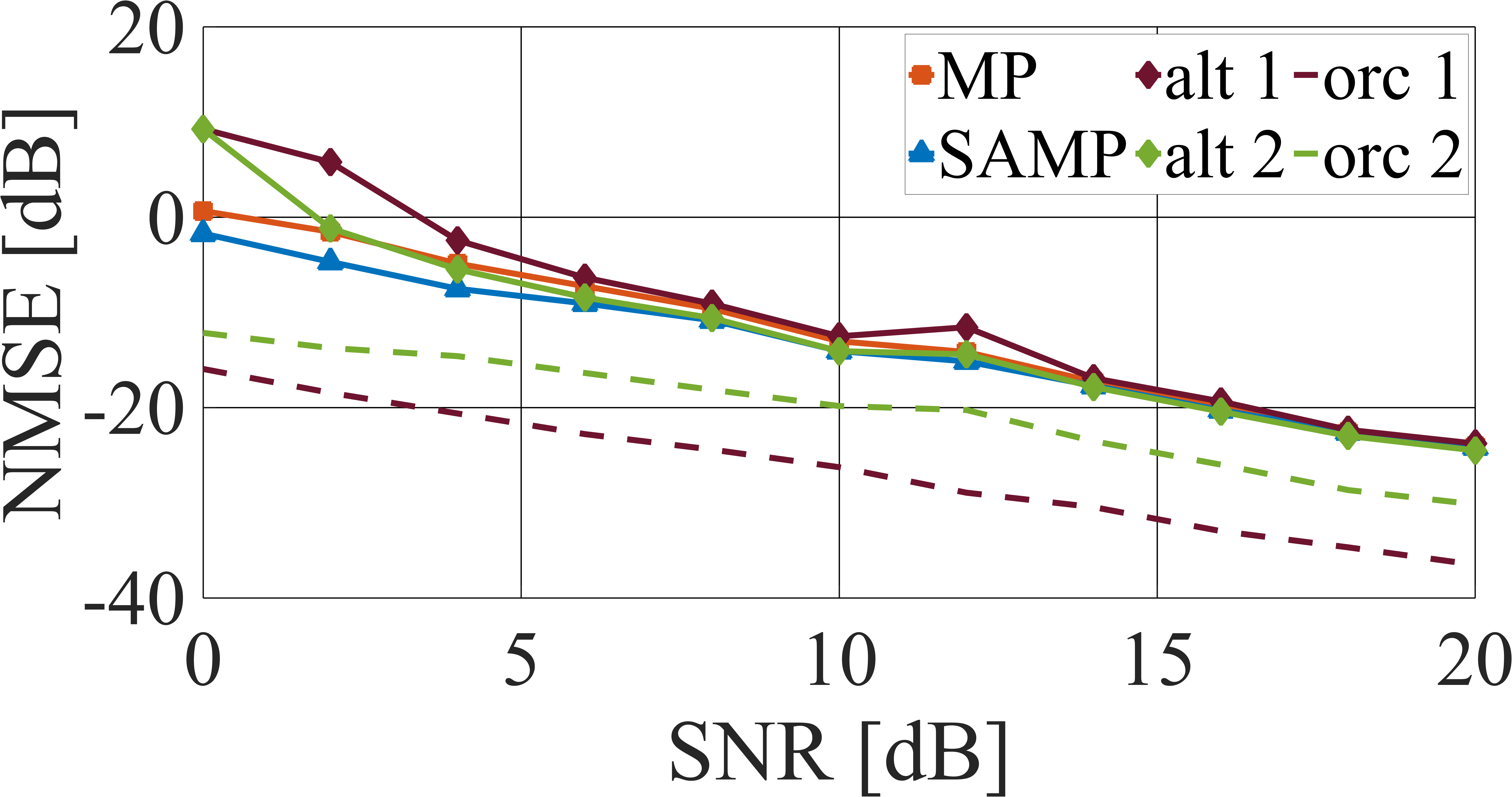}
        \caption{$M=2$, undamped}
        \label{fig: NMSE_amps_2_undmp_with_alt}
    \end{subfigure}
    \begin{subfigure}[b]{0.49\columnwidth}
        \centering
        \includegraphics[width=\linewidth]{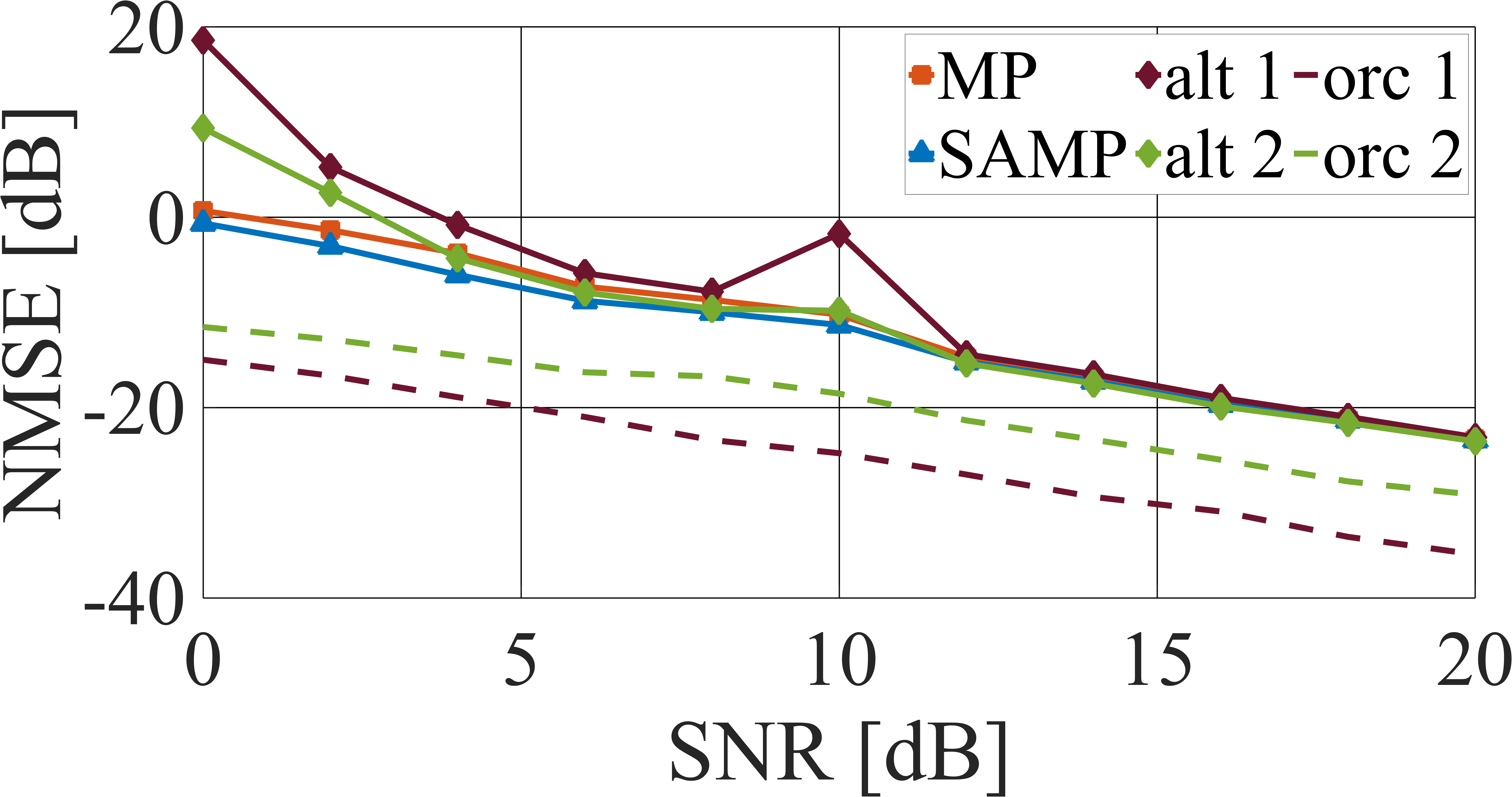}
        \caption{$M=2$, damped}
        \label{fig: NMSE_amps_2_dmp_with_alt}
    \end{subfigure}
\hfill
    \begin{subfigure}[b]{0.49\columnwidth}
        \centering
        \includegraphics[width=\linewidth]{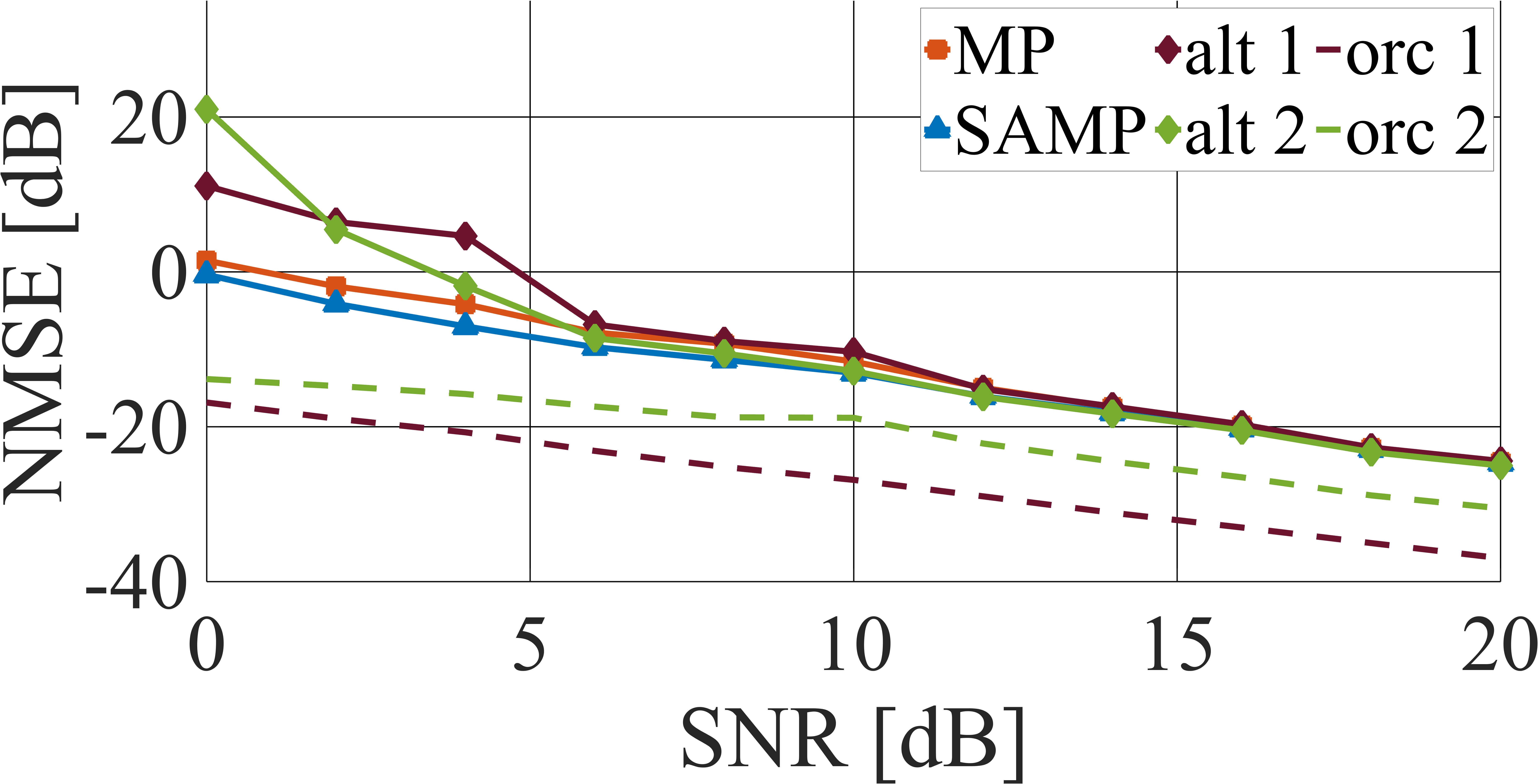}
        \caption{$M=4$, undamped}
        \label{fig: NMSE_amps_4_undmp_with_alt}
    \end{subfigure}
    \begin{subfigure}[b]{0.49\columnwidth}
        \centering
        \includegraphics[width=\linewidth]{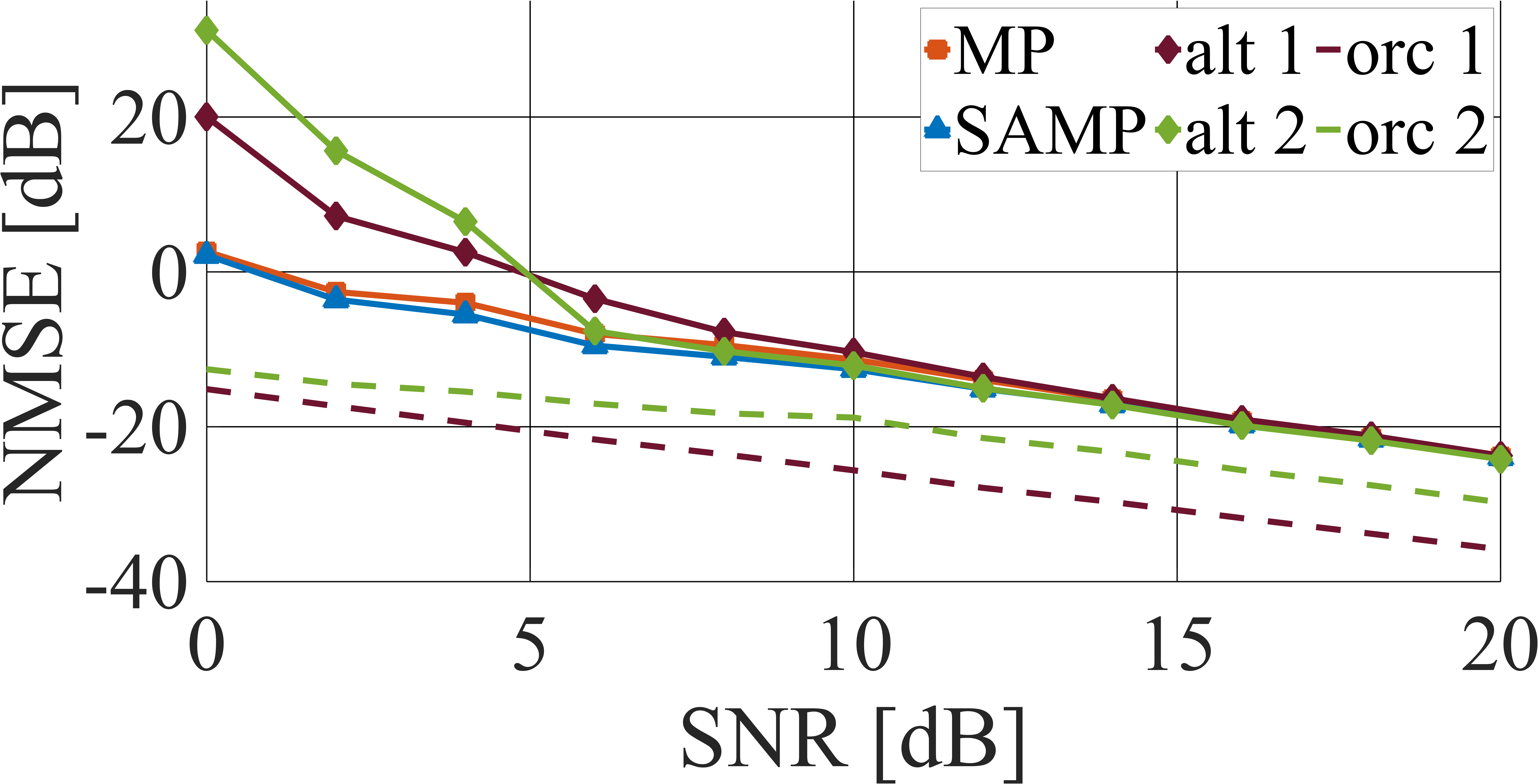}
        \caption{$M=4$, damped}
        \label{fig: NMSE_amps_4_dmp_with_alt}
    \end{subfigure}
        \caption{{NMSE vs SNR comparison between the proposed, the alternatives, and the standard amplitude estimators. The oracle estimators for the alternatives are also included for reference. The frequency spacing is the Rayleigh limit, and $N=71$.}}\vspace{-6mm}
    \label{fig: amps est comparison with alternatives SNR}
\end{figure}

{
\section{Sensitivity analysis w.r.t the constant $c_i$} \label{app: c_i sensitivity analysis}
In this appendix, we provide a sensitivity analysis of the proposed SAMP algorithm with respect to the constant $c_i$ in \eqref{eq: feature threshold}.
We conduct the same experiments as in Section~\ref{sec:Results}, and report the results of the undamped case, as the damped case shows similar trends.

Fig.~\ref{fig:Pd and RMSE c_i sensitivity analysis} shows $p_d$ and the average RMSE of the frequency estimates $\{\hat{\theta}_i\}_{i=1}^M$ versus SNR in dB for $M=2$ and $M=4$ undamped exponentials. The results reveal a trade-off between the two settings: for $M=2$, the best performance is obtained with the largest tested $c_i$, whereas for $M=4$ the opposite trend is observed. Nevertheless, the minimum AUC achieved across the tested values is $0.75$ for $M=2$ and $0.56$ for $M=4$, both exceed the baseline performance reported in Section~\ref{sec:Results}. We therefore adopt $c_i=10$ as a reasonable choice.}

\begin{figure}[t]
    \centering
    \begin{subfigure}[b]{0.49\columnwidth}
        \centering
        \includegraphics[width=\linewidth]{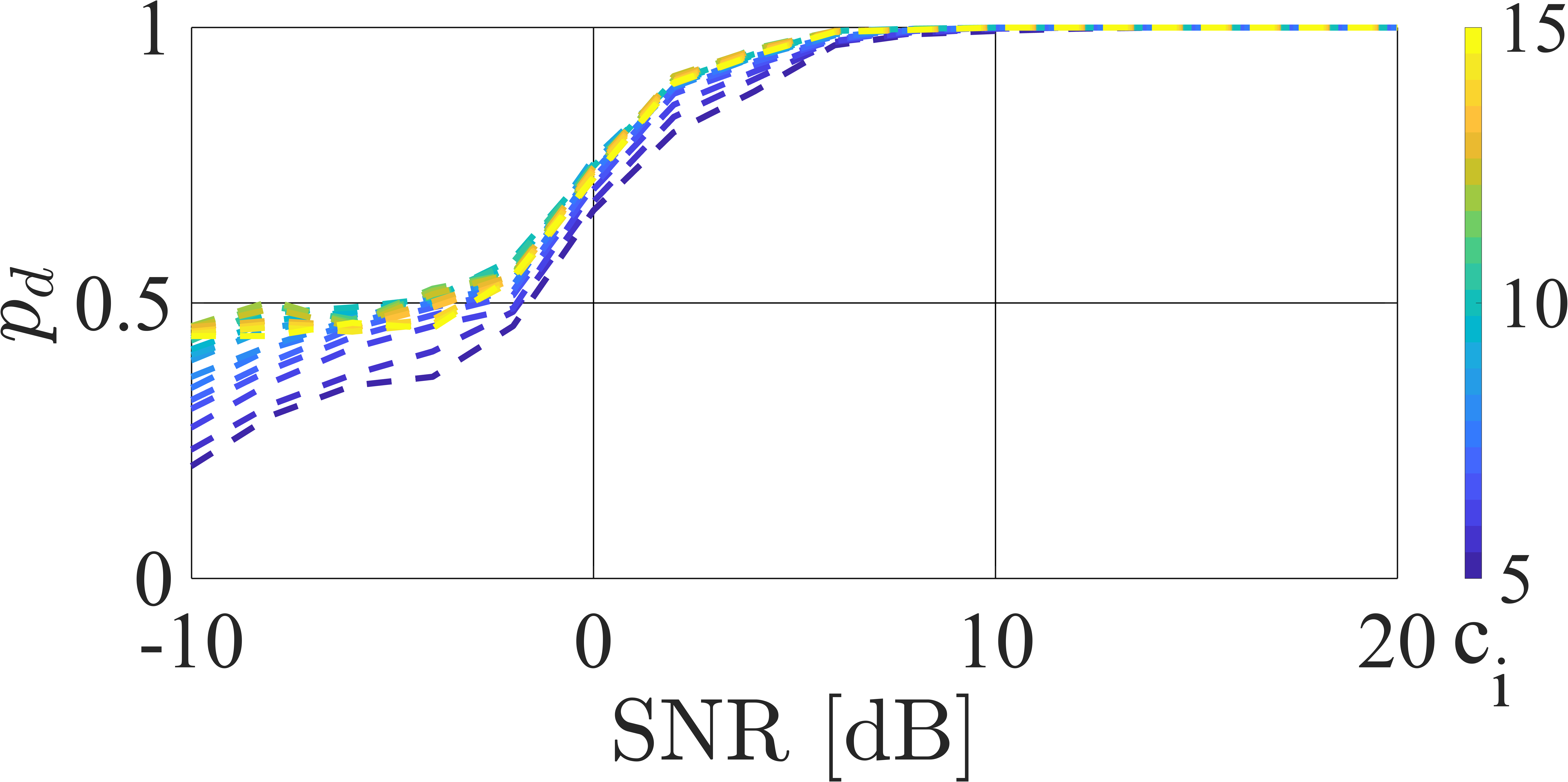}
        \caption{$M=2$, undamped}\label{fig: p_d vs SNR c_i 2 undamped}
    \end{subfigure}
    \begin{subfigure}[b]{0.49\columnwidth}
        \centering
        \includegraphics[width=\linewidth]{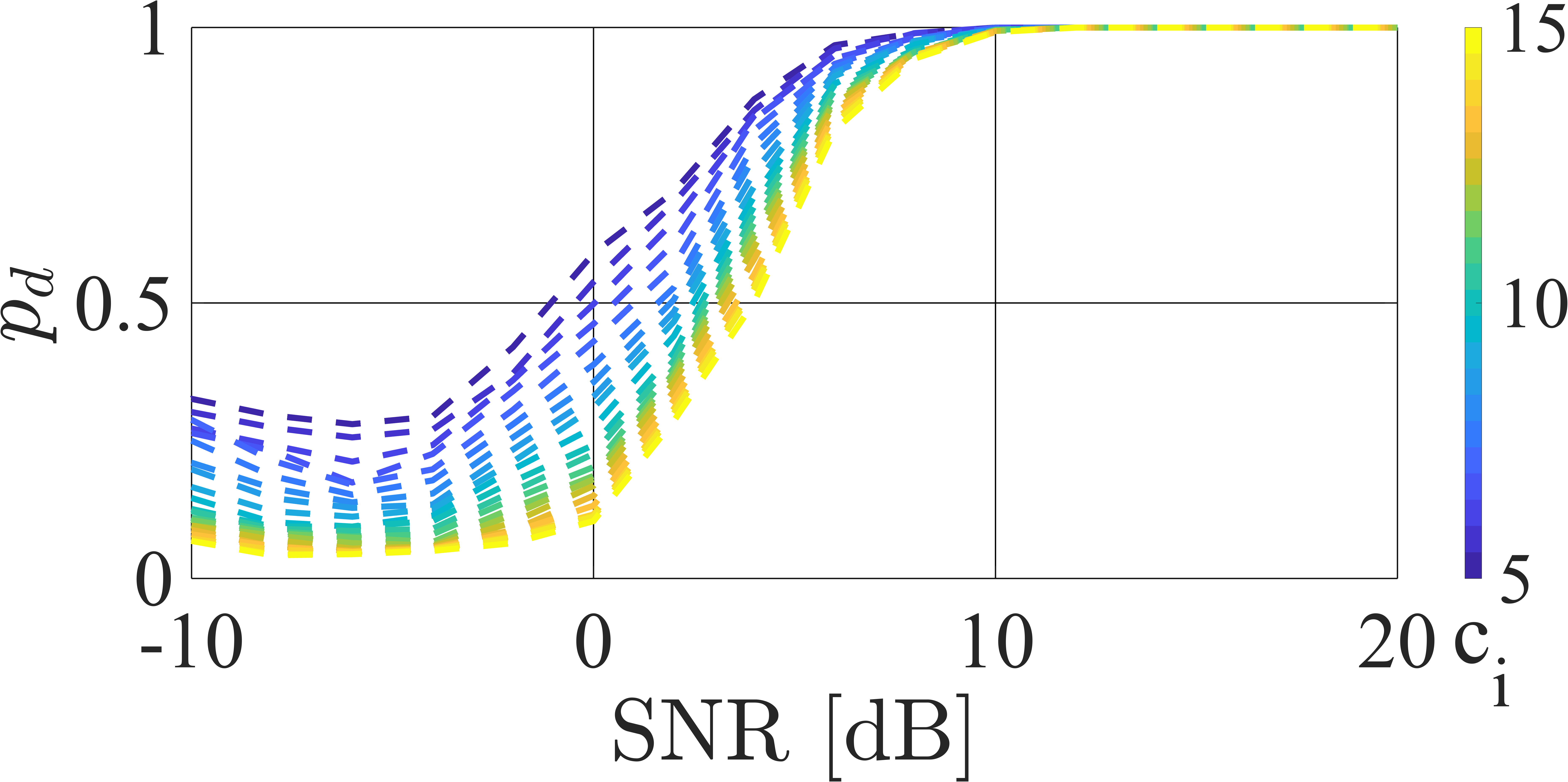}
        \caption{$M=4$, undamped}\label{fig: p_d vs SNR c_i 4 undamped}
    \end{subfigure}
     \begin{subfigure}[b]{0.49\columnwidth}
        \centering
        \includegraphics[width=\linewidth]{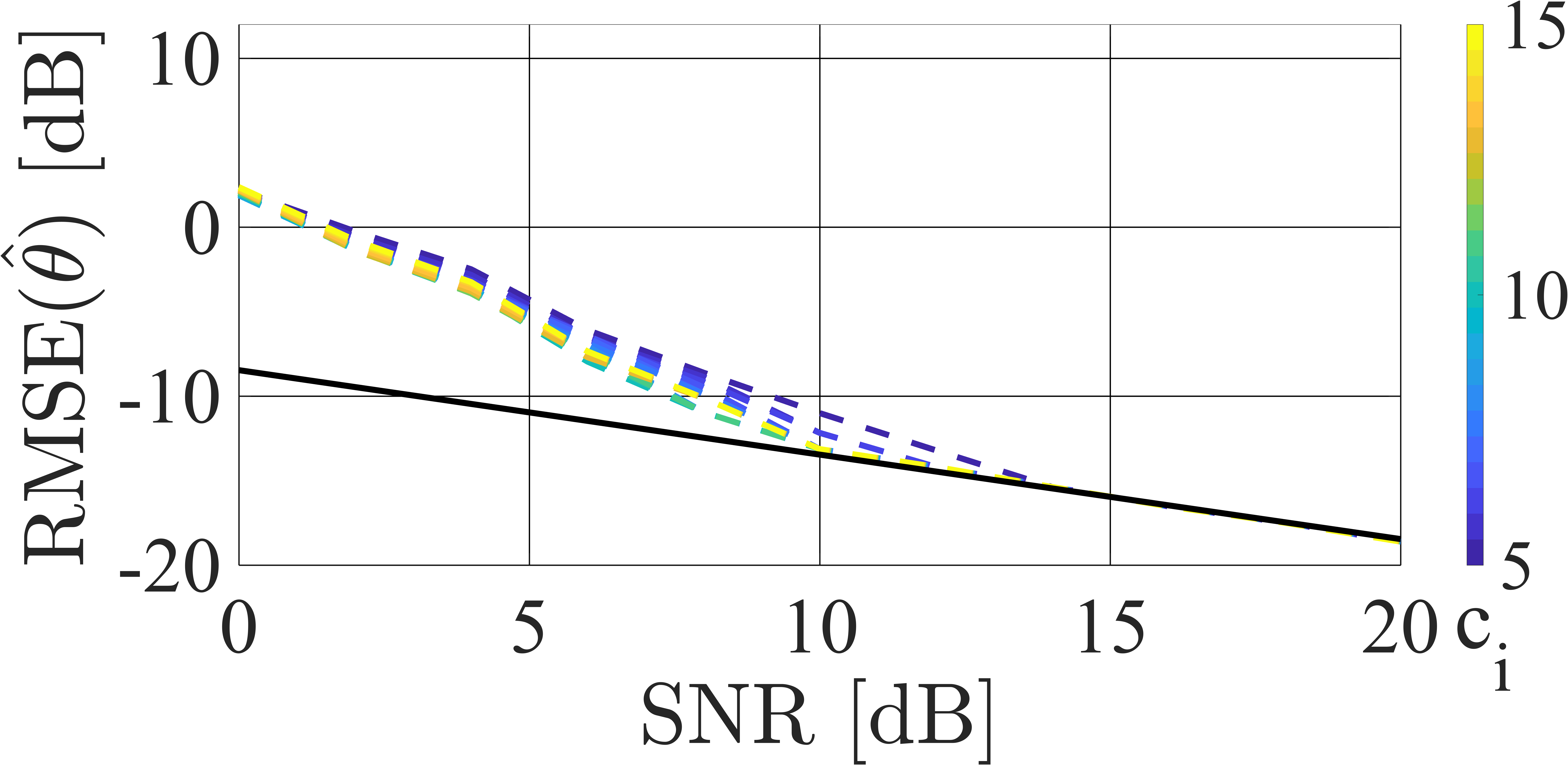}
        \caption{$M=2$, undamped}\label{fig: RMSE vs SNR c_i 2 undamped}
    \end{subfigure}
    \begin{subfigure}[b]{0.49\columnwidth}
        \centering
        \includegraphics[width=\linewidth]{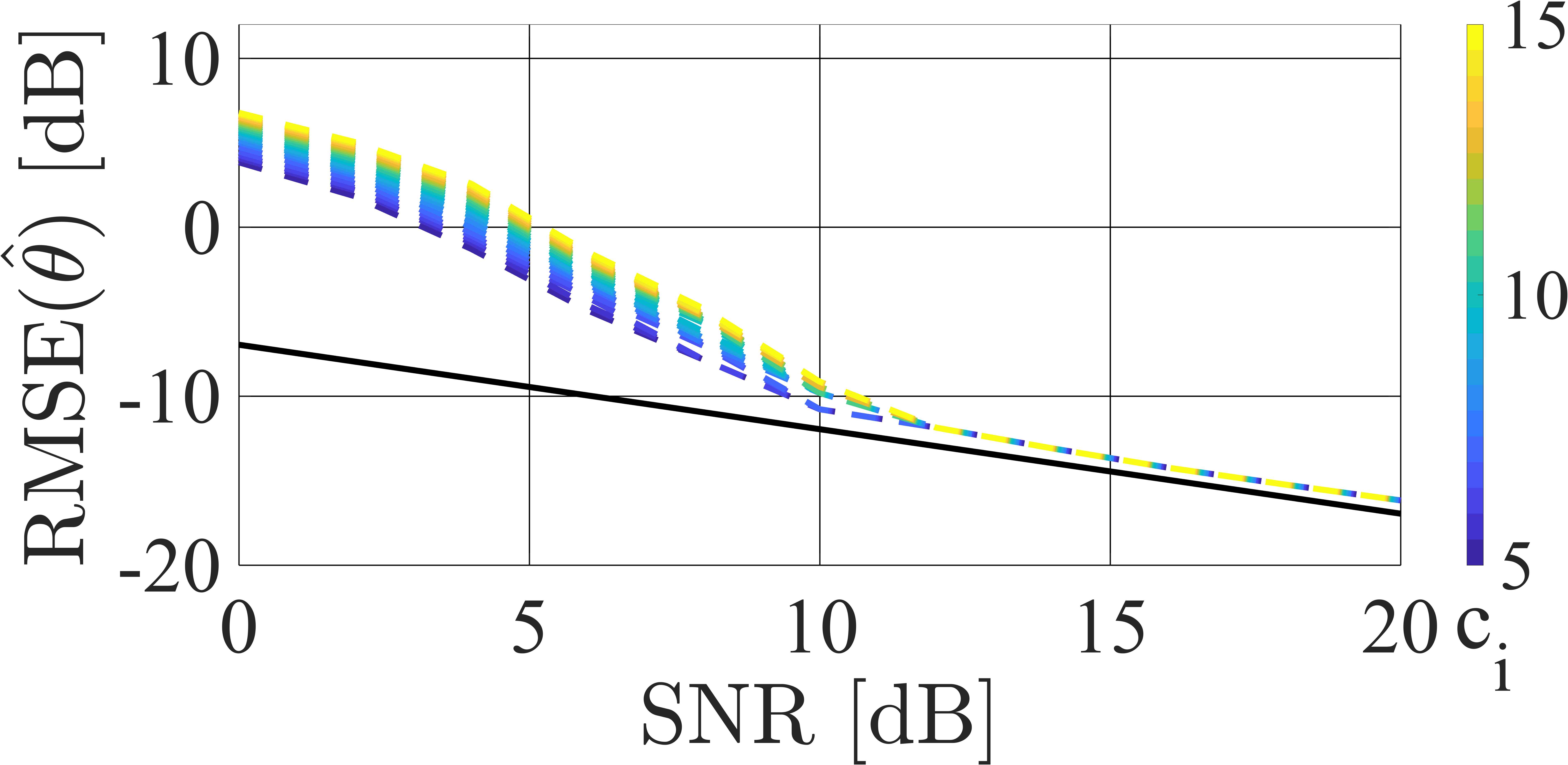}
        \caption{$M=4$, undamped}\label{fig: RMSE vs SNR c_i 4 undamped}
    \end{subfigure}
      \caption{{Probability of correct model order estimation versus $\text{SNR}_{\text{dB}}$ in (a)-(b), and average RMSE of $\{ \hat\theta_i \}_{i=1}^M$ versus the $\text{SNR}_{\text{dB}}$ in (c)-(d). Colors correspond to different values of $c_i$.}}
      \label{fig:Pd and RMSE c_i sensitivity analysis}
\end{figure}

\section{Additional Numerical Results}
\label{app: Additional Numerical Results}

In this appendix, we provide the complementary results for Section \ref{sec:Results}. 
Fig. \ref{fig: freqs_SNR_BIAS} is complementary for Fig. \ref{fig: freqs_SNR_RMSE}, and displays the average $\text{RMSE}$ of the estimates $\{\hat{\theta}_i\}_{i=1}^M$ in dB versus the SNR.
Fig. \ref{fig: freqs_SAMPLES_RMSE_BIAS} is the same as Fig. \ref{fig: freqs_SNR_RMSE}, but displays the average $\text{RMSE}$ of the estimates $\{\hat{\theta}_i\}_{i=1}^M$ and corresponding empirical bias in absolute value versus the number of samples $N$, for $M=2$, $|\Delta \theta|=2\pi/N$, and $\text{SNR}_{\text{dB}} = 8$. The complementary results for $M=4$ show high bias for the entire simulated range and are omitted.
Fig. \ref{fig: freqs_DIFF_RMSE_BIAS} is the same as Fig. \ref{fig: freqs_SNR_RMSE}, but displays the average $\text{RMSE}$ of the estimates $\{\hat{\theta}_i\}_{i=1}^M$ and corresponding empirical bias in absolute value versus the frequencies separation, for $M=2$, $N=71$, and $\text{SNR}_{\text{dB}} = 10$. The complementary results for $M=4$ shows high bias for the entire simulated range and are omitted.

\begin{figure}[t]
    \centering
    \begin{subfigure}[b]{0.49\columnwidth}
        \centering
        \includegraphics[width=\linewidth]{New_Figures/legend_fig_noCRB.png}
    \end{subfigure}
    \begin{subfigure}[b]{0.49\columnwidth}
        \centering
        \includegraphics[width=\linewidth]{New_Figures/legend_fig_noCRB.png}
    \end{subfigure}
    
    \begin{subfigure}[b]{0.49\columnwidth}
        \centering
        \includegraphics[width=\linewidth]{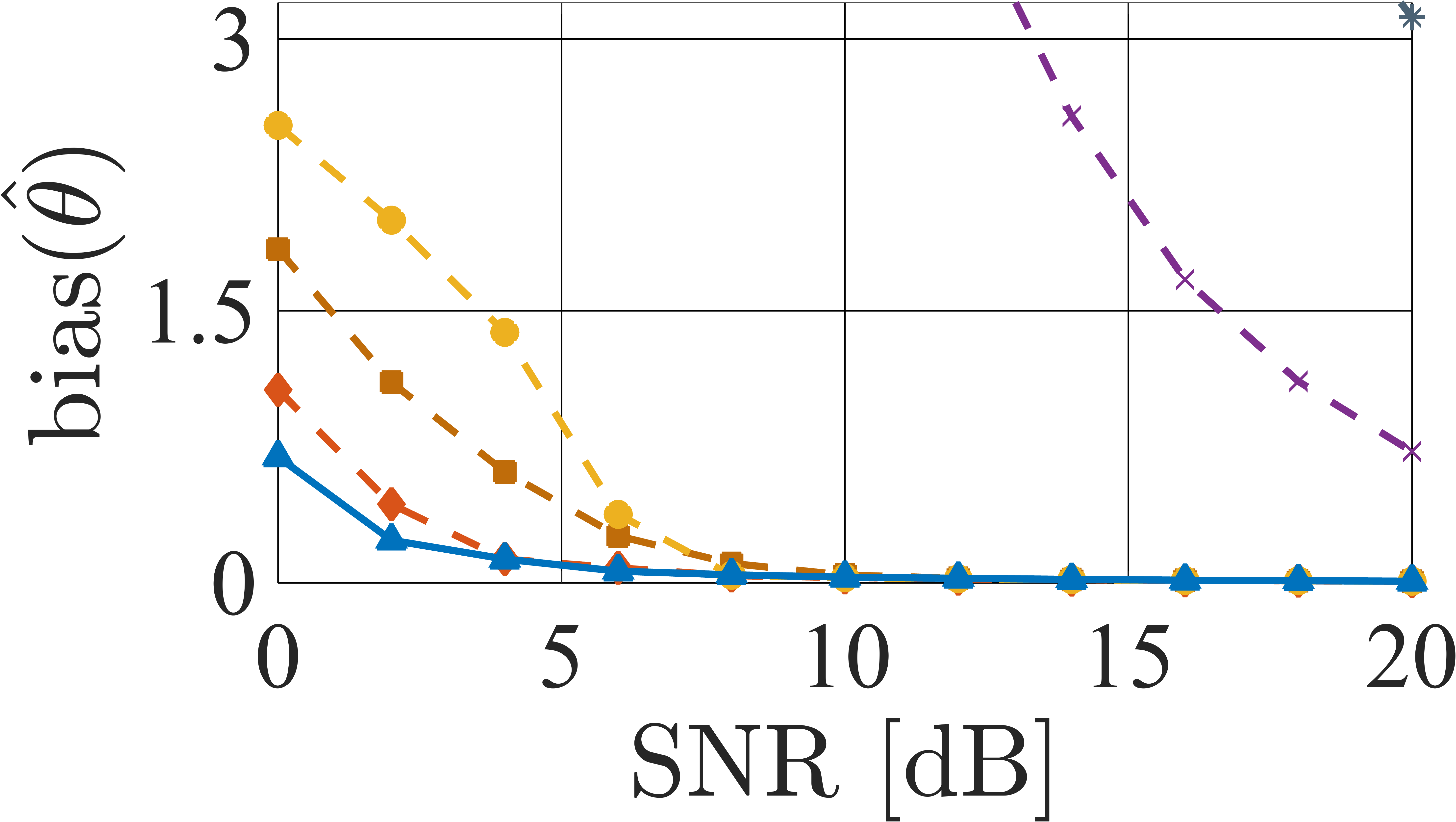}
        \caption{$M=2$, undamped}
        \label{fig:2_freqs_BIAS_undamped_SNR}
    \end{subfigure}
    \begin{subfigure}[b]{0.49\columnwidth}
        \centering
        \includegraphics[width=\linewidth] {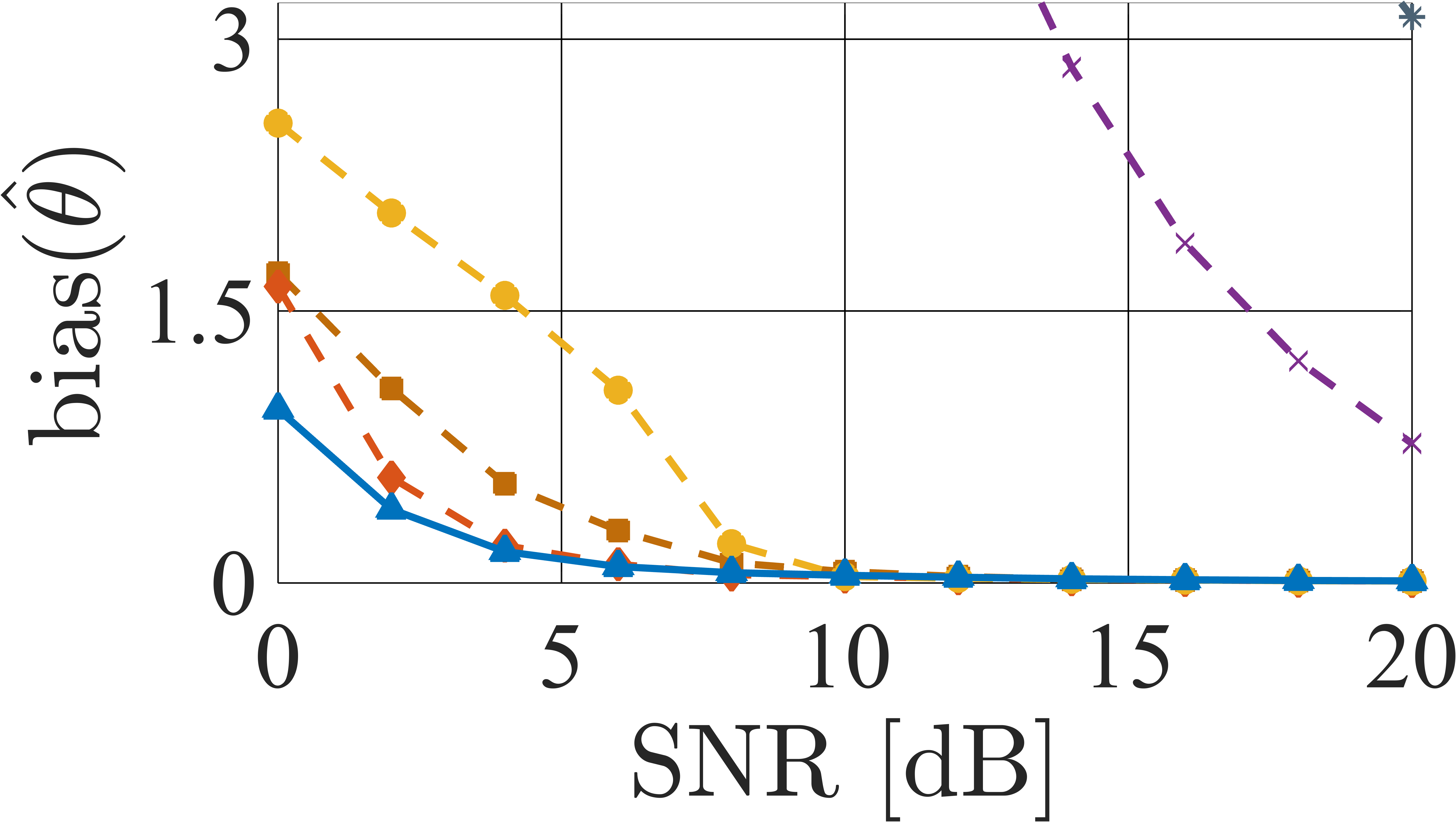}
        \caption{$M=2$, damped}
        \label{fig:2_freqs_BIAS_damped_SNR}
    \end{subfigure}

    \hfill
\begin{subfigure}[b]{0.49\columnwidth}
        \centering
        \includegraphics[width=\linewidth]{New_Figures/legend_fig_noCRB.png}
    \end{subfigure}
    \begin{subfigure}[b]{0.49\columnwidth}
        \centering
        \includegraphics[width=\linewidth]{New_Figures/legend_fig_noCRB.png}
    \end{subfigure}
    
    \begin{subfigure}[b]{0.49\columnwidth}
        \centering
        \includegraphics[width=\linewidth]{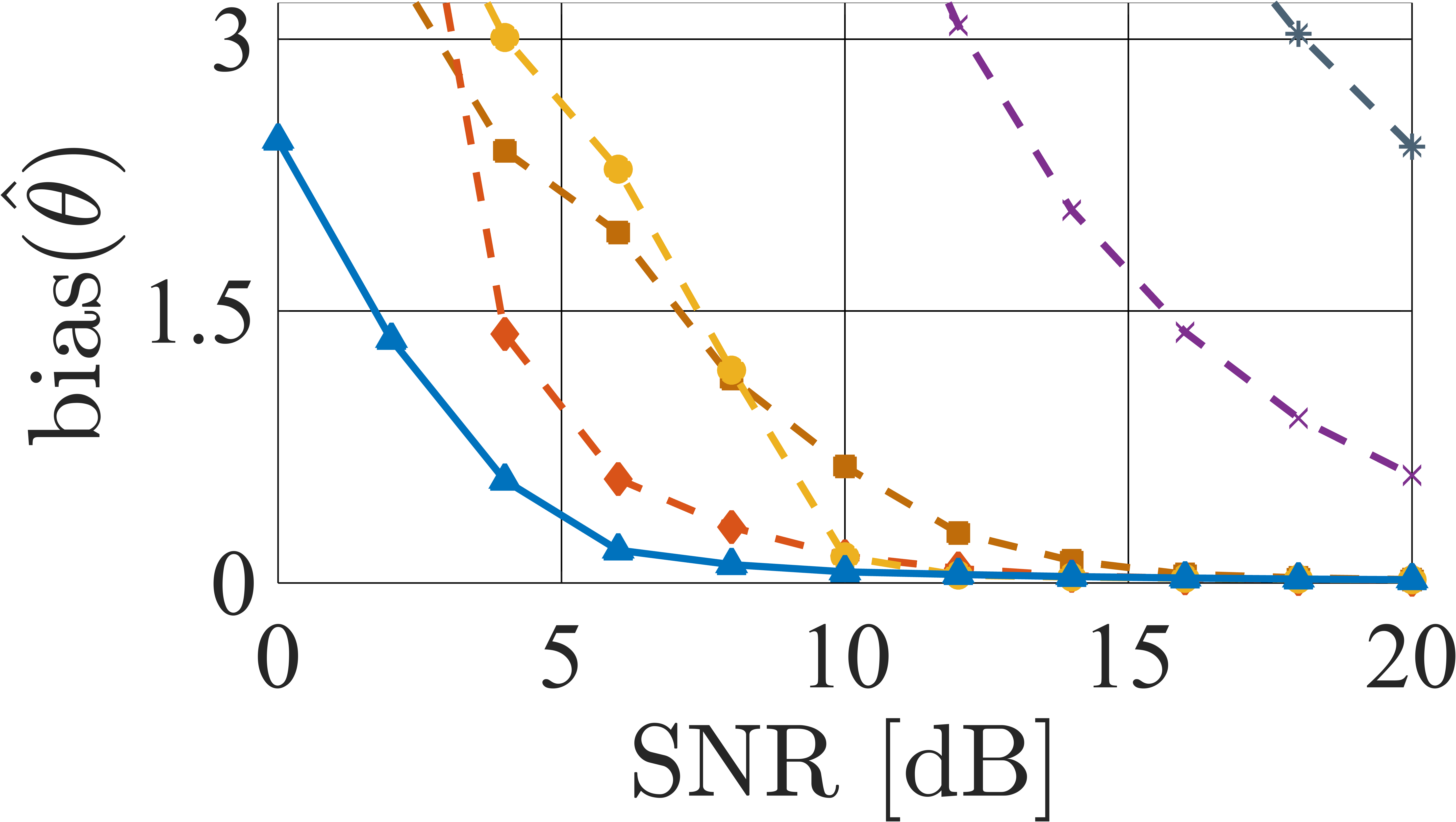}
        \caption{$M=4$, undamped}
        \label{fig:4_freqs_BIAS_undamped_SNR}
    \end{subfigure}
    \begin{subfigure}[b]{0.49\columnwidth}
        \centering
        \includegraphics[width=\linewidth] {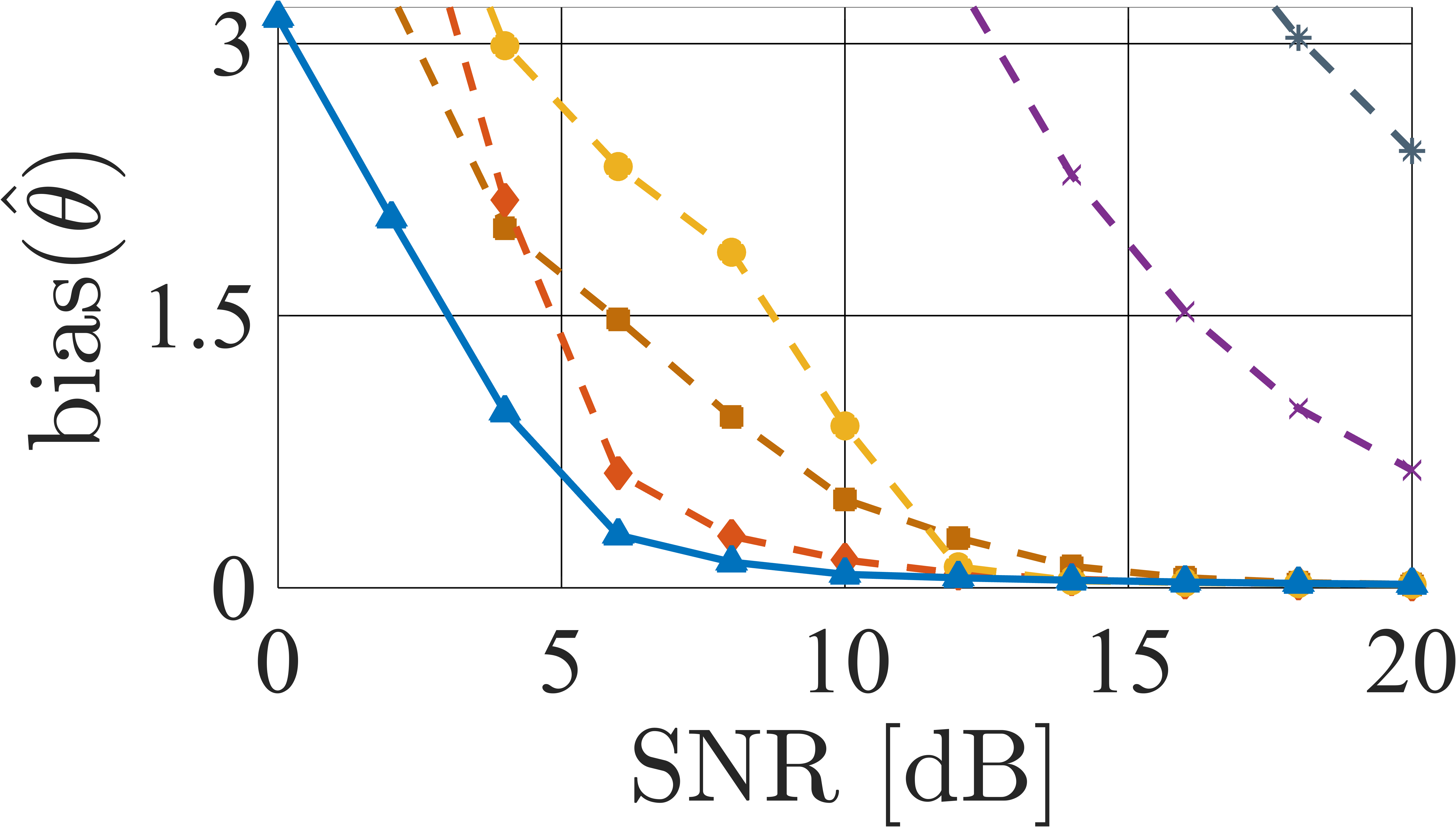}
        \caption{$M=4$, damped}
        \label{fig:4_freqs_BIAS_damped_SNR}
    \end{subfigure}
    \hfill
    \caption{Average bias in absolute value of $\{ \hat \theta_i \}_{i=1}^M$ versus the $\text{SNR}_{\text{dB}}$.}\vspace{-4mm}
    \label{fig: freqs_SNR_BIAS}
\end{figure}

\begin{figure}[t]
    \centering
    \begin{subfigure}[b]{0.49\columnwidth}
        \centering
        \includegraphics[width=\linewidth]{New_Figures/legend_fig.png}
    \end{subfigure}
    \begin{subfigure}[b]{0.49\columnwidth}
        \centering
        \includegraphics[width=\linewidth]{New_Figures/legend_fig.png}
    \end{subfigure}
    
    \begin{subfigure}[b]{0.49\columnwidth}
        \centering
        \includegraphics[width=\linewidth]{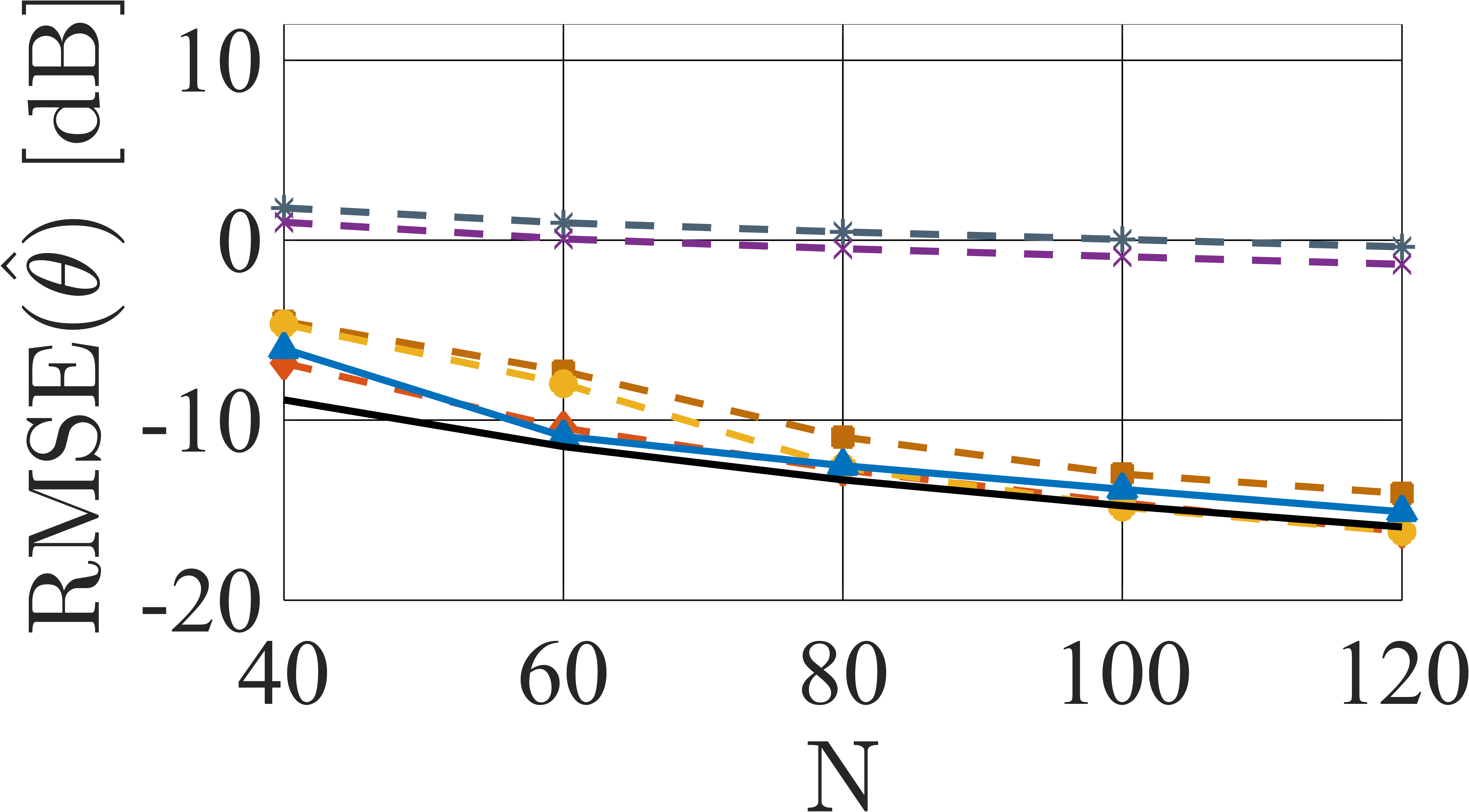}
        \caption{$M=2$, undamped}
        \label{fig:2_freqs_RMSE_undamped_SAMPLES}
    \end{subfigure}
    \begin{subfigure}[b]{0.49\columnwidth}
        \centering
        \includegraphics[width=\linewidth] {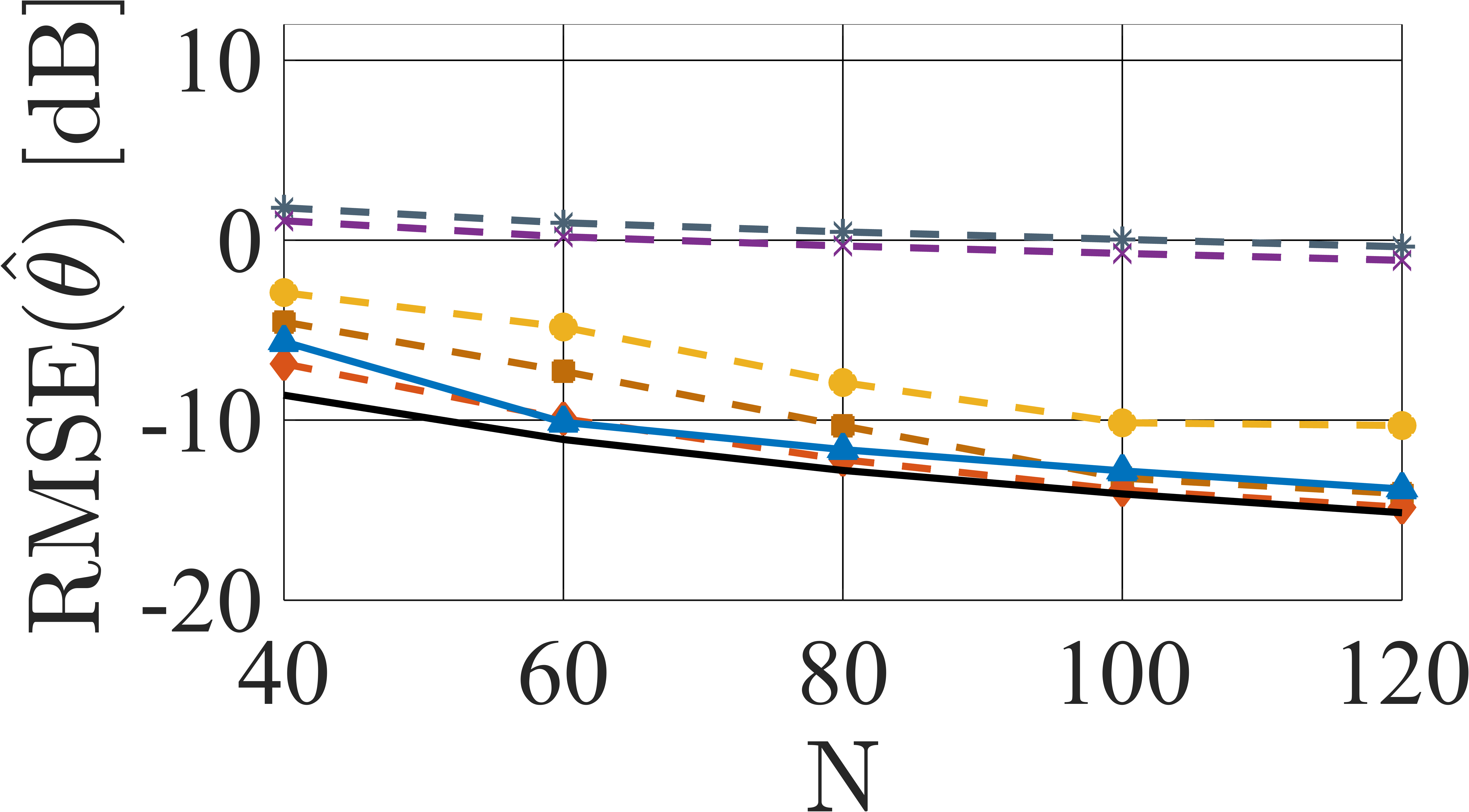}
        \caption{$M=2$, damped}
        \label{fig:2_freqs_RMSE_damped_SAMPLES}
    \end{subfigure}

    \hfill
\begin{subfigure}[b]{0.49\columnwidth}
        \centering
        \includegraphics[width=\linewidth]{New_Figures/legend_fig_noCRB.png}
    \end{subfigure}
    \begin{subfigure}[b]{0.49\columnwidth}
        \centering
        \includegraphics[width=\linewidth]{New_Figures/legend_fig_noCRB.png}
    \end{subfigure}
    
    \begin{subfigure}[b]{0.49\columnwidth}
        \centering
        \includegraphics[width=\linewidth]{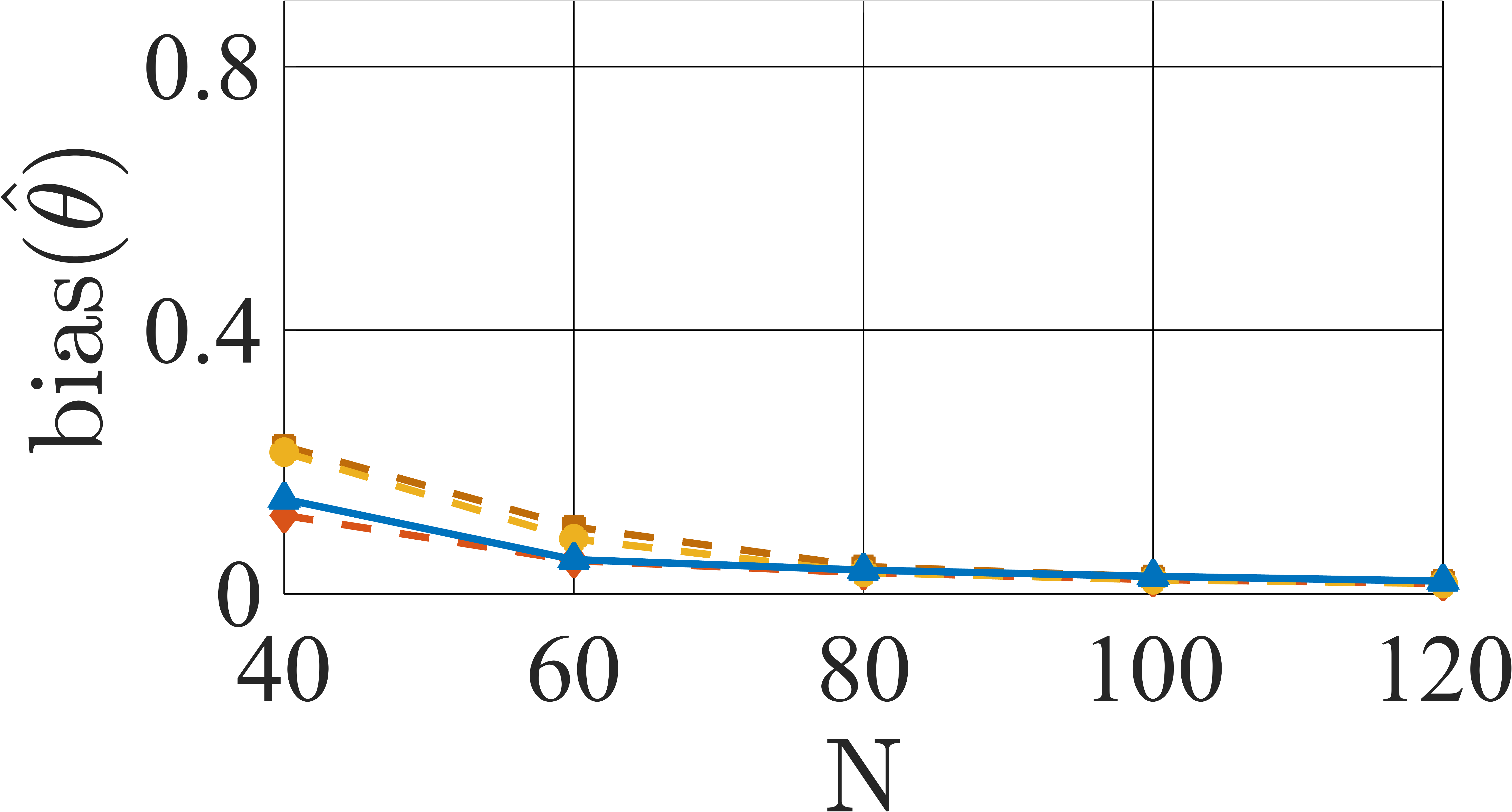}
        \caption{$M=2$, undamped}
        \label{fig:2_freqs_BIAS_undamped_SAMPLES}
    \end{subfigure}
    \begin{subfigure}[b]{0.49\columnwidth}
        \centering
        \includegraphics[width=\linewidth] {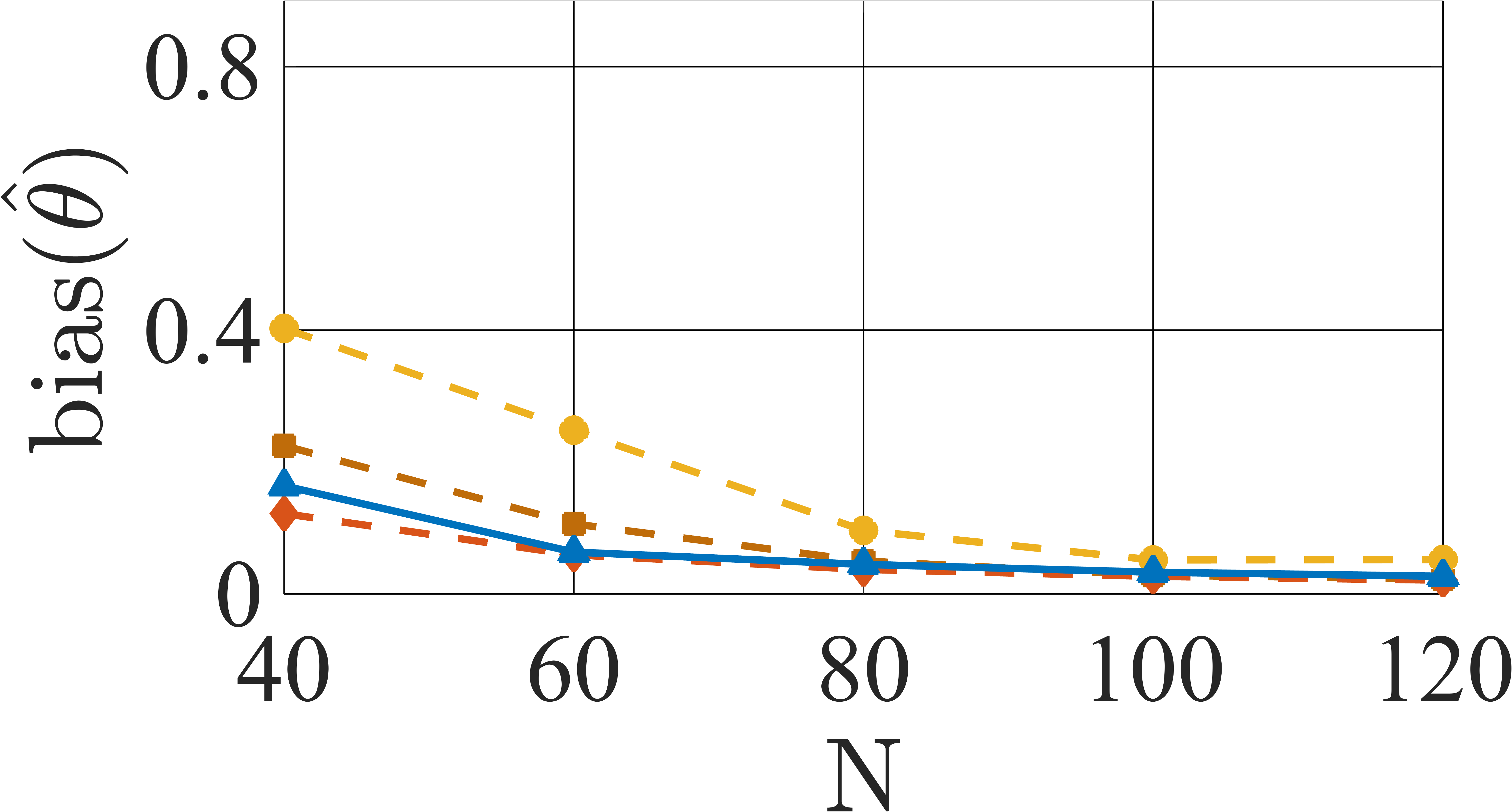}
        \caption{$M=2$, damped}
        \label{fig:2_freqs_BIAS_damped_SAMPLES}
    \end{subfigure}
    \hfill
    \caption{Average RMSE and bias in absolute value of $\{\hat \theta_i \}_{i=1}^M$ versus the number of samples.}
    \label{fig: freqs_SAMPLES_RMSE_BIAS}
\end{figure}

\begin{figure}[t]
    \centering
    \begin{subfigure}[b]{0.49\columnwidth}
        \centering
        \includegraphics[width=\linewidth]{New_Figures/legend_fig.png}
    \end{subfigure}
    \begin{subfigure}[b]{0.49\columnwidth}
        \centering
        \includegraphics[width=\linewidth]{New_Figures/legend_fig.png}
    \end{subfigure}
    
    \begin{subfigure}[b]{0.49\columnwidth}
        \centering
        \includegraphics[width=\linewidth]{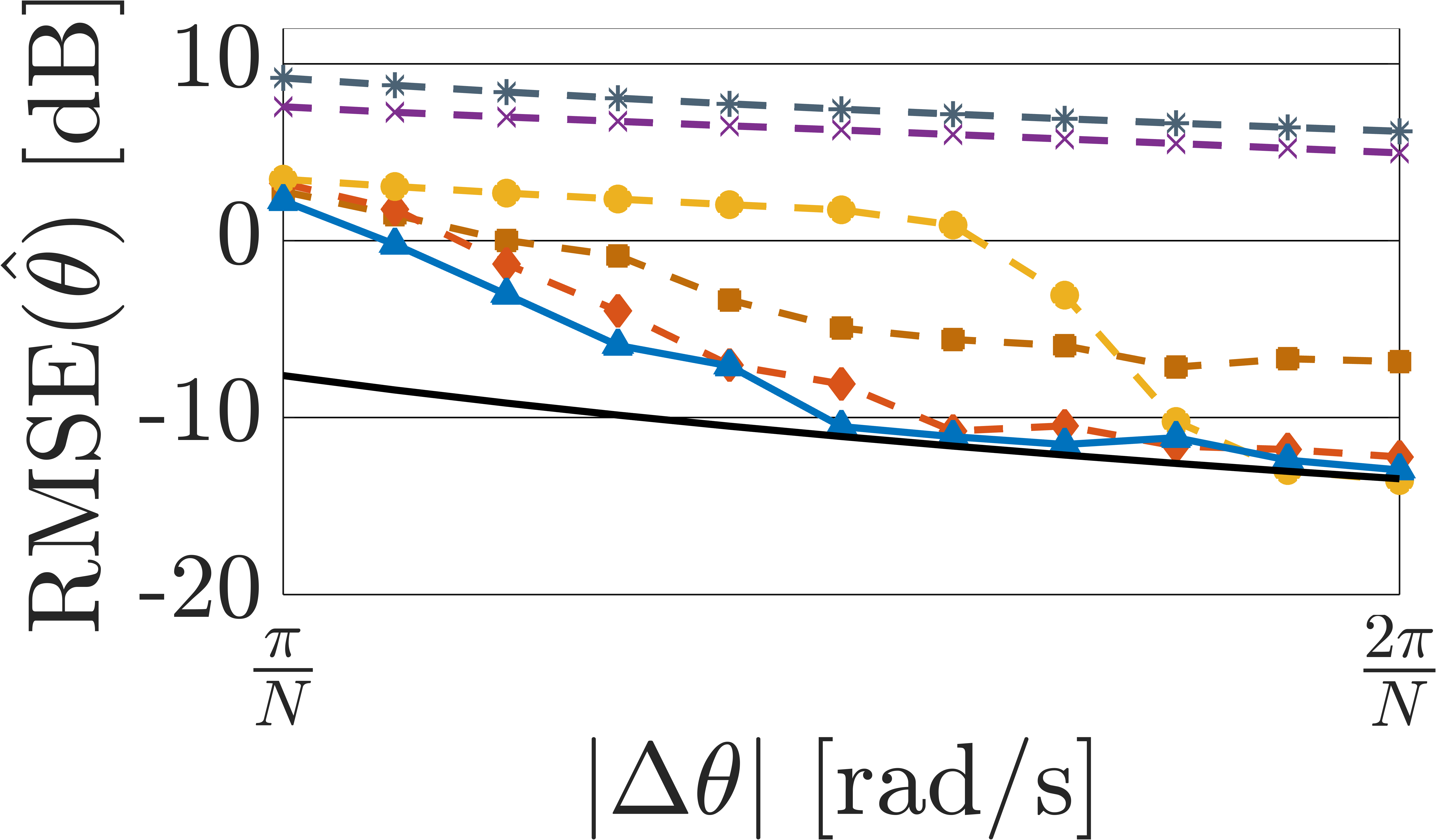}
        \caption{$M=2$, undamped}
        \label{fig:2_freqs_RMSE_undamped_DIFF}
    \end{subfigure}
    \begin{subfigure}[b]{0.49\columnwidth}
        \centering
        \includegraphics[width=\linewidth] {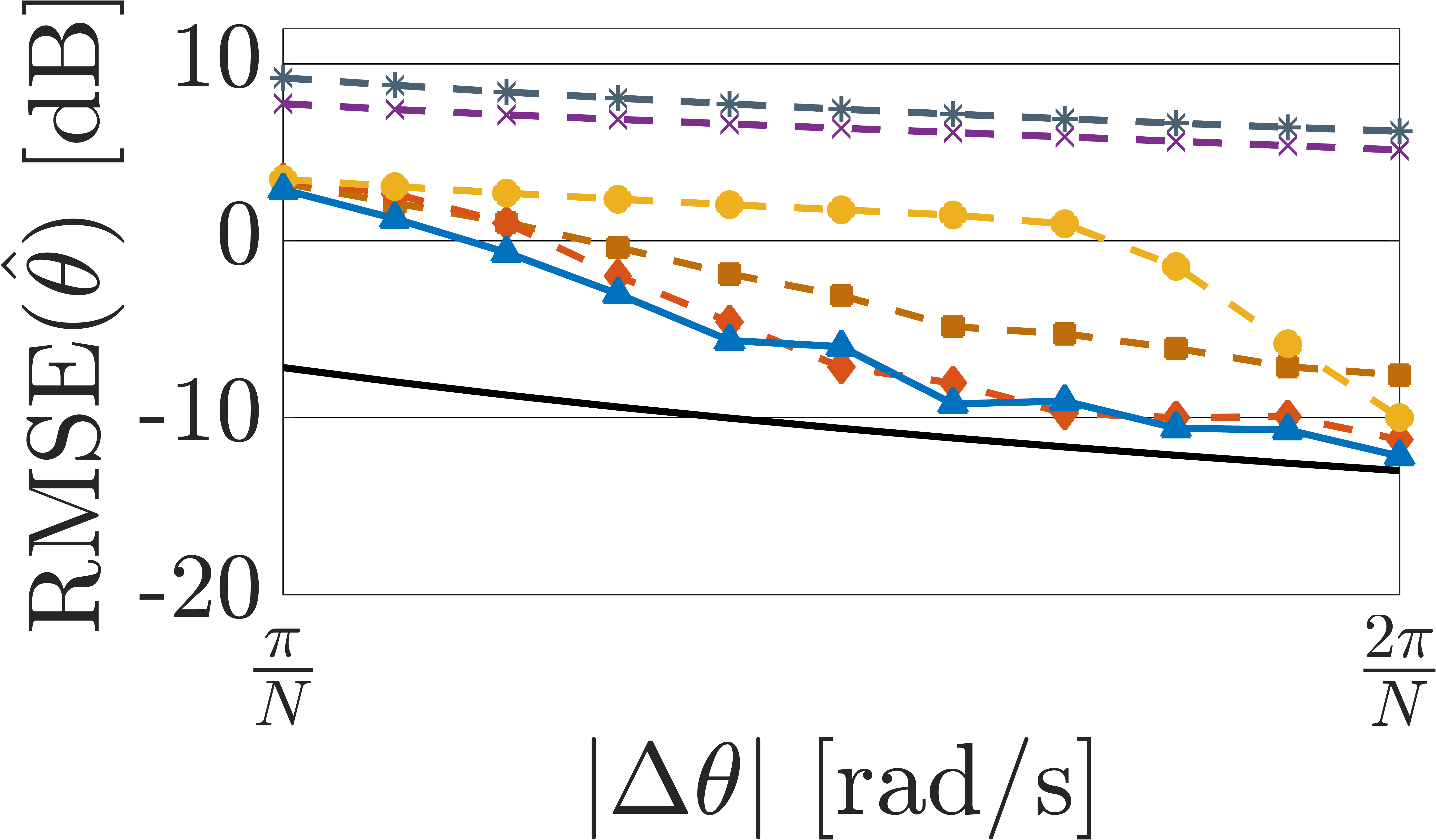}
        \caption{$M=2$, damped}
        \label{fig:2_freqs_RMSE_damped_DIFF}
    \end{subfigure}

    \hfill
\begin{subfigure}[b]{0.49\columnwidth}
        \centering
        \includegraphics[width=\linewidth]{New_Figures/legend_fig_noCRB.png}
    \end{subfigure}
    \begin{subfigure}[b]{0.49\columnwidth}
        \centering
        \includegraphics[width=\linewidth]{New_Figures/legend_fig_noCRB.png}
    \end{subfigure}
    
    \begin{subfigure}[b]{0.49\columnwidth}
        \centering
        \includegraphics[width=\linewidth]{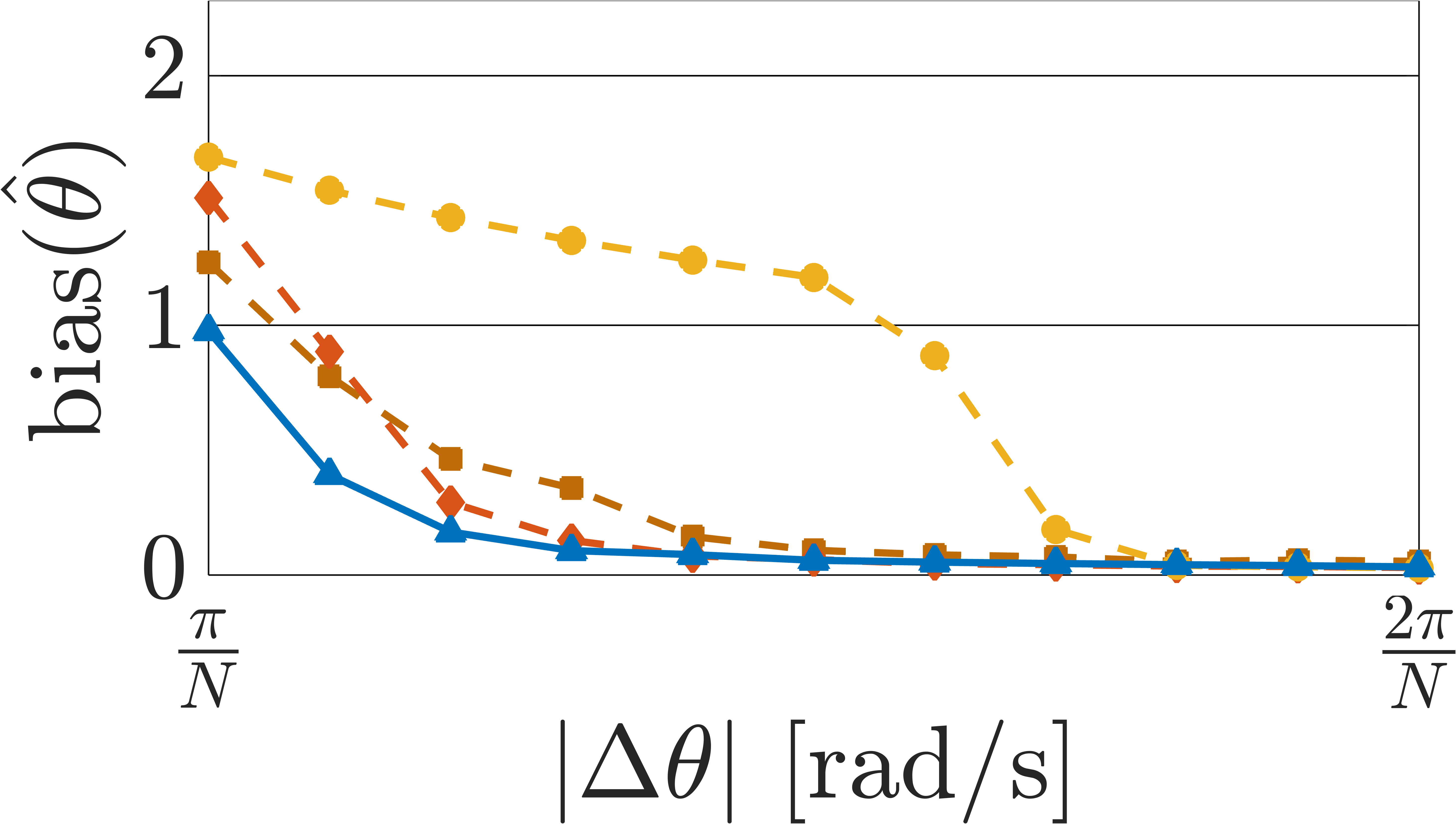}
        \caption{$M=2$, undamped}
        \label{fig:4_freqs_RMSE_undamped_DIFF}
    \end{subfigure}
    \begin{subfigure}[b]{0.49\columnwidth}
        \centering
        \includegraphics[width=\linewidth] {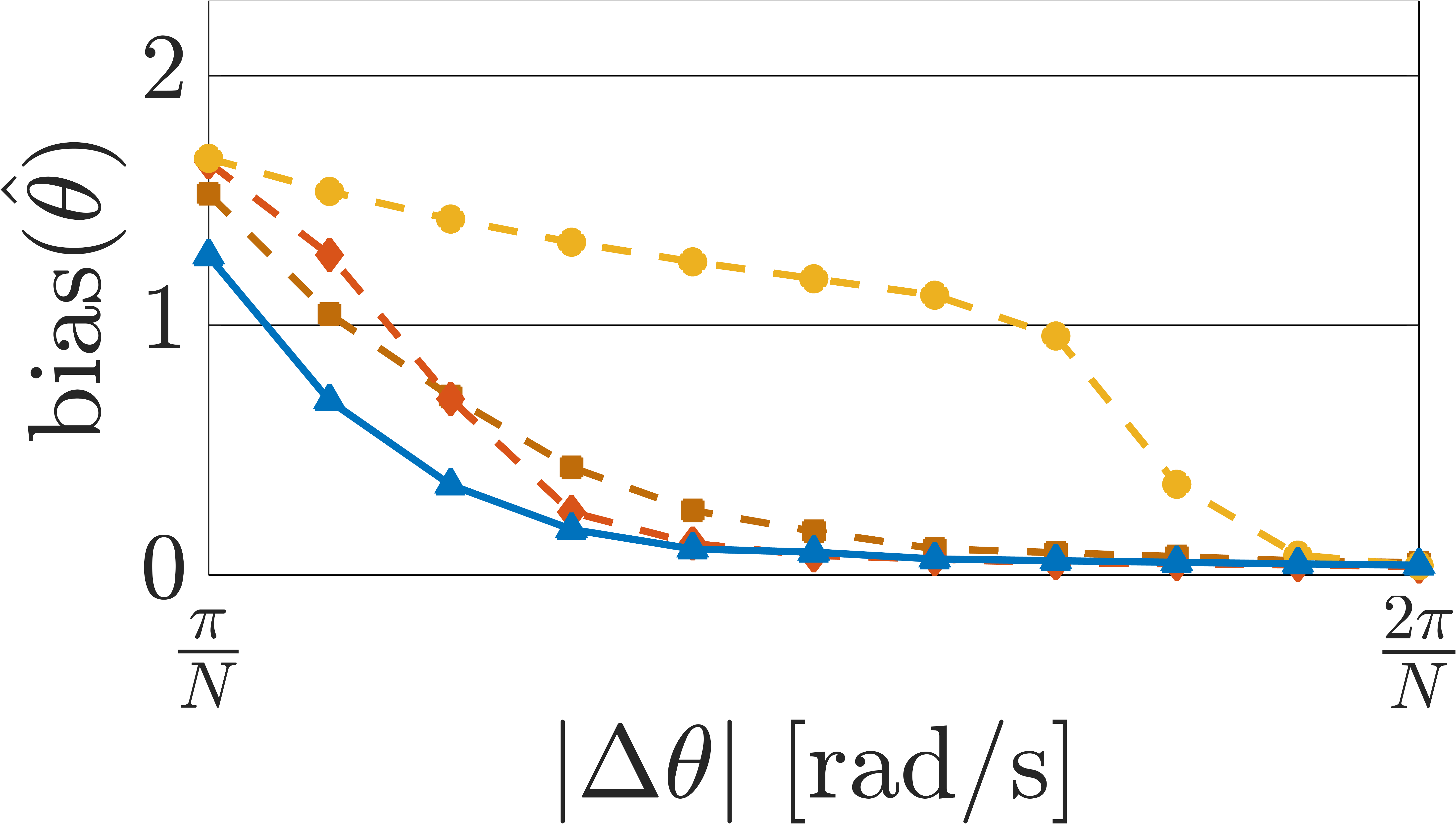}
        \caption{$M=2$, damped}
        \label{fig:4_freqs_RMSE_damped_DIFF}
    \end{subfigure}
    \hfill
    \caption{Average RMSE and bias in absolute value of $\{ \hat\theta_i \}_{i=1}^M$ versus the frequency separation.}\vspace{-4mm}
    \label{fig: freqs_DIFF_RMSE_BIAS}
\end{figure}

\printbibliography
\end{refsection}
\end{NoHyper}
}

\end{document}